\documentclass[18pt]{article}
\usepackage[left=2cm,right=2cm,top=2cm,bottom=2cm]{geometry} 

\usepackage{amsthm,amssymb,amsbsy,amsmath,amsfonts,amssymb,amscd}
\usepackage{amsmath}
\usepackage{amssymb}
\usepackage[toc,page]{appendix}
\usepackage{caption}
\usepackage{subcaption}
\usepackage{float}
\usepackage{graphicx}

\usepackage[backend=bibtex,style=numeric,maxbibnames=10,giveninits=true]{biblatex}
\renewbibmacro{in:}{}
\usepackage[colorlinks=true,
            linkcolor=blue,    
            citecolor=red,     
            urlcolor=magenta]  
            {hyperref}
\usepackage{bm}
\usepackage{bigints}

\newtheorem{theorem}{Theorem}[section]
\newtheorem{lemma}[theorem]{\textbf{Lemma}}
\newtheorem{prop}[theorem]{\textbf{Proposition}}

\theoremstyle{remark}
\newtheorem{remark}[theorem]{\textbf{Remark}}

\date{}

\author{Sirine Boucenna
  \and
  Vasilis Dakos
  \and
  Ga\"el Raoul
}

\newcommand{\Addresses}{{
  \bigskip
  \footnotesize

  Sirine Boucenna, \textsc{Institut des Sciences de l'Evolution de Montpellier (ISEM), Université de Montpellier, CNRS, IRD, EPHE, Montpellier, France.}\par\nopagebreak
  \textit{E-mail address}, S. Boucenna: \texttt{sirine.boucenna@umontpellier.fr}

  \medskip

  Vasilis Dakos, \textsc{Institut des Sciences de l'Evolution de Montpellier (ISEM), Université de Montpellier, CNRS, IRD, EPHE, Montpellier, France.}\par\nopagebreak
  \textit{E-mail address}, V. Dakos: \texttt{vasilis.dakos@umontpellier.fr}

  \medskip

  Ga\"el Raoul, \textsc{CMAP, CNRS, Institut polytechnique de Paris, Inria, route de Saclay, 91128 Palaiseau, France}\par\nopagebreak
  \textit{E-mail address}, G. Raoul: \texttt{gael.raoul@polytechnique.edu}

}}

\title{A model for a population of trees structured by phenological traits}

\begin{document}

\maketitle

\begin{abstract}
In the context of global warming, tree populations rely on two primary mechanisms of adaptation: phenotypic plasticity, which enables individuals to adjust their behavior in response to environmental stress, and genetic evolution, driven by natural selection and genetic diversity within the population. Understanding the interplay between these mechanisms is crucial for assessing the impacts of climate change on forest ecosystems and for informing sustainable management strategies. In this manuscript, we focus on a specific phenological adaptation: the ability of trees to enter summer dormancy once a critical temperature threshold is exceeded. Individuals are characterized by this threshold temperature and by their seed production capacity. We first establish a detailed mathematical model describing the population dynamics under these traits, and progressively reduce it to a system of two coupled ordinary differential equations. This simpler macroscopic model is then analyzed numerically, to investigate how the population reacts to a shift in its environment: an temperature increase, a drop in precipitation levels, or a combination of the two. Our results highlight contrasting effects of water stress and temperature stress on population dynamics, as well as the ambivalent effect of the plasticity.
\end{abstract}


\medskip


\bigskip

\section{Introduction}

As global temperatures rise and precipitation patterns change, the phenological responses of plants can have significant ecological consequences. This is especially true for trees, where earlier leaf budburst or delayed dormancy can alter the competitive dynamics between species, affect the availability of resources, and disrupt the synchrony with pollinators and herbivores. This, in turn, impacts forest composition, carbon sequestration, and overall ecosystem stability (\cite{cleland2007shifting}, \cite{parmesan2003globally}). These phenological responses are not only influenced by inherited genetic factors, which can be represented by breeding values, but also exhibit significant phenotypic plasticity. Evolution of individuals breeding values as well as phenotypic plasticity can allow trees  to adjust their developmental processes (eg timing of their phenology) in response to varying environmental conditions. Thus, understanding the role of evolution and plasticity of phenological-related responses of trees under changing climatic conditions is essential to predict their adaptive capacities for informing conservation strategies to maintain biodiversity and ecosystem services (\cite{menzel2006european,matesanz2010global}). \\

In this manuscript, we are interested in phenotypic adaptation to summer stress related to warm temperatures and lack of water availability. Warm summer temperature associated with water stress may lead to summer dormancy as an adaptive strategy (\cite{gillespie2017winter}). Typically, the term dormancy refers to winter dormancy, which is the period of inactivity observed in plants and some animals during winter months due to colder temperatures and reduced daylight (\cite{maurya2017photoperiod}). More broadly, seasonal dormancy describes an adaptive strategy developed by certain species to endure unfavorable conditions regardless of the season (\cite{vegis1964dormancy}), and summer dormancy refers to a strategy developed by plants to cope with hot, dry summers, by slowing down their activity. This plastic adaptation can be triggered by heat but also drought. Mostly described on perennial plants so far, summer dormancy is associated with greater survival rates after severe and repeated summer droughts in many perennial grasses (\cite{volaire2006summer}). For instance, it was demonstrated that \emph{Poaceae} enter summer dormancy in the arid summer climate of southern California, even when they were supplied with water throughout the dry season (\cite{laude1953nature}). Plastic adaptations to heat and/or drought have also been described in tree populations (\cite{valladares2000plastic,ravn2022phenotypic,aspelmeier2006genotypic,mediavilla2019foliar}). Since extreme conditions like heat and drought are becoming more frequent with climate change, it is important to study plastic adaptation to such climatic conditions(\cite{matesanz2010global}). \\

The starting point of our work is the yearly dynamics of a population of trees structured by two phenotypic traits: seed production (number of seeds) and a threshold temperature inducing dormancy (summer dormancy). We assume a full inheritability of these traits, that are therefore breeding values. With environmental conditions as input, we model the dynamics of a tree population over a long period of time (multiannual). We assume that individuals reproduce sexually, which is represented by a non-linear integral operator. We couple this with an ecological dynamics, where water stress can lead to the death of adult trees. Tree death typically results in freeing up space that allows a rejuvenation process. We therefore consider that deaths are compensated by the maturation of young individuals born from seeds produced the year before. The individual phenotypic traits, as well as plastic effects will affect the life cycle of the trees. Our model should provide an interesting tool to study the resulting evolutionary dynamics on the population and its adaptive capacities. Our focus on relevant environmental input and on phenological aspects should be an asset for future interaction with other approaches existing in forest management and conservation ecology.  \\

An important component of our model is a term representing the effect of sexual reproductions. We use the Fisher’s Infinitesimal Model (\cite{fisher1919xv,barton2017infinitesimal}), which provides a framework for understanding the evolution of continuous quantitative traits in sexually reproducing populations. Structured populations for asexual populations have attracted the attention of the mathematical community (\cite{diekmann2005dynamics,champagnat2007invasion, desvillettes2008selection,lorz2011dirac}), and, more recently, mathematical models for sexual populations have been studied. Heuristic macroscopic limits of sexual structured population models have been proposed (\cite{mirrahimi2013dynamics,dekens2022evolutionary,dekens2022best}), and several rigorous framework for such limits have been constructed. In (\cite{calvez2019asymptotic, patout2023cauchy}), a regularity approach was proposed, based on the contraction properties of the infinitesimal operator on regularity functional spaces, while in (\cite{raoul2017macroscopic,calvez2024ergodicity,frouvelle2023fisher,calvez2024ergodicity}) the contraction was expressed in terms of Wasserstein distances. Finally, in (\cite{guerand2023moment}), an alternative approach based on and moments has been introduced.\\

Our manuscript is structured as follows. We first introduce a detailed annual structured model in Section~\ref{subsec:annualmodel}, where we describe precisely the life cycle of individual trees as a function of their phenotypic traits and the environmental conditions.  In Section~\ref{subsec:continuousmodel}, we take advantage of the low mortality rate of adult trees to derive a simplified continuous structured model. We build an existence and uniqueness setting for solutions of both models and describe the asymptotic limit between these models. The model is simplified further in Section~\ref{sec:macromain} by a second asymptotic limit, where we assume that the phenotypic variance of the tree population is small. This limit is related to the macroscopic limit mentioned in the last paragraph, and we do not provide a rigorous justification of this limit. This argument however allows us to obtain a simple model describing the dynamics of the mean phenotypic traits of the population through a system of two coupled ordinary differential equations. Finally, in Section~\ref{sec:numerics} we use numerical simulations to study the final ordinary differential equation model, and we use this model to investigate the effect of a shift in the environmental condition of the tree population : a temperature increase and/or a drop in precipitation levels. This manuscript is concluded by a discussion.

\section{Models, asymptotic limits and main results}\label{sec:main}
\subsection{The annual structured model}\label{subsec:annualmodel}

In this section we introduce a detailed model describing the dynamics of a population of trees. The population of trees is structured by two fully inheritable phenotypic traits: a seed production trait $x\in\mathbb R$ which quantifies the rate at which the tree produces seeds during the reproduction season, and a dormancy trait $y\in\mathbb R$, describing the temperature above which the individual tree becomes dormant. To model the effect of the two traits accurately, we consider a discrete year structure with $k\in\mathbb N\cup \{0\}$ numbering the years. We also consider a season time $s\in [0,1]$: $s=0$ corresponds to the beginning of the year $k$, and $s=1$ to the end of year $k$ and the transition to year $k+1$. We assume that the end of the year corresponds to the end of the summer season, and the end of the summer dormancy of trees, if a summer dormancy occurs.

The meteorological environment is described through  $T_k(s)$, the temperature in year $k$ and seasonal time $s$, and the yearly precipitation quantity $P_k> 0$. We assume moreover that \[T_k(s)=T^M_k+V_k T^V(s),\]
where $(T^M_k)_k$ is the  mean annual temperature and $(T^V_k)_k$ quantifies the variation of the temperature around its yearly mean, and $V(\cdot)$ is the amplitude  of yearly temperature variations (we  assume $ T^V(0)=T^V(1)=0$). We assume that $(P_k, T^M_k,V_k)_k$ is a bounded sequence. These environmental data will have an impact on the availability of water, denoted by $\phi^{water}_k(s,x)\in\{0,1\}$ (the water is lacking for an individual of seed production trait $x\in\mathbb R$ at seasonal time $s\in[0,1]$ and year $k\in\mathbb N\cup\{0\}$ if $\phi^{water}_k(s,x)=0$, it is sufficient if $\phi^{water}_k(s,x)=1$), and the dormancy of individuals, denoted by $\phi^{plast}_k(s,y)\in[0,1]$ (an individual of threshold temperature $y\in\mathbb R$ is active at seasonal time $s\in[0,1]$ and year $k\in\mathbb N\cup\{0\}$ if $\phi^{plast}_k(s,y)\sim 1$, and is dormant if $\phi^{plast}_k(s,y)\sim 0$). We will define these functions precisely later on.

\medskip

When water is available, we assume that mature individuals die at a constant rate $\varepsilon d>0$; with $d\geq 0$ a base death rate that is independent from environmental conditions; and that this death rate is increased when individuals encounter water stress. More precisely, an additional death rate proportional to the temperature $T_k(s)$ appears when the water is lacking and when the individuals are not in dormancy (that is when $\phi^{water}_k(s,x)=0$ and $\phi^{plast}_k(s,y)= 1$). The individuals are in dormancy if $\phi^{plast}_k(s,y)=1$, and this represents summer dormancy: we do not consider winter dormancy in this manuscript. During the year $k\in\mathbb N$, the mature population then satisfies, for $s\in[0,1]$
\begin{align}
    \partial_s m_k(s,x,y) =-\varepsilon d - \varepsilon \phi^{plast}_k(s,y)\gamma (T_k(s))^+\left(1-\phi^{water}_k(s,x)\right) m_k(s,x,y).\label{model:m00}
\end{align}

Note that the notation $x^+$, for $x\in\mathbb R$, designates the positive part of $x$, here and throughout this manuscript.   Experimental and theoretical work have shown that if the ground water is too limited (that is if $\phi^{water}_k(s,x)=0$ here), plants evapotranspiration diminishes linearly (\cite{sinclair2005daily}), and the plant water stress grows linearly (\cite{rizza2004use}). The effect of water stress on tree mortality is not fully described yet (see \cite{mcdowell2008mechanisms} for field data analysis) and we therefore decided to use a simple linear relation, leading to the model above. In equation \eqref{model:m00}, we have included $\varepsilon>0$, which is a constant parameter, that we will eventually assume to be small, while $\gamma>0$ quantifies the effect of temperature on mortality. $\varepsilon>0$ captures the fact that only a small fraction of mature individuals die during a given year, which is a reasonable assumption for tree species with a long life expectancy, provided the climatic conditions are not too extreme, and we will consider the asymptotics $\varepsilon\to 0$ in Section~\ref{subsec:continuousmodel}. This equation on mature individuals is completed by the following equation that provides the initial value of the first equation at the beginning of each year $k\in\mathbb N$:
\[m_k(0,x,y)=m_{k-1}(1,x,y)+\frac{1-\iint_{\mathbb R^2} m_{k-1}(\hat x,\hat y)\,d\hat x\,d\hat y}{\iint_{\mathbb R^2} s_{k-1}(\hat x,\hat y)\,d\hat x\,d\hat y} s_{k-1}(x,y).\]
The first term on the right hand side of this equation reflects the fact that the mature individuals alive at the end of the previous year are still present. The second term represents the maturation of the seeds $s_{k-1}(x,y)$ (seeds with traits $x$ and $y$) produced during the previous year: we do not consider an intermediate juvenile state between the seed and the mature state. We assume that any space liberated by the death of mature individuals is occupied by new mature individuals, leading to a fixed population size at the beginning of each year, $\iint_{\mathbb R^2} m_k(x,y)\,dx\,dy=1$; this explains the factor in front of the last term of the equation. Our modeling choice implies $\iint_{\mathbb R^2}m_k(0,x,y)\,dx\,dy=1$, so that the total size of the population remains constant; it is however possible to monitor the mortality rate across years as an output of the model.

\medskip

During the year $k\in\mathbb N\cup\{0\}$ and at season time $s\in[0,1]$, an adult with traits $(x,y)$ produces eggs (ie seeds that need fertilization by pollen) at a rate $x^+$ (the positive part of the seed production  trait $x\in\mathbb R$), provided water is not lacking (that is if $\phi^{water}_k(s,x)=1$), and provided the individual is not dormant (that is if $\phi^{plast}_k(s,y)\sim 1$). If water is lacking or the plant is dormant, no egg is produced. The number of eggs produced during year $k$ at time $s\in[0,1]$ is then given by the following distribution
\[(x_1,y_1)\mapsto x_1^+m_k(s,x_1, y_1) \phi^{plast}_k(s,y_1)\phi^{water}_k(s,x_1).\]
Pollen is produced by all mature individuals that do not lack water and are not dormant, so that the quantity of pollen present at season time $s$ is given by the following distribution 
\[(x_2,y_2)\mapsto m_k(s, x_2,y_2) \phi^{plast}_k(s,y_2)\phi^{water}_k(s,x_2).\]
We assume that the pollen fertilizing an egg is drawn uniformly among the pollen present at time $s$, and we introduce the parameter $\eta>0$ to model the fact that if the quantity of pollen production is very low (that is if $\iint m_k(s,\hat x,\hat y) \phi^{plast}_k(s,\hat y)\phi^{water}_k(s,\hat x) \,d\hat x\,d\hat y\ll \eta$), fertilization is unlikely to occur. More precisely, the probability that an egg is fertilized is 
\[(x_2,y_2)\mapsto \frac{\iint_{\mathbb R^2} m_k(s,\hat x,\hat y) \phi^{plast}_k(s,\hat y)\phi^{water}_k(s,\hat x) \,d\hat x\,d\hat y}{\eta+\iint_{\mathbb R^2} m_k(s,\hat x,\hat y) \phi^{plast}_k(s,\hat y)\phi^{water}_k(s,\hat x) \,d\hat x\,d\hat y},\]
and if the fertilisation occurs, the law of the traits of the fertilizing pollen is given by
\[(x_2,y_2)\mapsto \frac{m_k(s, x_2,y_2) \phi^{plast}_k(s,y_2)\phi^{water}_k(s,x_2)}{\iint_{\mathbb R^2} m_k(s,\hat x,\hat y) \phi^{plast}_k(s,\hat y)\phi^{water}_k(s,\hat x) \,d\hat x\,d\hat y}.\]
If we consider a seed produced by an egg of traits $(x_1,y_1)$ and a pollen of traits $(x_2,y_2)$, the traits of the seed follows a normal law of covariance 
\[\left(\begin{array}{cc}
    \sigma_x^2 &0  \\
     0&\sigma_y^2 
\end{array}\right),\]
centered on the average traits of the parents. We assume that $\sigma_x,\sigma_y>0$. This is the so-called infinitesimal model (\cite{fisher1919xv,barton2017infinitesimal}), which represents the effect of sexual reproduction on phenotypic traits. The law of the seed's traits is then
\begin{equation}\label{eq:infinitesimal}
    (x,y)\mapsto\Gamma_{\sigma_x^2}\left(\frac{x_{1}+x_{2}}{2}-x\right)\Gamma_{\sigma_y^2}\left(\frac{y_{1}+y_{2}}{2}-y\right).
\end{equation}
We also assume a constant immigration of seeds at rate $\nu> 0$, and the traits of these immigrant seeds are normally distributed around $(0,0)$, with the covariance matrix $\left(\begin{array}{cc}
    2\sigma_x^2 &0  \\
     0&2\sigma_y^2 
\end{array}\right)$. Bringing all these assumptions together, the seeds produced during the year $k$, structured by the traits $(x,y)$, are given by
\begin{align*}
    &s_k(x,y)= \nu \Gamma_{2\sigma_x^2}(x)\Gamma_{2\sigma_y^2}(y)\\
    &\quad +\int_0^1\iint_{\mathbb R^2}\iint_{\mathbb R^2}  \Gamma_{\sigma_x^2}\left(x-\frac{x_{1}+x_{2}}{2}\right)\Gamma_{\sigma_y^2}\left(y-\frac{y_{1}+y_{2}}{2}\right)  x_1^+m_k(s,x_1,y_1)\phi^{plast}_k(s,y_1)\phi^{water}_k(s,x_1) \\
    &\hspace{2.5cm} \frac{m_k(s, x_2,y_2) \phi^{plast}_k(s,y_2)\phi^{water}_k(s,x_2)}{\eta+\iint_{\mathbb R^2} m_k(s,\hat x,\hat y) \phi^{plast}_k(s,\hat y)\phi^{water}_k(s,\hat x) \,d\hat x\,d\hat y}  \,dx_1\,dx_2\,dy_1\,dy_2\,ds,
\end{align*}
and as mentioned earlier, a fraction of these seeds will turn into mature individuals at the beginning of year $k+1$.

\medskip

\textcolor{blue}{The functions $\phi_k^{water}$ and $\phi_k^{plast}$ constructed above define the effect of water stress and dormancy on individual trees. We have constructed these in an explicit manner, from simple biological processes. We have however used simplifying assumptions. We have assumed that trees consume a given \emph{water budget} $P_k$ from precipitations, but in practice,  soil moisture and plant water potential respond to multi-year storage dynamics. We have also assumed that trees enter dormancy as soon as the temperatures exceed the threshold temperature $y$, even though a progressive effect of high temperatures would be more realistic. Finally, multi-year plastic effect should also play an important rule, for instance the depth of tree roots can be impacted by environmental conditions and this multi-year plastic trait impacts the ability of the tree to access water. This choices were made to obtain tractable effects of the model parameters on the dynamics of the population, and we believe it will be possible to use the analysis framework we introduce to investigate to investigate the effect of more detailed biological processes;}\textcolor{red}{ we refer to \cite{satake2022cross} for a review of detailed phenology models.}

\medskip

We will now define the function $\phi^{water}_k(s,x)\geq 0$, that represents the availability of water for an adult individual of seed production trait $x\in\mathbb R$ in year $k\in\mathbb N\cup\{0\}$, at the season time $s\in[0,1]$. We assume that each individual receives an annual quantity of water $P_k$. When individuals are not dormant, mature individuals consume water at a rate $1+\alpha x^+$, where $x^+$ refers to the rate of seed production of the individual. Moreover, we assume that heat induces an additional consumption of water (for each mature individual that is not dormant) proportional to the temperature. Indeed, a large part of a plant's water consumption is related to its evapotranspiration, which can be modeled as a linear function of the temperature (see \cite{penman1948natural,thornthwaite1948approach}), provided enough ground water is available to the plant. These assumptions lead to a total water consumption in the year $k\in\mathbb N\cup\{0\}$ seasonal time $s\geq 0$ given by $(1+\alpha x^+)s+\beta\int_{0}^sT_k(\tau)\,d\tau$, provided the individual is not dormant. To determine if the water is available at season time $s$ (that is to see if $\phi^{water}_k(s,x)=1$) or if it is lacking (that is $\phi^{water}_k(s,x)= 0$), we compare the water consumption to  the precipitations level during year $k$, that is $P_k$. We obtain:
\begin{equation}\label{def:phiwater}
\phi^{water}_k(s,x)=\left\{\begin{array}{l}0\textrm{ if }(1+\alpha x^+)s+\beta\int_{0}^s(T_k(\tau))^+\,d\tau> P_k,\\
1\textrm{ otherwise}.
\end{array}\right.
\end{equation} 
We should also define the function $\phi^{plast}_k(s,y)\in[0,1]$, which indicates if the individual of threshold temperature trait $y$ is active (that is $\phi^{plast}_k(s,y)\sim 1$) or dormant (that is $\phi^{plast}_k(s,y)\sim 0$). In the definition below, we define $s\mapsto \phi^{plast}_k(s,x)$ as a Lipschitz function valued in $[0,1]$. We assume that the dormancy is heat triggered: when the temperature $T_k^0(s)$ is higher than the threshold temperature trait $y$, individuals become dormant with a given rate $\xi$. Then,
\begin{equation}\label{def:phiplast}
\phi^{plast}_k(s,y)= e^{-\xi \int_{0}^s 1_{T_k(\tau)>y}\,d\tau}.
\end{equation}
Note that the population $m_k(t,x,y)$ is structured by the phenotypic traits $(x,y)\in\mathbb R$ that are actually breeding values: the value of these traits is determined at birth, and are fully inherited. 


\bigskip

The annual structured model, for $\varepsilon>0$, is then
\begin{equation}\label{model:m}
    \left\{\begin{array}{l}\partial_s m_k(s,x,y) = -\varepsilon d- \varepsilon \phi^{plast}_k(s,y)\gamma (T_k(s))^+\left(1-\phi^{water}_k(s,x)\right) m_k(s,x,y) ,\textrm{ for }(k,s,x,y)\in(\mathbb N\cup\{0\})\times\mathbb R_+\times\mathbb R^2,\\ \\
m_k(0,x,y)=m_{k-1}(1,x,y)+\frac{1-\iint_{\mathbb R^2} m_{k-1}(1,\hat x,\hat y)\,d\hat x\,d\hat y}{\iint_{\mathbb R^2} s_{k-1}(\hat x,\hat y)\,d\hat x\,d\hat y} s_{k-1}(x,y),\textrm{ for }(k,x,y)\in\mathbb N\times\mathbb R^2,\\ \\
m_0(0,x,y)=m^0(x,y),\textrm{ for }(x,y)\in\mathbb R^2,\\ \\
s_k(x,y)= \nu \Gamma_{2\sigma_x^2}(x)\Gamma_{2\sigma_y^2}(y)\\ 
\phantom{erazer}+\int_{0}^1\iint_{\mathbb R^2}\iint_{\mathbb R^2}  \Gamma_{\sigma_x^2}\left(x-\frac{x_{1}+x_{2}}{2}\right)\Gamma_{\sigma_y^2}\left(y-\frac{y_{1}+y_{2}}{2}\right)x_1^+ m_k(s,x_1, y_1) \phi^{plast}_k(s,y_1)\phi^{water}_k(s,x_1)  \\ 
\phantom{erazew<wr} \frac{m_k(s, x_2,y_2) \phi^{plast}_k(s,y_2)\phi^{water}_k(s,x_2)}{\eta+\iint_{\mathbb R^2} m_k(s,\hat x,\hat y) \phi^{plast}_k(s,\hat y)\phi^{water}_k(s,\hat x) \,d\hat x\,d\hat y}  \,dx_1\,dx_2\,dy_1\,dy_2\,ds,\textrm{ for }(k,x,y)\in(\mathbb N\cup\{0\})\times\mathbb R^2,
\end{array}\right.
\end{equation}
where $\phi^{plast}_k$, $\phi^{water}_k$ are defined by \eqref{def:phiplast}, \eqref{def:phiwater} respectively. This system defines the population $m_k(x,y)$, in the sense of the existence and uniqueness setting given by the following theorem. In that theorem, we denote by $L^1((1+e^x)\,dx\,dy,\mathbb R_+)$ the set of non-negative functions on $\mathbb R^2$ that are integrable against a weight $(x,y)\mapsto (1+e^x)\,dx\,dy$:

\begin{theorem}\label{thm:m-existence}
    Assume $m^0\in L^1((1+e^x)\,dx\,dy,\mathbb R_+)$, with  $\iint_{\mathbb R^2}m^0(x,y)\,dx\,dy=1$. There exists a unique global non-negative solution $(m_k,p_k)\in \left(L^\infty([0,1],L^1((1+e^x)\,dx\,dy))\right)^2$ to the annual structured model \eqref{model:m}, for $k\in\mathbb N\cup\{0\}$. Moreover, for $k\in\mathbb N\cup\{0\}$,
    \begin{equation}\label{eq:mp1}
        \iint_{\mathbb R^2} m_k(0,x,y)\,dx\,dy=1.
    \end{equation}
\end{theorem}

\subsection{Asymptotic limit to an intermediate continuous structured model}\label{subsec:continuousmodel}

The annual structured model \eqref{model:m} involves a parameter $\varepsilon>0$, which quantifies the death rate of adult individuals: the life expectancy of adults trees in our model is of the order of $1/\varepsilon$. For numerous species of trees, year-to-year mortality is low for adults (typically of the order of $1\%$, see (\cite{das2016trees})), and we take advantage of this to propose an asymptotic limit of the model. Heuristically, when $\varepsilon>0$ is small, \eqref{model:m} implies , for $(k,s,x,y)\in\mathbb N\times[0,1]\times \mathbb R\times\mathbb R$,
\[m_k(1,x,y)=\left(1-\varepsilon a_k(1,x,y)\right)m_k(0,x,y)+\mathcal O(\varepsilon^2),\]
where 
\begin{align}\label{eq:gk}
    a_k(s,x,y) &= ds+\int_0^{s} \phi^{plast}_k(\tau,y) \gamma (T_k(\tau))^+\left(1-\phi^{water}_k(\tau,x)\right)\,d\tau.
\end{align}
Then,
\begin{align}
 &   m_{k+1}(0,x,y) \sim e^{-\varepsilon a_k(1,x,y)}m_k(0,x,y) \nonumber\\
& \quad  +\frac{1- \iint_{\mathbb R^2} e^{-\varepsilon a_k(1,\hat x, \hat y)} m_{k}(0,\hat x,\hat y)\,d\hat x\,d\hat y}{\iint_{\mathbb R^2} s_{k}(\hat x,\hat y)\,d\hat x\,d\hat y} \nonumber\\
&\quad \bigg(  \int_{0}^{1}\iint_{\mathbb R^2}\iint_{\mathbb R^2}  \Gamma_{\sigma_x^2}\left(x-\frac{x_{1}+x_{2}}{2}\right)\Gamma_{\sigma_y^2}\left(y-\frac{y_{1}+y_{2}}{2}\right)x_1^+ e^{-\varepsilon a_k(s,x_1,y_1)}m_k(0,x_1,y_1)\phi^{water}_k(s,x_1)\phi^{plast}_k(s,y_1) \nonumber \\
& \qquad   \frac{ e^{-\varepsilon a_k(s,x_2,y_2)}m_k(0,x_2,y_2) \phi^{water}_k(s,x_2)\phi^{plast}_k(s,y_2)}{\eta+\iint_{\mathbb R^2}  e^{-\varepsilon a_k(s,\hat x,\hat y)}m_k(0,\hat x,\hat y) \phi^{water}_k(s,\hat x)\phi^{plast}_k(s,\hat y) \,d\hat x\,d\hat y}  \,dx_1\,dx_2\,dy_1\,dy_2\,ds + \nu \Gamma_{2\sigma_x^2}(x)\Gamma_{2\sigma_y^2}(y) \bigg).\label{model:m-epsilon}
\end{align}
To pass to the limit in the integral term above, we assume that the environmental parameters $T^{M,\varepsilon},V^\varepsilon,P^\varepsilon$ change gradually from one year to another. More precisely, for $t\geq 0$ and $\varepsilon>0$ small,
\begin{equation}\label{ass:weather}
    (T^{M,\varepsilon},V^\varepsilon,P^\varepsilon)_{\lfloor t/\varepsilon\rfloor}=(\bar T^M,\bar V,\bar P)(t)+\mathcal O(\varepsilon),
\end{equation}
which corresponds to a situation where the climate changes on a time scale of $1/\varepsilon$, which is the time scale of the life expectancy of the trees. We assume that the functions $\bar T^M$, $\bar V$, $\bar P$ are Lipschitz continuous. Then, for  $s\in[0,1]$,
\begin{equation}\label{eq:limitT}
    T^{\varepsilon}_{\lfloor t/\varepsilon\rfloor}(s)=\bar T(t,s)+\mathcal O(\varepsilon),
\end{equation}
where $\bar T(t,s):=\bar T^M(t)+\bar V(t) T^V(s)$. We may then define 
\begin{equation}\label{def:barphip}
\bar \phi^{plast}(t,s,y)= e^{-\xi\int_0^s 1_{\bar T(t,\sigma)>y}\,d\sigma}
\end{equation}
\begin{equation}\label{def:barphiw}
\bar \phi^{water}(t,s,x)=\left\{\begin{array}{l}0\textrm{ if }(1+\alpha x^+)s+\beta\int_{0}^s (\bar T(t,\tau))^+\,d\tau> \bar P(t),\\
1\textrm{ otherwise},
\end{array}\right.
\end{equation}
and the population $m_k$ then satisfies
\begin{align*}
 &   \frac{m_{\lfloor t/\varepsilon\rfloor+1}(0,x,y)-m_{\lfloor t/\varepsilon\rfloor}(0,x,y)}{\varepsilon} \sim - a_{\lfloor t/\varepsilon\rfloor}(1,x,y)m_{\lfloor t/\varepsilon\rfloor}(0,x,y)  + \mathcal O(\varepsilon^2) \nonumber\\
& \quad  +  \left(\iint_{\mathbb R^2} a_{\lfloor t/\varepsilon \rfloor}(1,x,y) m_{\lfloor t/\varepsilon \rfloor}(0,x,y) + \mathcal O(\varepsilon^2) dx\,dy\right)\\
&\qquad\frac{1}{\bar R_{\lfloor t/\varepsilon\rfloor}[m] +\nu} \bigg[\iint_{\mathbb R^2}\iint_{\mathbb R^2}  \Gamma_{\sigma_x^2}\left(x-\frac{x_{1}+x_{2}}{2}\right)\Gamma_{\sigma_y^2}\left(y-\frac{y_{1}+y_{2}}{2}\right)x_1^+m_{\lfloor t/\varepsilon\rfloor}(0,x_1,y_1) m_{\lfloor t/\varepsilon\rfloor}(0,x_2,y_2) \nonumber\\
&\phantom{\qquad \frac{1}{\bar R_{\lfloor t/\varepsilon\rfloor}[m] +\nu} \Big[\iint_{\mathbb R^2}\iint_{\mathbb R^2} }\rho_{\lfloor t/\varepsilon\rfloor}[m](x_1,y_1,x_2,y_2)\, dx_1\,dx_2\,dy_1\,dy_2\,ds  +\nu \Gamma_{2\sigma_x^2}(x) \Gamma_{2\sigma_y^2}(y) \bigg],  \nonumber
\end{align*}
where
\[Q_{\lfloor t/\varepsilon \rfloor}[m]=\iint_{\mathbb R^2} m_{\lfloor t/\varepsilon\rfloor}(0,\hat x,\hat y)\int_0^1 \bar \phi^{water}(t,s,\hat x)\bar \phi^{plast}(t,s,\hat y)\,ds\,d\hat x\,d\hat y,\]
\begin{align*}
    &R_{\lfloor t/\varepsilon\rfloor}[m] =\iint_{\mathbb R^2}\iint_{\mathbb R^2} \hat x^+ m_{\lfloor t/\varepsilon\rfloor}(0,\hat x,\hat y) m_{\lfloor t/\varepsilon\rfloor}(0,\tilde x,\tilde y)\\
    &\phantom{R_{\lfloor t/\varepsilon\rfloor}[m] =\iint_{\mathbb R^2}\iint_{\mathbb R^2} }\int_0^1 \bar \phi^{water}(t,s,\hat x)\bar \phi^{plast}(t,s,\hat y) \frac{\bar \phi^{water}(t,s,\tilde x)\bar \phi^{plast}(t,s,\tilde y)}{\eta + Q_{\lfloor t/\varepsilon \rfloor}[m]}\,ds\,d\hat x\,d\hat y\,d\tilde x\,d\tilde y,
\end{align*}
\[ \rho_{\lfloor t/\varepsilon\rfloor}[m](x_1,y_1,x_2,y_2) =  \int_{0}^{1}  \frac{\bar \phi^{plast}(t,s,y_1)\bar \phi^{plast}(t,s,y_2)\bar \phi^{water}(t,s,x_1)\bar \phi^{water}(t,s,x_2)}{\eta+Q_{\lfloor t/\varepsilon\rfloor}[m]} \,ds. \]

We therefore expect the population to converge to a limit when $\varepsilon>0$ is small:
\[m_{\lfloor s/\varepsilon\rfloor}(0,x,y)\xrightarrow[\varepsilon\to 0]{}n(s,x,y),\] where $n$ is the solution of the following continuous structured model:
\begin{align}
    &\partial_t n(t,x,y)    = - \bar a(t,x,y) n(t,x,y) \nonumber \\
    &\quad + \frac{\iint_{\mathbb R^2} \bar a(t,\hat x,\hat y) n(t,\hat x,\hat y) \,d\hat x\,d\hat y}{\nu+\bar R[n](t)} \bigg[\iint_{\mathbb R^2}\iint_{\mathbb R^2}  \Gamma_{\sigma_x^2}\left(x-\frac{x_{1}+x_{2}}{2}\right)\Gamma_{\sigma_y^2}\left(y-\frac{y_{1}+y_{2}}{2}\right)  x_1^+n(t,x_1, y_1)n(t, x_2,y_2) \nonumber \\
    & \phantom{\quad + \frac{\iint_{\mathbb R^2} \bar a(t,\hat x,\hat y) n(t,\hat x,\hat y) \,d\hat x\,d\hat y}{\nu+\bar R[n](t)} \bigg[\iint_{\mathbb R^2}\iint_{\mathbb R^2}  }\bar\rho[n(t,\cdot,\cdot)](t,x_1,y_1,x_2,y_2) \,dx_1\,dx_2\,dy_1\,dy_2+\nu \Gamma_{2\sigma_x^2}(x)\Gamma_{2\sigma_y^2}(y)\bigg],
  \label{model:n}
\end{align}
with 
\begin{equation}\label{def:gt}
    \bar a(t,x,y) = d+\int_0^{1} \bar \phi^{plast}(t,s,y) \gamma (\bar T(t,s))^+\left(1-\bar \phi^{water}(t,s,x)\right)\,ds,
\end{equation}
\begin{equation}\label{def:barQ}
\bar Q[n](t) =  \iint_{\mathbb R^2} n(t,\hat x,\hat y)\int_0^1 \bar \phi^{water}(t,s,\hat x)\bar \phi^{plast}(t,s,\hat y)\,ds\,d\hat x\,d\hat y,
\end{equation}
\[
\bar R[n](t) = \iint_{\mathbb R^2}\iint_{\mathbb R^2} \hat x^+ n(t,\hat x,\hat y)n(t,\tilde x,\tilde y)\bar\rho[n(t,\cdot,\cdot)](t,\hat x,\hat y,\tilde x,\tilde y) \,d\hat x\,d\hat y\,d\tilde x\,d\tilde y,
\]
\[ \bar\rho[n(t,\cdot,\cdot)](t,x_1,y_1,x_2,y_2) =  \int_{0}^{1}  \frac{\bar \phi^{plast}(t,s,y_1)\bar \phi^{plast}(t,s,y_2)\bar \phi^{water}(t,s,x_1)\bar \phi^{water}(t,s,x_2)}{\eta+\bar Q[n](t)} \,ds. \]

We introduce an existence and uniqueness setting for solutions $n$ of \eqref{model:n} in the following theorem:
\begin{theorem}\label{thm:n-existence}
    Assume $n^0 \in L^{1}(\mathbb R^{2}, \mathbb R_+)\cap W^{1,\infty}(\mathbb R^2)$. There exists a unique solution $n \in  L^\infty(\mathbb R_+,L^{1}( \mathbb R^2,\mathbb R_+)) $ of \ref{model:n} with initial data $n^0$. If $n^0 \in L^\infty(\mathbb R^{2}, \mathbb R_+)$, then the solution $n$ is bounded locally in time:
    \[\|n(t,\cdot,\cdot)\|_{L^\infty(\mathbb R^2)}\leq C(1+t),\]
    for some $C>0$. Moreover, if $n^0$ is Lipschitz continuous, there is $C>0$ such that for $t\geq 0$, $x,x',y,y'\in\mathbb R$ satisfying $y'\leq y$,
    \begin{equation}\label{eq:Lips}
        \big|n(t,x,y)-n(t,x',y')\big|\leq C(1+t^2)\left(|x-x'|+|y-y'|\right).
    \end{equation}
    Finally, there is $C>0$ such that
    \begin{equation}\label{est:moment2}
        \iint_{\mathbb R^2} x^2 n(t,x,y)\,dx\,dy\leq C e^{Ct}.
    \end{equation}
\end{theorem}

We can also build a rigorous connection between the annual structured model \eqref{model:m} and the continuous structured model \eqref{model:n} when $\varepsilon>0$ is asymptotically small, thanks to the following result:
\begin{theorem}\label{thm:n-to-m} 
Assume $n^0 \in L^{1}(\mathbb R^{2}, \mathbb R_+)\cap W^{1,\infty}(\mathbb R^2)$. 
For $\varepsilon>0$, let $(m_k^\varepsilon,s_k^\varepsilon) \in L^{1}([0,1],\mathbb{R}^2 \times L^{\infty}(\mathbb{R}^{2})$the solution of \eqref{model:n} in the sense of Theorem~\ref{thm:m-existence}, with initial data $m^0:=n^0$. Let $n \in L^{1}(\mathbb{R_+} \times \mathbb{R}^2)$ the solution of \ref{model:n} in the sense of Theorem~\ref{thm:n-existence}, with initial data $n^0$. For any $T\geq 0$, we have 
\[\lim_{\varepsilon\to 0}\|n(t,x,y)-m_{\lfloor t/\varepsilon\rfloor}^\varepsilon(0,x,y)\|_{L^\infty([0,T],L^1(\mathbb R^2))}=0.\]
\end{theorem}


\subsection{Derivation of the macroscopic model} 
\label{sec:macromain}

The model \eqref{model:m} was constructed according to the life cycle processes of individual trees. The drawback of this derivation is the complexity of the model, which was reduced thanks to an asymptotic limit leading to the continuous structured model \eqref{model:n}. That second model remains complex and in this chapter we simplify it further. Our idea is to use the asymptotic limit where the phenotypic variance of the population is small, that is $\sigma_x^2>0$ and $\sigma_y^2>0$ small. This assumption is related to the weak selection that is often used in population genetics (\cite{wakeley2005limits}), and it has been used in mathematical studies involving the infinitesimal model (\cite{patout2023cauchy,raoul2017macroscopic}). Specifically, we assume:
\[(\sigma_x^2,\sigma_y^2)=\sigma^2(\bar \sigma_x^2,\bar \sigma_y^2),\]
for some parameter  $\sigma>0$ small. We also assume $\nu=\eta=0$ and $\xi\gg 1$ to simplify the notations in this section. 

\medskip

We define $(X(t), Y(t))$ as the mean phenotypic traits of the population: $X(t)$ is the mean  seed production trait; $Y(t)$ is the mean dormancy trait. Since $\iint_{\mathbb R^2} n(t,x,y)\,dx\,dy\equiv 1$,
\begin{equation}\label{def:meanpheno}
    X(t):=\iint_{\mathbb R^2} x\,n(t,x,y)\,dx\,dy,\quad Y(t):=\iint_{\mathbb R^2} y\,n(t,x,y)\,dx\,dy.
\end{equation}
When $\sigma>0$ is small, heuristically, the solution $n$ will be distributed as a normal distribution with covariance
\[\sigma^2\left(\begin{array}{cc}
    \bar\sigma_x^2 &0  \\
     0&\bar\sigma_y^2 
\end{array}\right),\]
centered around $(X(t),Y(t))$, that is
\begin{equation}\label{eq:Maxwellian}
    \tilde n(t,x,y)\sim \Gamma_{2\sigma^2\bar \sigma_x^2}\left(x-X(t)\right)\Gamma_{2\sigma^2\bar \sigma_y^2}\left(y-Y(t)\right).
\end{equation}
We explain why $\sigma>0$ small leads to this approximation  in Section~\ref{sec:Appendixmacro}. Thanks to this approximation formula, we simply need to describe the time dynamics of the mean phenotypic traits $X(t)$ and $Y(t)$, which, thanks to their definition \eqref{def:meanpheno} and to the model \eqref{model:n}, satisfy
\begin{align}
    &\frac d{dt}X(t)  = - \iint_{\mathbb R^2} x\bar a(t,x,y) n(t,x,y)\,dx\,dy \nonumber \\
    &\quad + \frac{\iint_{\mathbb R^2} \bar a(t,x,y) n(t,x,y) \,dx\,dy}{\iint_{\mathbb R^2}\iint_{\mathbb R^2}  x_1^+n(t,x_1, y_1)n(t, x_2,y_2)\bar\rho[n(t,\cdot,\cdot)](t,x_1,y_1,x_2,y_2) \,dx_1\,dy_1\,dx_2\,dy_2}\nonumber  \\
    &\qquad \iint_{\mathbb R^2}\iint_{\mathbb R^2}  \frac{x_{1}+x_{2}}{2}x_1^+n(t,x_1, y_1)n(t, x_2,y_2) \bar\rho[n(t,\cdot,\cdot)](t,x_1,y_1,x_2,y_2) \,dx_1\,dy_1\,dx_2\,dy_2,\label{eq:dX0}
\end{align}
\begin{align}
    &\frac d{dt}Y(t)  = - \iint_{\mathbb R^2} y\bar a(t,x,y) n(t,x,y)\,dx\,dy \nonumber \\
    &\quad + \frac{\iint_{\mathbb R^2} \bar a(t,x,y) n(t,x,y) \,dx\,dy}{\iint_{\mathbb R^2}\iint_{\mathbb R^2}  x_1^+n(t,x_1, y_1)n(t, x_2,y_2)\bar\rho[n(t,\cdot,\cdot)](t,x_1,y_1,x_2,y_2) \,dx_1\,dy_1\,dx_2\,dy_2}\nonumber  \\
    &\qquad \iint_{\mathbb R^2}\iint_{\mathbb R^2}  \frac{y_{1}+y_{2}}{2}x_1^+n(t,x_1, y_1)n(t, x_2,y_2) \bar\rho[n(t,\cdot,\cdot)](t,x_1,y_1,x_2,y_2) \,dx_1\,dy_1\,dx_2\,dy_2.\label{eq:dY0}
\end{align}
Using the approximation \eqref{eq:Maxwellian}, it is possible to simplify these expressions, as we detail in the appendix (Section~\ref{subsec:macrolimit}), leading to the following macroscopic model: 
\begin{align}\label{eq:dXY}
\left\{\begin{array}{cl}
\frac d{dt}X(t)&= 2\sigma^2\bar \sigma_x^2\bigg[-\partial_x \bar a(t,X(t),Y(t))+\frac{\bar a(t,X(t),Y(t))}{2X(t)} +\bar a(t,X(t),Y(t))\frac{\partial_{x}(\hat\rho[X(t),Y(t)])(t,X(t),Y(t))}{\hat\rho[ X(t),Y(t)](t,X(t),Y(t))} \bigg],\\ \\
    \frac d{dt}Y(t)&= 2\sigma^2\bar\sigma_y^2\bigg[-\partial_y \bar a(t,X(t),Y(t))+\bar a(t,X(t),Y(t))\frac{\partial_{y}(\hat \rho[X(t),Y(t)])(t,X(t),Y(t))}{\hat\rho[ X(t),Y(t)](t,X(t),Y(t))}\bigg],
\end{array}\right.
\end{align}
where $\bar a$ is defined by \eqref{def:gt} and 
\begin{equation}\label{def:rh}
    \hat  \rho[X(t),Y(t)](t,x,y) =  \int_{0}^{1}  \bar \phi^{plast}(t,s,y)\bar \phi^{plast}(t,s,Y(t))\bar \phi^{water}(t,s,x)\bar \phi^{water}(t,s,X(t))\,ds.
\end{equation}

Note that we do not provide a rigorous result for the derivation of \eqref{eq:dXY} from the continuous structured model \eqref{model:n}: our arguments are heuristic only. The ordinary differential equation  \eqref{eq:dXY} and approximation \eqref{eq:Maxwellian} provides a simple description of the dynamics of the population, and we detail in Section~\ref{subsec:coeff} how the coefficients of this macroscopic model can be computed. It summarizes the complex biological features taken into account to derive the original model \eqref{model:m} into a two-dimensional dynamical system, and we refer to Section~\ref{sec:numerics} for simulations of this macroscopic model. The different terms of equations \eqref{eq:dXY} can be understood biologically, as we describe below:
\begin{itemize}
    \item The first term on the right hand side of the equation on $X(t)$ comes from the impact of $X(t)$ on the mortality rate $\bar a(t,X(t),Y(t))$ and it represents the fact that evolution of $X(t)$ tends to lower the mortality rate of individuals, which is a typical effect of selection. A similar description can be made on the first term of the equation for $Y(t)$.
    \item The term $\frac{\bar a(t,X(t),Y(t))}{2X(t)}$ in the equation for $X(t)$ results from the definition of the phenotypic trait $x$: it designates the production of (the female part of) seeds, and increasing $X(t)$ therefore has a direct positive contribution on the seed production, which this term represents. The factor $\frac 12$ appears because the trait $X(t)$ only impacts one of the parents ($X(t)$ has no impact on pollen production), while the factor $\frac 1 {X(t)}$ appears because of our modeling choice of having a regulation of the population size by birth (new offspring mature only when adults in the population die).
    \item The last terms in equations for both $X(t)$ and $Y(t)$, involving $\hat \rho[X(t),Y(t)]$, represent the effect of assortative mating. For instance, if $\partial_{x}(\hat\rho[X(t),Y(t)])(t,X(t),Y(t))>0$, individual trees with a larger trait $x$ engage in more reproduction events and have more offspring. This assortative mating results from the phenological consequences of the traits we consider: at any seasonal time $s\in[0,1]$, reproduction only happens among adults that are not dormant and have sufficient water. 
\end{itemize}
Since we assume that $\xi\gg 1$, we can define the  seasonal time when individuals become dormant (or the \emph{onset of dormancy}), that we denote $D(t)$, and the seasonal time when individuals run out of water (or \emph{onset of water stress}), that we denote by $W(t)$. These plastic traits can indeed be defined as follows:
\begin{equation}\label{def:Dt}
D(t):=\min\left\{s\in(0,1); \bar T(s,t)\geq Y(t)\right\},
\end{equation}
\begin{equation}\label{Def:Wt}
    W(t):=\min\left\{s\in(0,1); (1+\alpha X(t))s+\beta\int_{0}^s (\bar T(t,\tau))^+\,d\tau\leq \bar P(t)\right\},
\end{equation}
with the convention $D(t)=1$ if the $\bar T(s,t)< Y(t)$ for $t\in(0,1)$ (and a similar convention for $W(t)$). At time $t$, the phenology of the population with mean phenotypic traits $(X(t),Y(t))$ can then be summarized using $D(t)$ and $W(t)$, which we represent in Figure~\ref{fig:csetup}. Moreover, the mortality rate $M(t)$ of the population is given by 
\begin{equation}\label{def:mortalityrate}
M(t):=\bar a(t,X(t),Y(t))= d+\int_0^{1} \bar \phi^{plast}(t,s,Y(t)) \gamma (\bar T(t,s))^+\left(1-\bar \phi^{water}(t,s,X(t))\right)\,ds
\end{equation}
The quantity $D(t),W(t)\in[0,1]$ and $M(t)\geq 0$ can be seen as plastic traits of the population, and they are paramount for applications: these traits are easier to monitor on field populations than $X(t)$ and $Y(t)$. Moreover, minimizing $t\mapsto M(t)$ (or $\sup_{t\geq 0}M(t)$) can be an objective of forest management policies.  

\begin{figure}
\centering
\begin{subfigure}{0.5\textwidth}
  \centering
  \includegraphics[width=.9\linewidth]{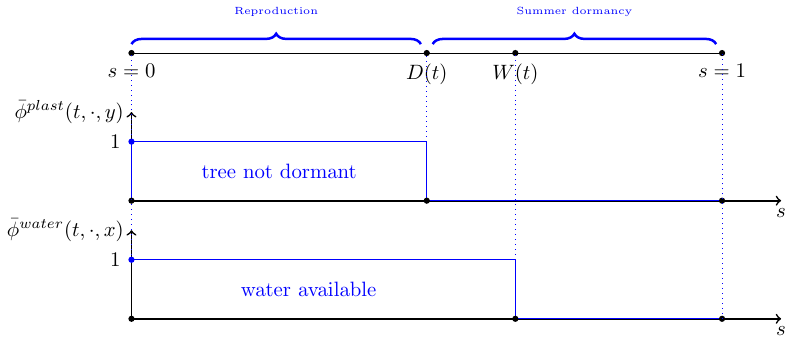}
  \caption{}
  \label{fig:sfig1}
\end{subfigure}%
\begin{subfigure}{.5\textwidth}
  \centering
  \includegraphics[width=.9\linewidth]{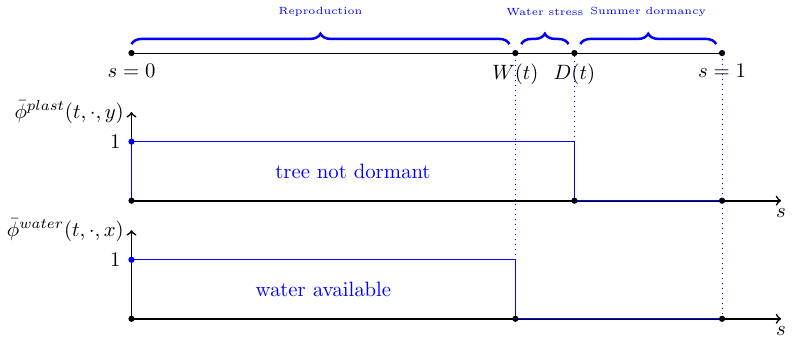}
  \caption{}
  \label{fig:sfig2}
\end{subfigure}

\caption{Phenology of a typical tree. (a) Cases where $D(t)<W(t)$, trees  produce seeds until seasonal time $s=D(t)$ when they become dormant until the next season. Note that in this case, the tree is not using all the available water. (b) Cases where $W(t)<D(t)$, the trees suffer from a lack of water from seasonal time $W(t)$ until they become dormant at seasonal time $D(t)$, this stress results in an increased mortality rate for trees during that seasonal period. }
\label{fig:csetup}
\end{figure}

\subsection{Computation of the coefficients of the macroscopic model}\label{subsec:coeff}

In this section, we explicit the coefficients of the macroscopic model \eqref{eq:dXY}.  To do so, we denote by $W(t,x)$ the seasonal time at which water is exhausted for the production of $x\geq 0$ seeds, and by $D(t,y)$ the seasonal time at which individuals with dormancy trait $y\in\mathbb R$ enter dormancy. Note that these notations are coherent with \eqref{Def:Wt}, \eqref{def:Dt}, since
\[D(t)=D(t,Y(t)),\quad W(t)=W(t,X(t)).\]
To estimate $\partial_y D(t,y)$, we notice that the definition of $D(t,y)$ (see \eqref{def:Dt}) implies $\bar T(D(t,y),t)=y$, and then
\[\partial_y D(t,y)=\frac 1{\partial_s \bar T(D(t,y),t)}.\]
Similarly, to compute $\partial_x W(t,x)$, we notice that the definition of $W(t,x)$ (see \eqref{Def:Wt}) implies
\[
(1+\alpha x)W(t,x)+\beta \int_{0}^{W(t,x)} \bar T(t,\tau)_+\,d\tau = \bar P(t).
\]
Differentiating with respect to $x$ yields:
\[
\partial_x W(t,x)=-\frac{\alpha W(t,x)}{1+\alpha x+\beta \bar T(W(t,x))_+}.
\]
Moreover, the definition of $\bar a$ and $\hat \rho$ (see \eqref{def:gt} and \eqref{def:rh}) imply
\[\bar a(t,x,y)=d+\gamma 1_{W(t,x)<D(t,y)}\int_{W(t,x)}^{D(t,y)}(\bar T(s,t))^+\,ds,\]
\[\hat \rho[\tilde x,\tilde y](x,y)=\min\big(D(t,\tilde y),D(t,y),W(t,\tilde x),W(t,x)\big).\]
We can then compute:
\[\partial_x\bar a(t,x,y)=-\gamma \partial_x W(t,x) (\bar T(W(t,x),t))^+1_{W(t,x)<D(t,y)}=\gamma \frac{\alpha W(t,x)(\bar T(W(t,x),t))^+}{1+\alpha x+\beta \bar T(W(t,x))_+}1_{W(t,x)<D(t,y)},\]
\[ \partial_y\bar a(t,x,y)=\gamma \partial_x D(t,y) (\bar T(D(t,y),t))^+1_{W(t,x)<D(t,y)}=\gamma \frac{(\bar T(D(t,y),t))^+}{\partial_s \bar T(D(t,y),t)}1_{W(t,x)<D(t,y)},\]
\begin{align*}
\partial_x\hat \rho[\tilde x,\tilde y](x,y)&=\partial_x W(t,x)1_{W(t,x)<\min\big(D(t,\tilde y),D(t,y),W(t,\tilde x)\big)}\\
&=-\frac{\alpha W(t,x)}{1+\alpha x+\beta \bar T(W(t,x))_+}1_{W(t,x)<\min\big(D(t,\tilde y),D(t,y),W(t,\tilde x)\big)},
\end{align*}
\begin{align*}
\partial_y\hat \rho[\tilde x,\tilde y](x,y)&=\partial_y D(t,y)1_{D(t,y)<\min\big(D(t,\tilde y),W(t,\tilde x),W(t,x)\big)}\\
&=\frac 1{\partial_s \bar T(D(t,y),t)}1_{D(t,y)<\min\big(D(t,\tilde y),W(t,\tilde x),W(t,x)\big)}.
\end{align*}
Finally, the mortality rate \eqref{def:mortalityrate} is given by
\[M(t)= d+\gamma 1_{W(t)<D(t)}\int_{W(t)}^{D(t)} (\bar T(t,s))^+\,ds.\]

\section{Numerical simulations of the macroscopic model}
\label{sec:numerics}

The structured model \eqref{model:m} provides a precise description of the life cycle of individuals and of the impact of the life cycle on the genetic evolution of the population. The model is however complex, even numerically: the evaluation of the birth term, that involves five integrals, makes the development of numerical simulations very challenging. The intermediate model \eqref{model:n} has similar integral terms and its simulation is also difficult. The macroscopic model \eqref{eq:dXY} consists of a system of two coupled differential equations in $\mathbb R$, and the dynamics of the model is then given by a vector field in $\mathbb R^2$. This vector field can be computed numerically thanks to the formula given in Section~\ref{subsec:coeff}, and we simulate the system \eqref{eq:dXY} using a fourth-order Runge–Kutta method.

\subsection{Dynamics of the population in a fixed environment}
\label{subsec:num-fixed}

\begin{figure}[H]
    \centering
    \begin{subfigure}{0.5\textwidth}
         \includegraphics[width=\textwidth]{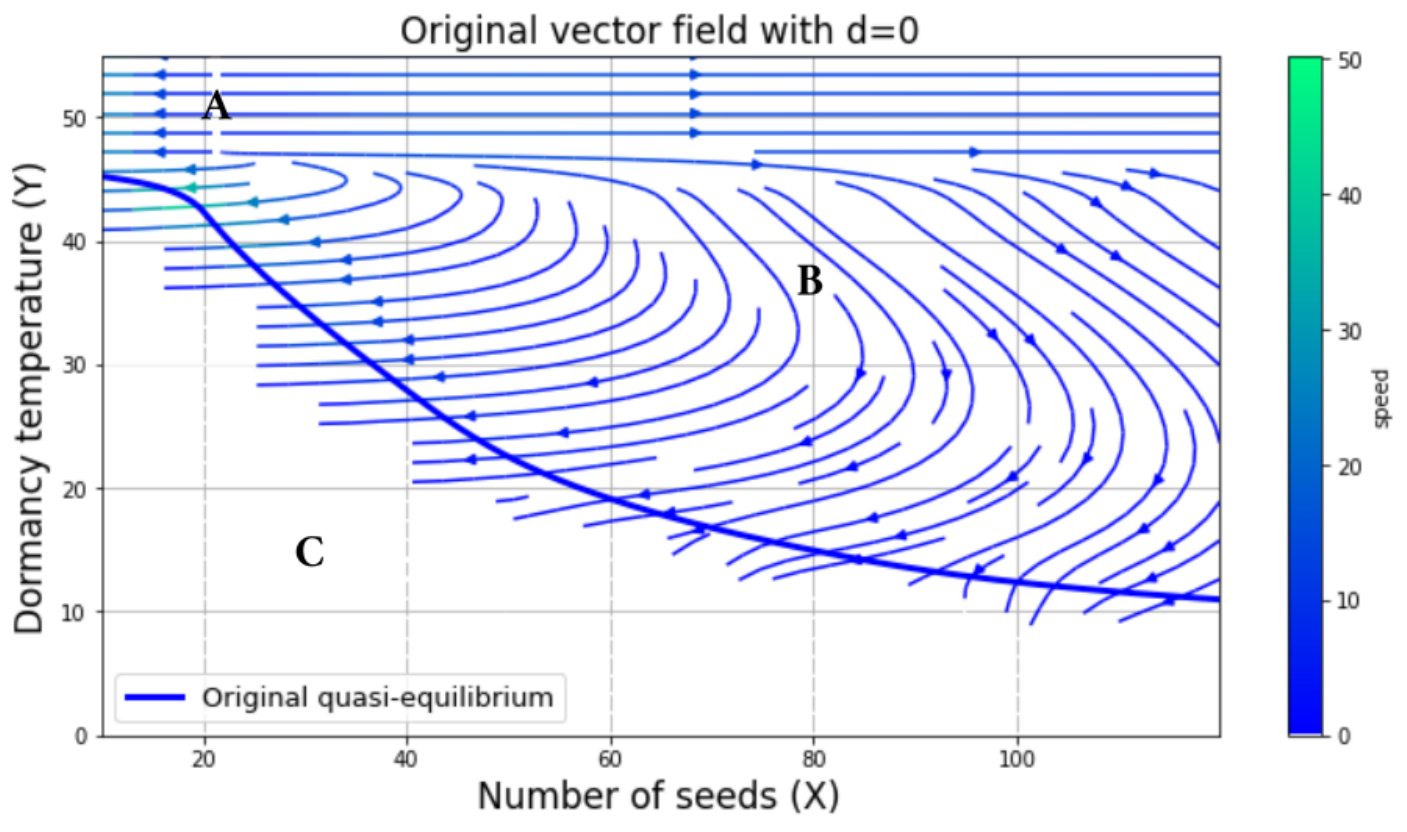}
         \caption{}
    \end{subfigure} \hfill
    \begin{subfigure}{0.5\textwidth}
         \includegraphics[width=\textwidth]{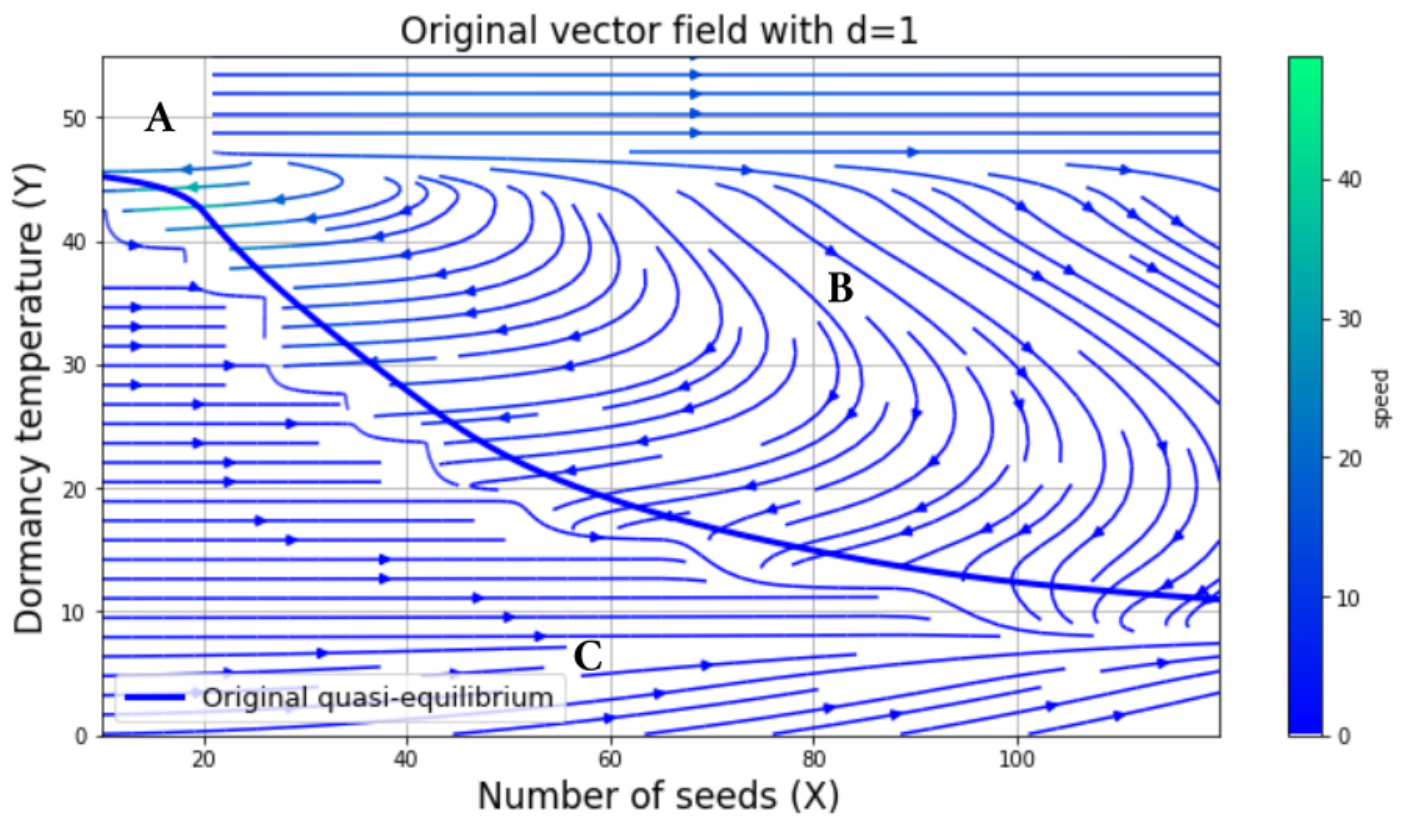}
         \caption{}
    \end{subfigure}
    \caption{Vector field for the differential equation \eqref{eq:dXY} on $(X(t),Y(t))$. The three graphs corresponds to different values of the base death rate $d\geq 0$: (a): case where the $d=0$, (b): case where $d>0$ is positive. The blue line represents the \textbf{CESS}, the letters $\bf A$, $\bf B$, $\bf C$ represent  the three regions described in Section~\ref{subsec:num-fixed}, and the blue arrows represent the vector field associated to the differential equation \eqref{eq:dXY}. 
    \textit{Parameters:} $\alpha=3, \beta=1, \gamma=2, s_0=0, \sigma_{x}=50, \sigma_y=1 \text{ and } d \in \{0,1,2\}$.}
    \label{fig:2}
\end{figure}

From \eqref{def:gt}, we observe that if $d=0$ (that is when deaths only results from water stress), $\bar a(t,x,y)$ is given by
\[\bar a(t,X(t),Y(t))=\int_0^1 1_{[0,D(t)]}(s)\gamma(\bar T(t,s))^+1_{[W(t),1]}(s)\,ds,\]
where we have used the notation $D(t)$, $W(t)$ introduced in \eqref{def:Dt}-\eqref{Def:Wt}. In particular, $\bar a(t,X(t),Y(t))=0$ if $D(t)<W(t)$. This defines a region of the phase plane $(X(t),Y(t))$ where no evolution occurs when $d=0$, that we denote by $\bf C$ in Figure~\ref{fig:2}(a). The boundary of this region is constituted of points where the evolution of the population halts, which we denote by \textbf{CESS}, for \emph{continuum of evolution singular strategies}, referring to the singular strategies of the theory of Adaptive Dynamics, which are fixed traits for the evolution dynamics. The area above the \textbf{CESS} can be divided further into two regions. The first one is denoted by $\bf A $ in Figure~\ref{fig:2}(a): it corresponds to the area where $Y$ is larger than the maximal temperature of the environment: since $Y(t)$ is the trigger temperature for dormancy, no dormancy occurs if $(X(t),Y(t))$ is in region $\bf A$, and no selection pressure is then exerted on $Y(t)$, that therefore does not evolve. We notice that region $\bf A$, the trait $X(t)$ increases indefinitely, which can be explained by the fact that individuals unable to enter dormancy inevitably encounter water stress, it is then beneficial to produce many seeds early in the season, that is to increase $X(t)$, to produce as many seeds as possible until the water resource is depleted. Finally, in region $\bf B$, both traits $X(t)$ and $Y(t)$ evolve, and they converge towards the \textbf{CESS}.

\begin{table}[h]
\centering
{
\begin{tabular}{|c|p{14cm}|}
\hline
\textbf{CESS} & Continuum of singular strategies: no evolution occurs when $(X,Y)$ is on this line  \\ \hline
$\bf{A}$ & The trees do not enter dormancy at all.  \\ \hline
$\bf{B}$  & Trees enter dormancy too late, and they face water stress, inducing some deaths and a rapid evolution of both $X$ and $Y$ as a result.         \\ \hline
$\bf{C}$  &  The trees are protected from water stress by dormancy. The mortality rate is then low and the evolution of $X$, $Y$ is slow.       \\ \hline
\end{tabular}}
\caption{{Ecological regimes discussed in Sections~\ref{subsec:num-fixed} and represented in Figures~\ref{fig:2}.}}
\label{tab:example}
\end{table}

In Figure~\ref{fig:2}(b), we consider a positive, but small, base death rate $d>0$, and a larger value of this parameter in Figure~\ref{fig:2}(c). We notice that the three regions $\bf A$, $\bf B$, $\bf C$ described above still make sense qualitatively when $d>0$, the main difference being that the traits in region $\bf C$ converge to the \textbf{CESS} if $d>0$, but this convergence is very slow if $d>0$ is small.

\subsection{Effect of an environmental shift}\label{subsec:shift}

We now consider a population that is well adapted to its initial environment, that is a population such that $(X(0),Y(0))$ lies on the \textbf{CESS} corresponding to the initial environment, represented by a thick blue line in Figures~\ref{fig:3}-\ref{fig:4}-\ref{fig:5}-\ref{fig:6}. We therefore assume that the population is in a stabilized situation initially, that is until time $t=0$. We consider that right after time $=0$ which we denote by $t=0^+$, there is a rapid shift of the environmental conditions, leading to a new vector field and a new \textbf{CESS},  represented in red in figures. We assume that this shift is rapid because the derivation of the macroscopic model we consider (see \eqref{eq:dXY}) involves a change of the time variable: any environmental shift that occurs on the time scale of a tree lifetime (or over an even shorter time) would appear instantaneous in the model we simulate here.  We consider three types of environmental shifts. We represent the dynamics through the phase-plan $(X(t),Y(t))$ in Figures~\ref{fig:4} to \ref{fig:6}. For some of these dynamics, we provide more details of the population's dynamics in  Figure~\ref{fig:7}, were we  represent the time dynamics of the traits $(X(t),Y(t))$ as well as the onset of summer dormancy $D(t)$ (see \eqref{def:Dt}) and the beginning of the water stress $W(t)$ (see \eqref{Def:Wt}). The traits $D(t)$ and $W(t)$ depend on both $(X(t),Y(t))$ and the environmental conditions; they can therefore be seen as plastic phenological traits of the individuals, which could be more easily estimated on real populations than $X(t)$ or $Y(t)$, and could be useful to investigate the connection between the model we built and field studies \textcolor{blue}{(see e.g. \cite{richards2020quantitative})}.

\medskip

    \noindent\textbf{Temperature shift:} We consider a situation where the temperatures increase, while the precipitation levels remain constant, see Figure~\ref{fig:3} and Figure~\ref{fig:7}(a). We observe that the original \textbf{CESS} (in blue) is below the new one (in red), so that $(X(0),Y(0)$ (that is on the original \textbf{CESS}) lies in the region $\bf C$ for the new environment (see the description of region $\bf C$ in section~\ref{subsec:num-fixed}). Therefore, the population adapts slowly to the new environment if $d>0$ is small (see Figure~\ref{fig:3}(b)), and no adaptation occurs if $d=0$  (see Figure~\ref{fig:3}(a)). To understand the biological meaning of this slow adaptation, we should remember that dormancy is triggered by the critical temperature $Y(t)$. The shift to higher temperatures then implies an earlier entry into dormancy, protecting individuals from the effect of the increasing temperatures. This protection implies that few individuals die, which leads in turn to a slow adaptation of the population. In Figure~\ref{fig:7}(a), we provide more details on the dynamics represented in Figure~\ref{fig:3}(a). We notice that the temperature shift implies an earlier onset of the summer dormancy $D(t)$, which protects the population: the shift does not imply a surge in mortality. As explain above, no evolution occurs: $X(t)$ and $Y(t)$ are constant.

\medskip

\noindent\textbf{Precipitation shift:} We consider a drop in the precipitation levels, while temperatures remain constant, see Figures~\ref{fig:4}-\ref{fig:5} and Figure\ref{fig:7}(b)-(c). We observe that the original \textbf{CESS} (in blue) is above the new one (in red), indicating an important water stress for individuals, and the important mortality it creates enables a fast evolution of the population. In Figure~\ref{fig:4}, we considered a situation where only $X(t)$ can evolve, which we obtain by considering $\bar \sigma_y=0$ in \eqref{eq:dXY}. We then observe two qualitatively different consequences of the precipitation shift, depending on its amplitude :
\begin{itemize}
    \item If the precipitation shift is small to moderate (Figure~\ref{fig:4}(a)), the trait $X(t)$ of the population adjusts rapidly to the new \textbf{CESS}.
    \item If the precipitation shift is large (Figure~\ref{fig:4}(b), the trait $X(t)$, that is the rate of seed production,  evolves to $+\infty$. This corresponds to the dynamics described in Section~\ref{subsec:num-fixed} for region $\bf A$: the population changes produces many seeds early in the season, accepting the effect of water stress once the seeds are produced.
\end{itemize}
The dependency of the dynamics of $X(t)$ in the amplitude of the drop in precipitations constitutes a tipping point: If the shift is small or moderate, the population is able to adapt its trait to the new environment, to produce a maximum number of seeds without facing water stress; if the shit is important, evolution brings the population to a different strategy, consisting in the production of many seeds early in the season. A large shift in precipitation can then lead the population to adopt a widely different behavior. The population represented in Figure~\ref{fig:4}(b) originally produces seed continuously over the season; after the shifts, it adopts a different niche, with $X(t)\gg 1$ corresponding to the production of all seeds at once early in the season. These two reproduction strategies can be observed in natural populations. In  \cite{neuffer1999colonization}, the local adaptation of the plant Capsella bursa-pastoris is described; plants adapted to cool and snowy regions the plant produces seeds over an extended period, while plants originating from warm and dry regions produce all seeds at once early in the season.

If both $X(t)$ and $Y(t)$ can evolve, the dynamics depicted in Figure~\ref{fig:4}(a) remains almost the same. The dynamics shown by Figure~\ref{fig:4}(b) also remains very similar when $\bar\sigma_y>0$ is small, although $Y(t)$ evolves slowly to a lower value, bringing $(X(t),Y(t))$ to the \textbf{CESS} corresponding to the new environment, as represented in Figure~\ref{fig:5}(a): we see that $X(t)$ increases  following the large precipitation shift, but after some time, the evolution of $Y(t)$ brings the traits of the population back to the \textbf{CESS}. Note that the final value of $X(t)$ significantly higher than its original value: this dynamics is coherent with the dynamics of Figure~\ref{fig:4}(b). If the coefficient $\bar \sigma_y>0$ is large enough (see Figure~\ref{fig:5}(b)), the rapid dynamics of $Y(t)$ seems to remove the tipping point described above. The dynamics of the population represented in Figure~\ref{fig:5}(a)-(b) is described further in Figure~\ref{fig:7}(b)-(c). We notice that the shift in precipitation levels induces an important water stress, with $W(t)$ dropping suddenly. This water stress implies an important mortality, and the population evolves as described above. 

\medskip

\noindent\textbf{Temperature \& Precipitation shift:} It is also possible to consider mixed senario, where the environmental shift combines an elevation of the temperatures with a drop in precipitations, see Figure~\ref{fig:6}. We observe that the original \textbf{CESS} (in blue) typically crosses the new one (in red). The dynamics of the population after the shift then depends on the initial position of the traits of the population along the original \textbf{CESS} (in blue): if the population has a trait $X(t)$ sufficiently small, it will suffer from mortality after the environmental shift, which will induce a rapid adaptation. If the trait $X(t)$ is large, however, the population will be protected from water stress by the dormancy, and will not adapt rapidly (if $d=0$ as it is the case in Figure~\ref{fig:6}, there is no mortality and therefore no adaptation). The dynamics of the population represented in Figure~\ref{fig:6}(c) is described further in Figure~\ref{fig:7}(d).

\begin{figure}[H]
    \centering
     \centering
     \begin{subfigure}{0.47\textwidth}
         \includegraphics[width=\textwidth]{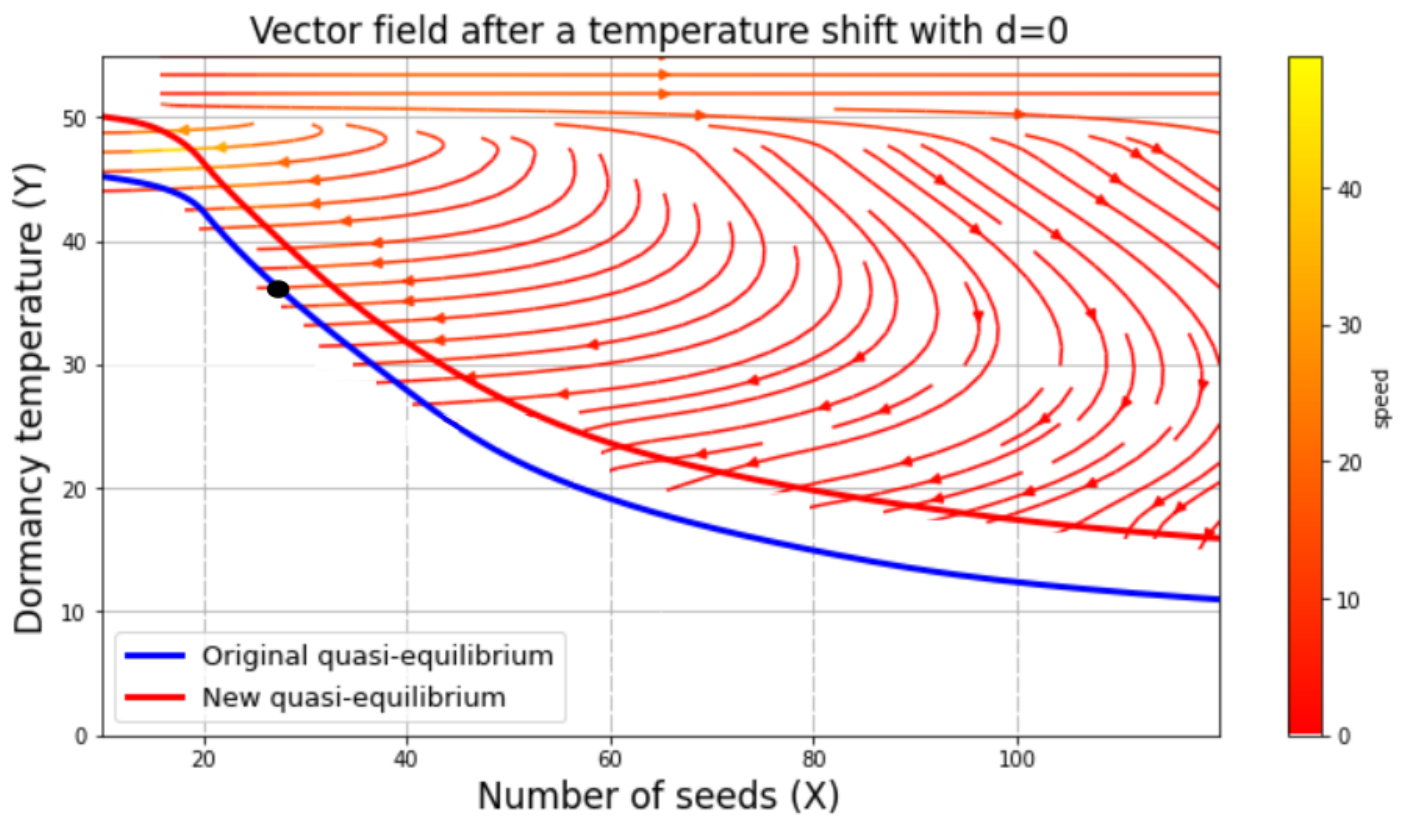}
         \caption{}
    \end{subfigure}\hfill
    \begin{subfigure}{0.47\textwidth}
         \includegraphics[width=\textwidth]{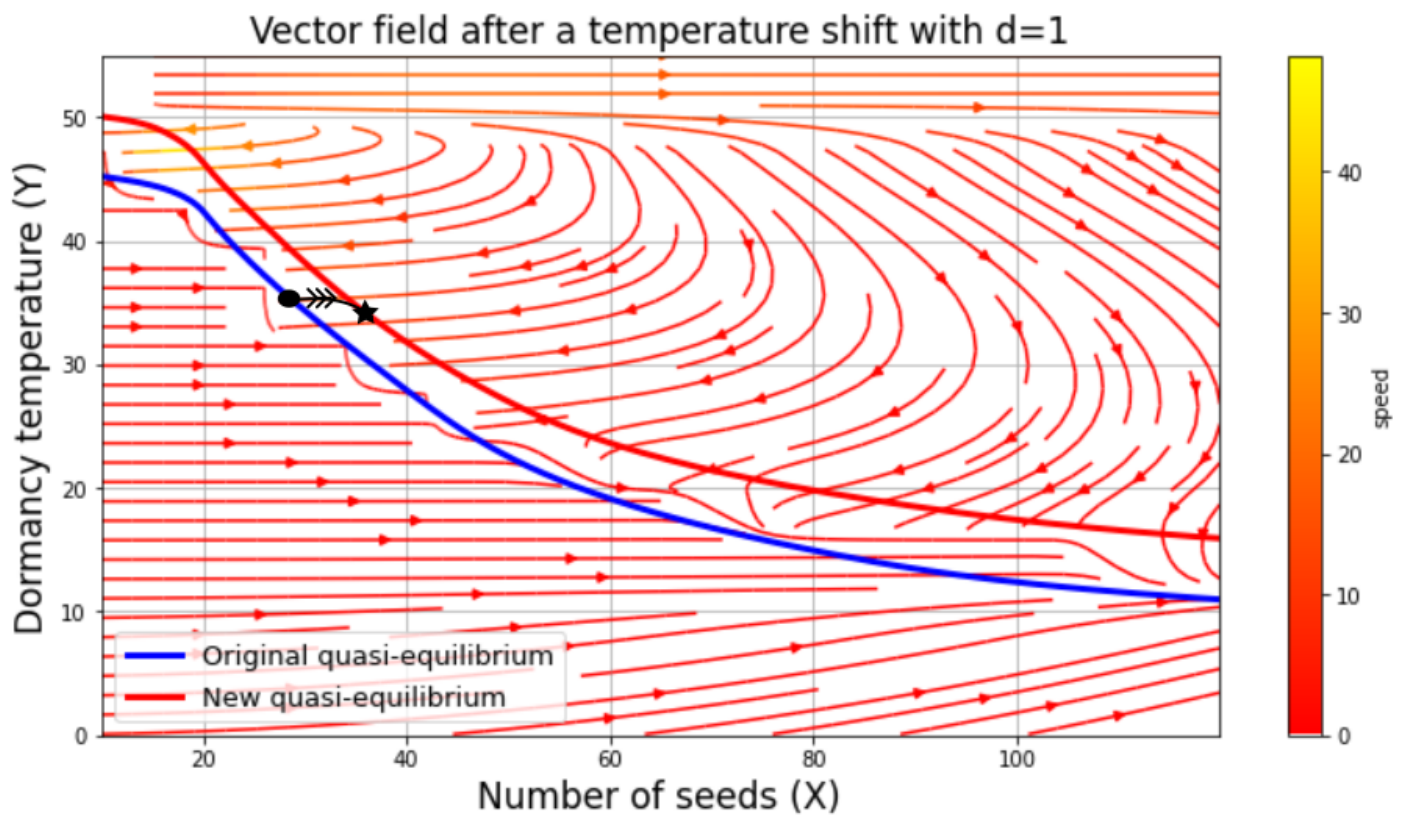}
         \caption{}
    \end{subfigure}
    \caption{Effect of an elevation of the temperatures: (a) when the base death rate is $d=0$; (b) when the base rate is positive but small $d>0$. The blue line represents the \textbf{CESS} before the environmental shift, while the red line and vectors represent the new \textbf{CESS} and the vector field of the differential equation~\eqref{eq:dXY} after the shift. The population trajectory is represented in black where the dot and the star represent initial and final positions respectively.
    \\
    \textit{Parameters:} $\alpha=3, \beta=1, \gamma=2, s_0=0, \sigma_{x}=50, \sigma_y=1 \text{ and } d \in \{0,1\}$.}
    \label{fig:3}
\end{figure}

\begin{figure}[H]
    \centering
     \begin{subfigure}{0.47\textwidth}
         \includegraphics[width=\textwidth]{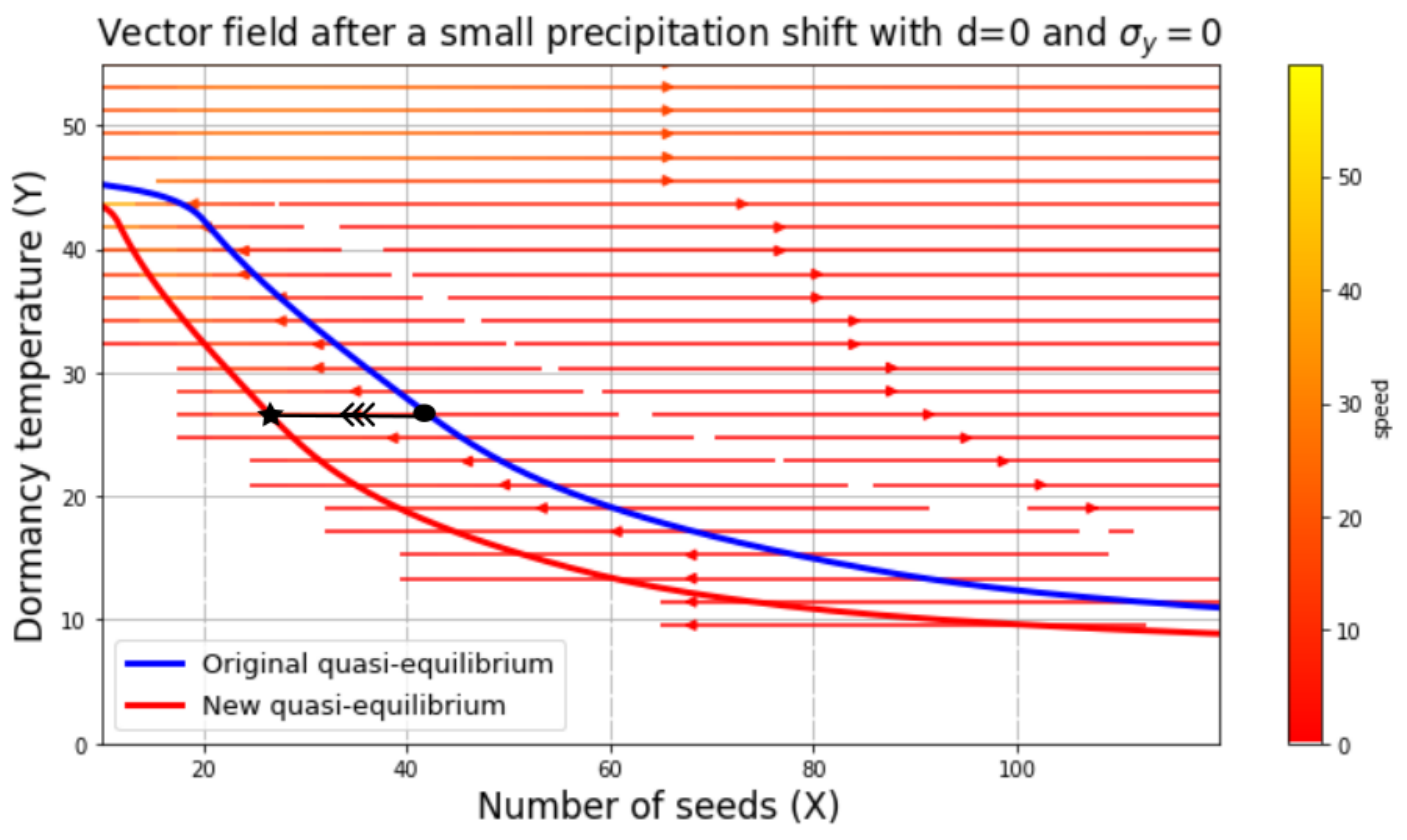}\caption{}
    \end{subfigure}\hfill
    \begin{subfigure}{0.47\textwidth}
         \includegraphics[width=\textwidth]{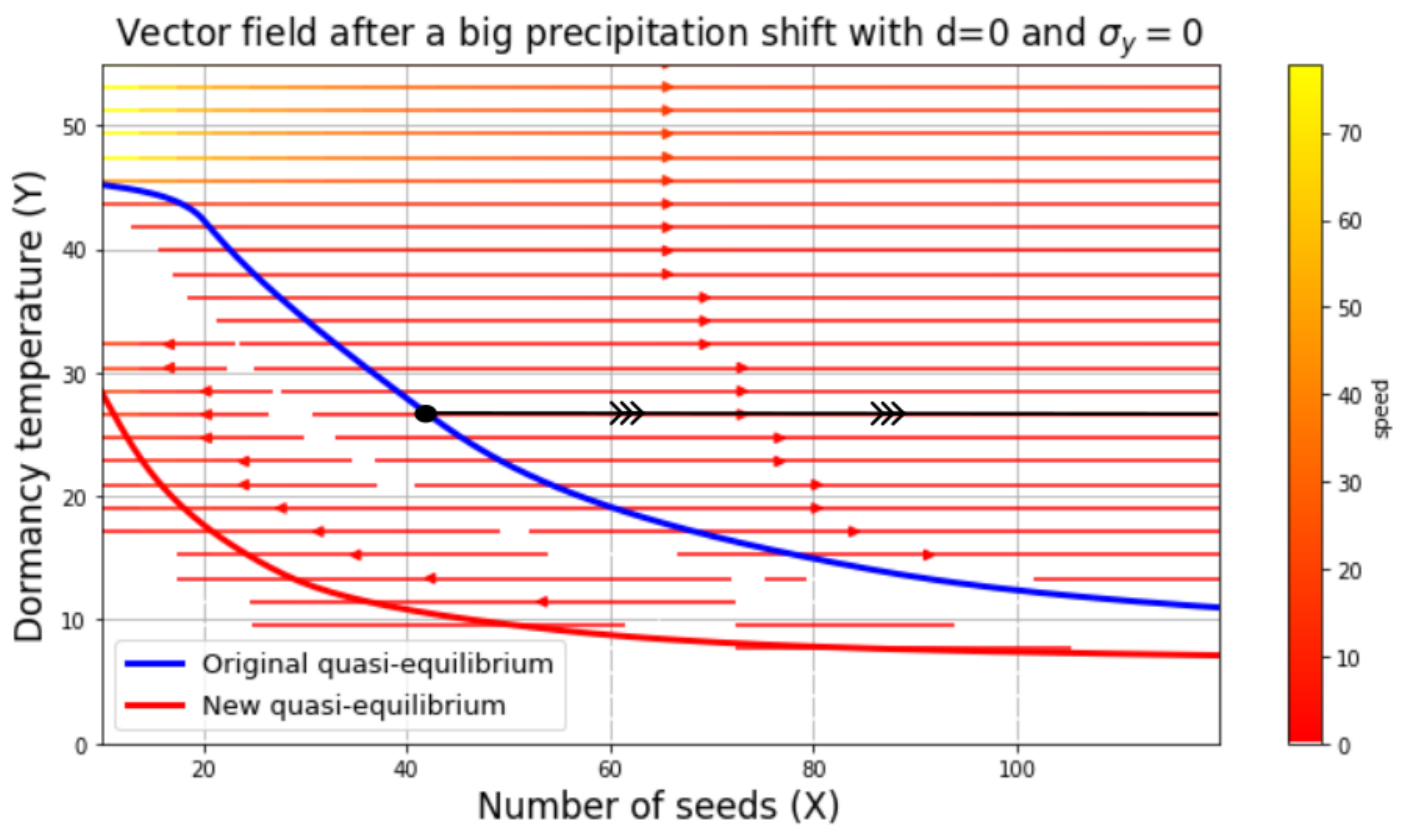}\caption{}
    \end{subfigure}
    \caption{Effect of a drop in precipitation levels, when only $X(t)$ evolves (that is when $\bar \sigma_y=0$): (a) effect of a small to moderate drop in precipitation levels (when $d=0$); (b) effect of a large drop in precipitation levels (when $d=0$).
     The blue line represents the \textbf{CESS} before the environmental shift, while the blue line and vectors represent the new \textbf{CESS} and the vector field of the differential equation~\eqref{eq:dXY} after the shift. The population trajectory is represented in black  where the dot and star represent initial and final positions respectively.
     \\
    \textit{Parameters:} $\alpha=3, \beta=1, \gamma=2, s_0=0, \sigma_{x}=50, \sigma_y=0$.}
    \label{fig:4}
\end{figure}

\begin{figure}[H]
    \centering
     \begin{subfigure}{0.45\textwidth}
         \includegraphics[width=\textwidth]{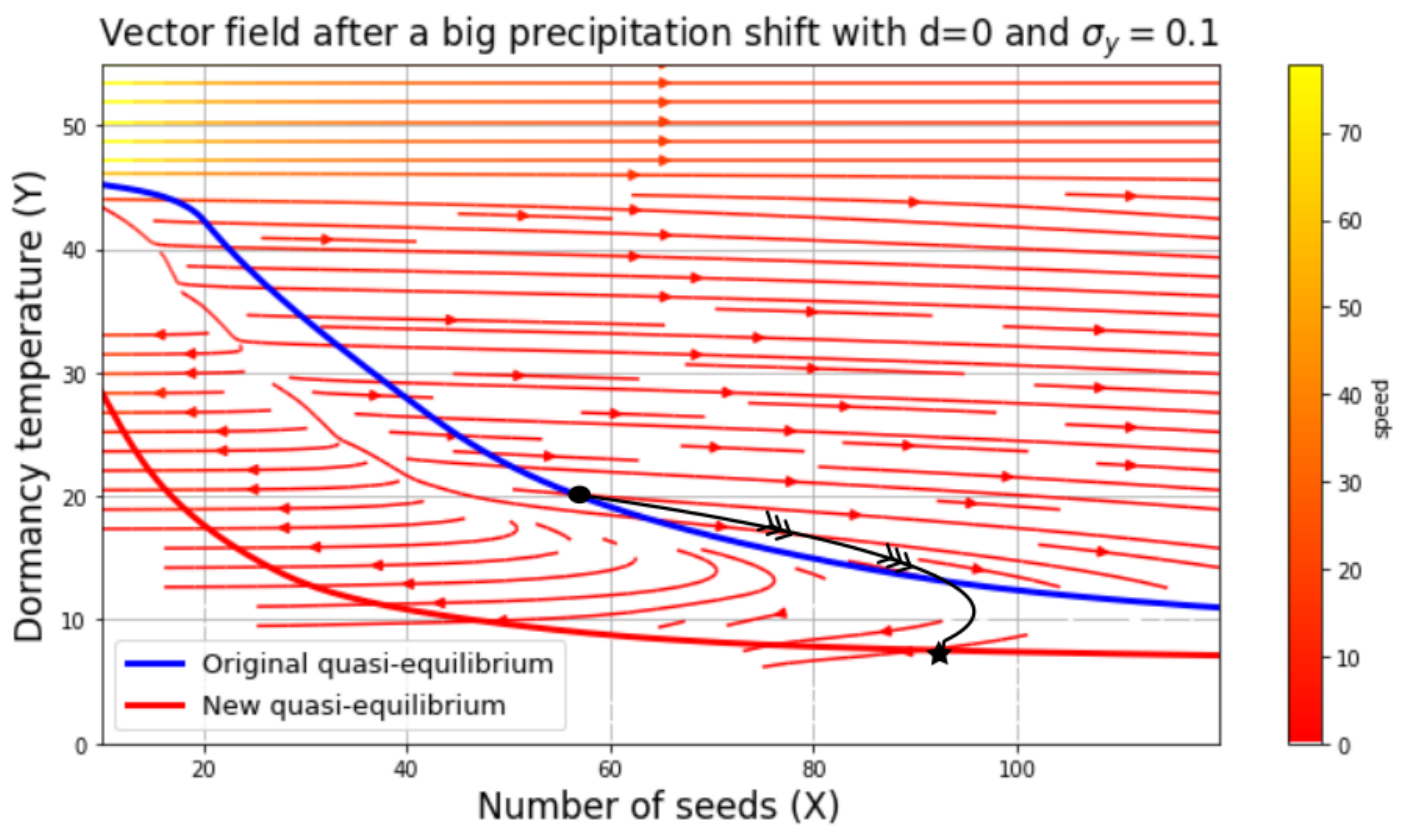}\caption{}
    \end{subfigure}\hfill
    \begin{subfigure}{0.45\textwidth}
         \includegraphics[width=\textwidth]{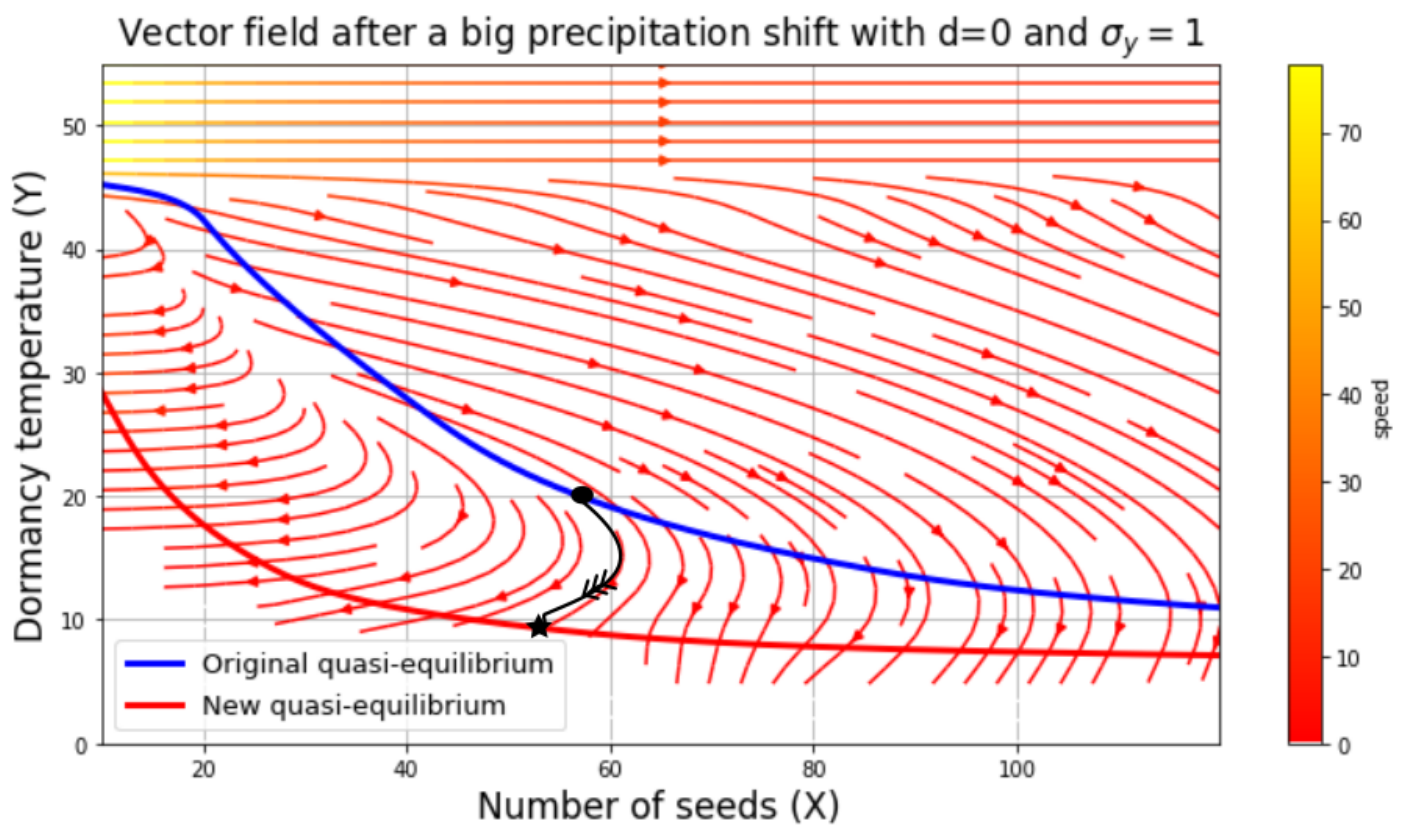}\caption{}
    \end{subfigure}
    \caption{Effect of a drop in precipitation levels, when both $X(t)$ and $Y(t)$ evolve: (a) case where $\bar \sigma_y>0$ is small; (b) case where $\bar\sigma_y>0$ is large.  The blue line represents the \textbf{CESS} before the environmental shift, while the blue line and vectors represent the new \textbf{CESS} and the vector field of the differential equation~\eqref{eq:dXY} after the shift. The population trajectory is represented in black  where the dot and  the star represent initial and final positions respectively..    
     \\
    \textit{Parameters:} $\alpha=3, \beta=1, \gamma=2, s_0=0, \sigma_{x}=50, \sigma_y \in \{0.1,1\} \text{ and } d=0$.}
    \label{fig:5}
\end{figure}

\begin{figure}[H]
    \centering
     \begin{subfigure}{0.45\textwidth}
         \includegraphics[width=\textwidth]{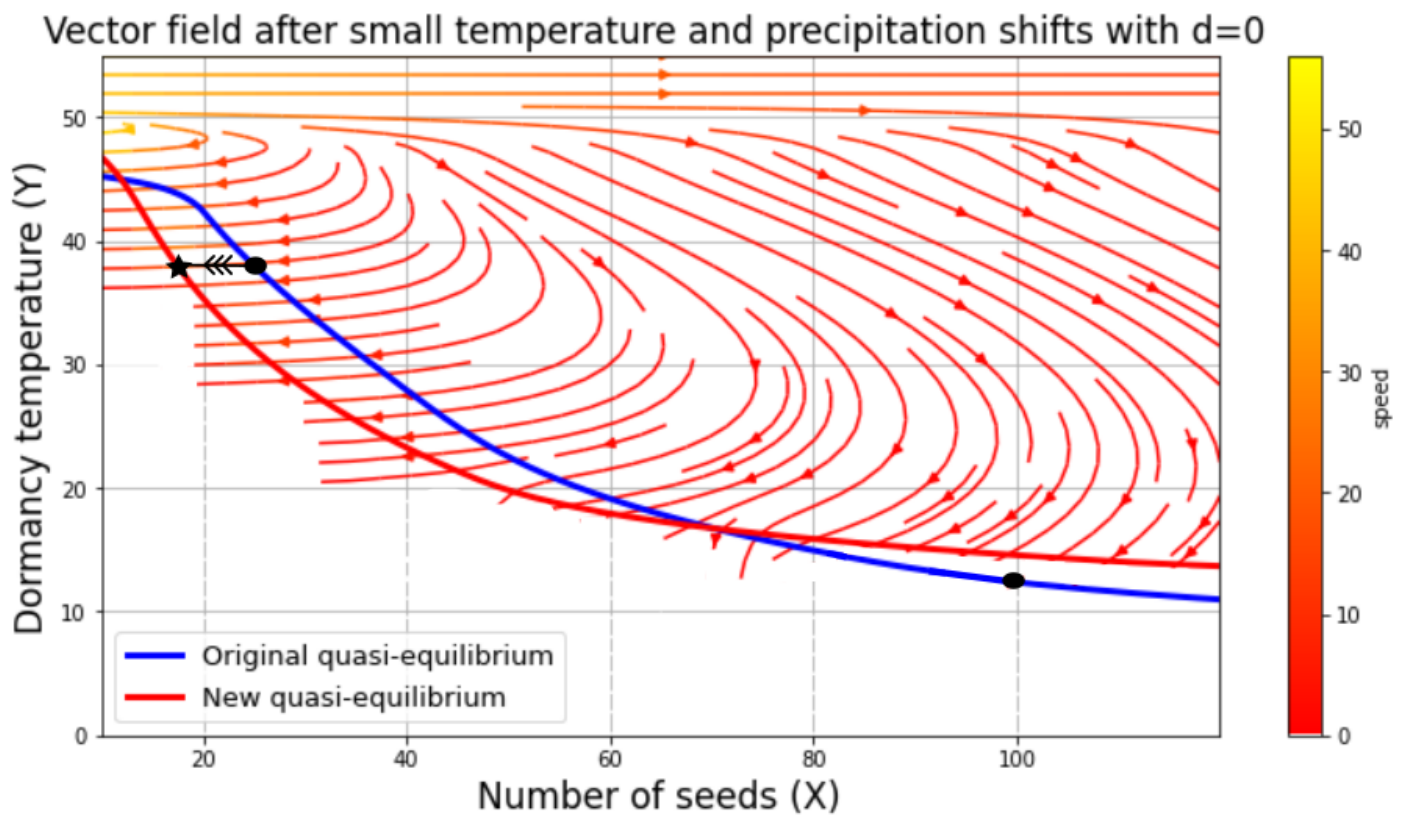}\caption{}
    \end{subfigure}\hfill
     \begin{subfigure}{0.45\textwidth}
         \includegraphics[width=\textwidth]{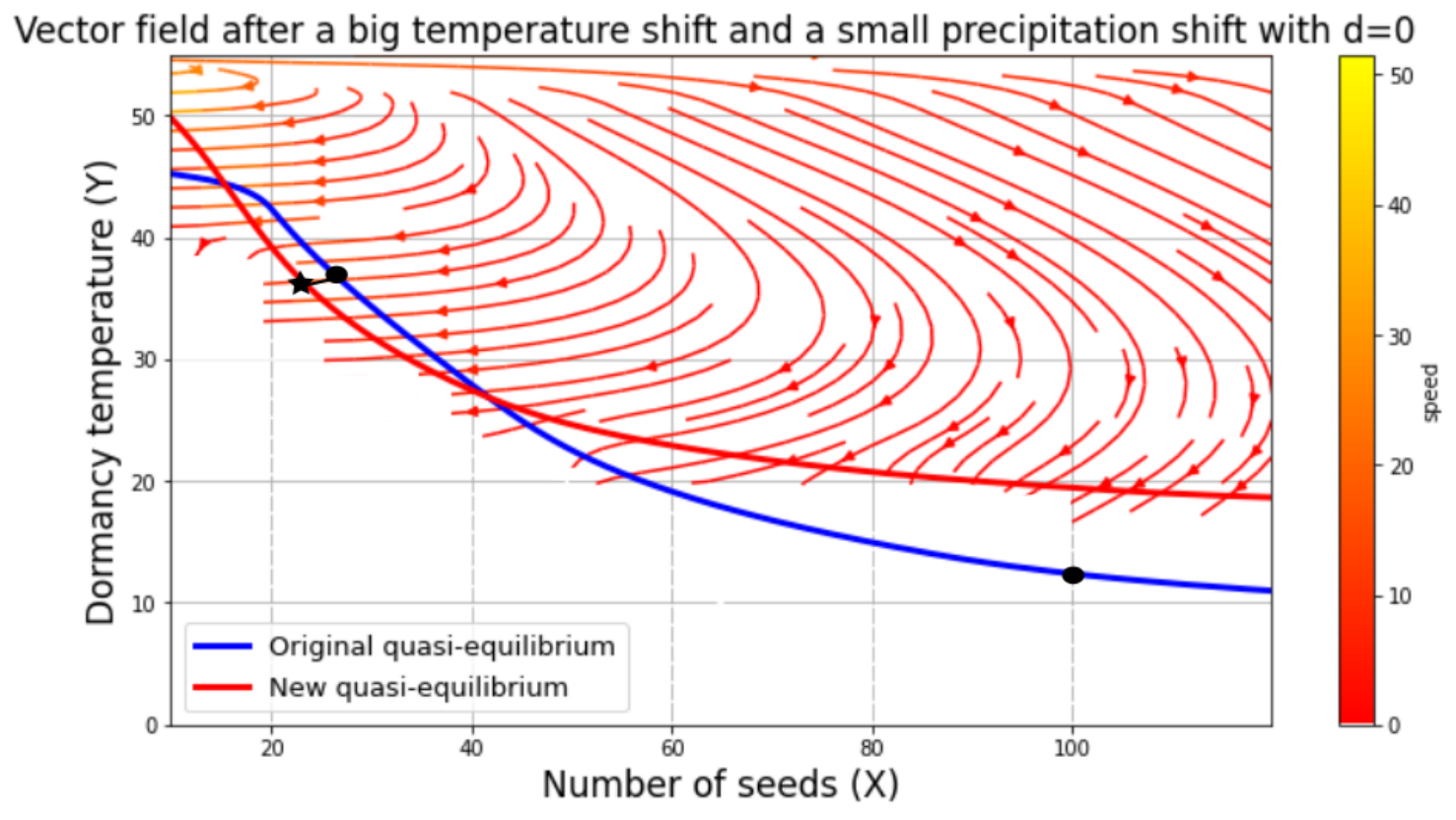}\caption{}
    \end{subfigure}\hfill
     \begin{subfigure}{0.45\textwidth}
         \includegraphics[width=\textwidth]{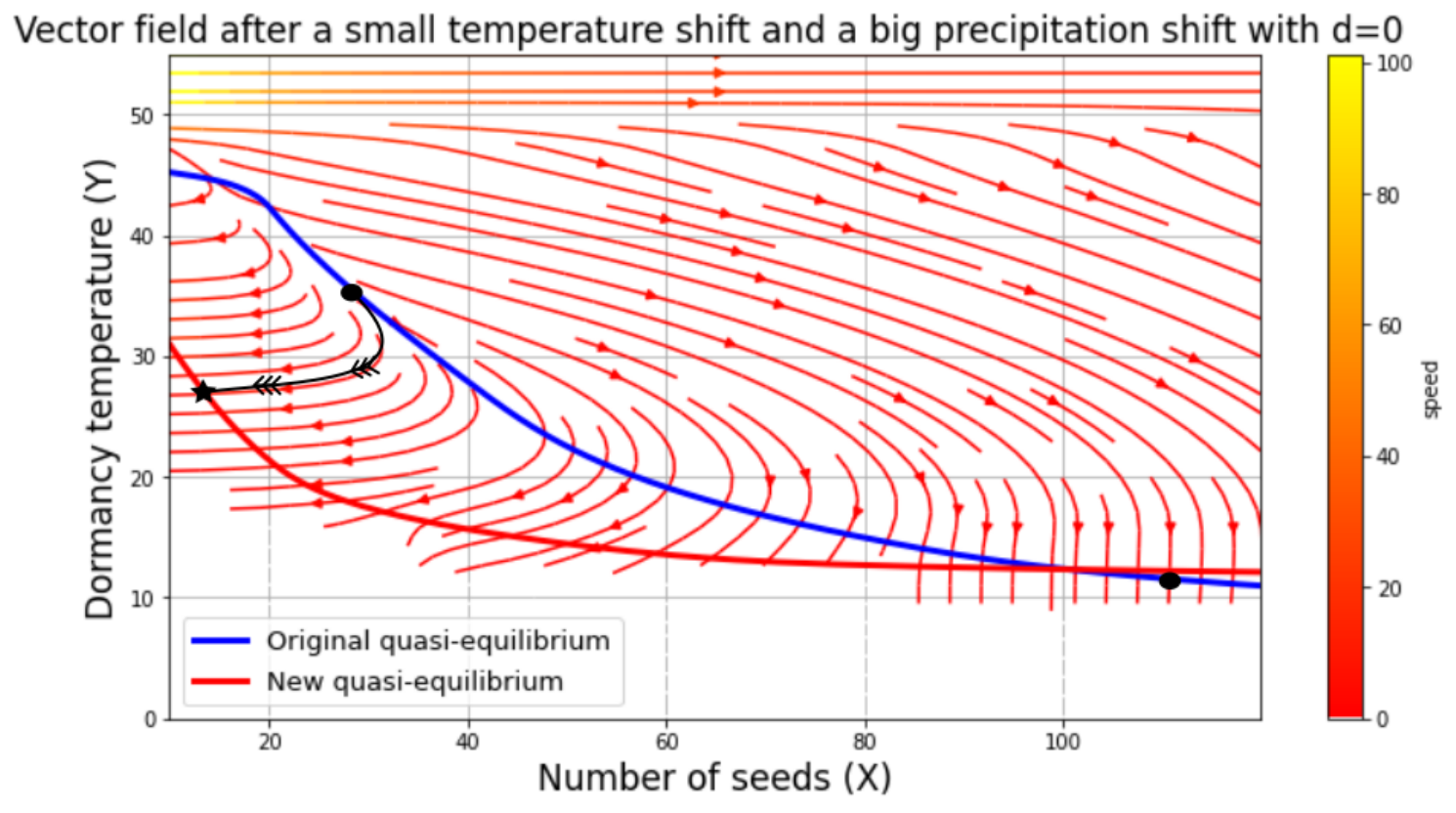}\caption{}
    \end{subfigure}\hfill
     \begin{subfigure}{0.45\textwidth}
         \includegraphics[width=\textwidth]{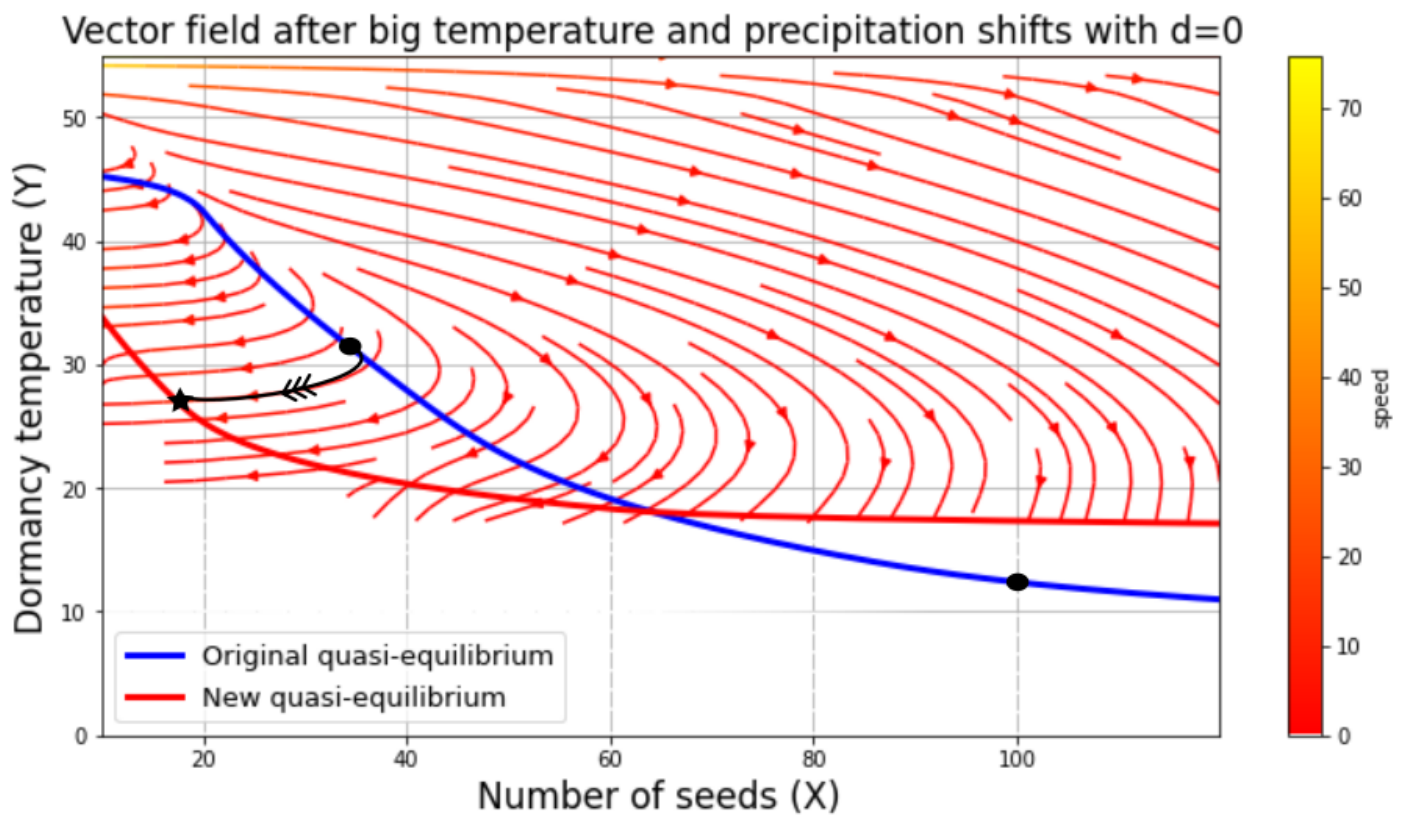}\caption{}
    \end{subfigure}\hfill
    \caption{Effect of environmental change that combines a temperature elevation and a drop in precipitation levels, with a base mortality $d=0$. (a) Case where both the drop in precipitation levels and the elevation of temperatures are small; (b) Case where the drop in precipitation levels is small, but the temperature shift is large; (c) Case where the drop in precipitation levels is large, but the elevation in temperatures is small; (d) Case where both the drop in precipitation levels and the elevation of temperatures are large.  The blue line represents the \textbf{CESS} before the environmental shift, while the blue line and vectors represent the new \textbf{CESS} and the vector field of the differential equation~\eqref{eq:dXY} after the shift. The population trajectory is represented in black.    
      \\
    \textit{Parameters:} $\alpha=3, \beta=1, \gamma=2, s_0=0, \sigma_{x}=50, \sigma_y=1 \text{ and } d=0$.}
    \label{fig:6}
\end{figure}

\begin{figure}[H]
    \centering
     \begin{subfigure}{1\textwidth}
         \includegraphics[width=\textwidth]{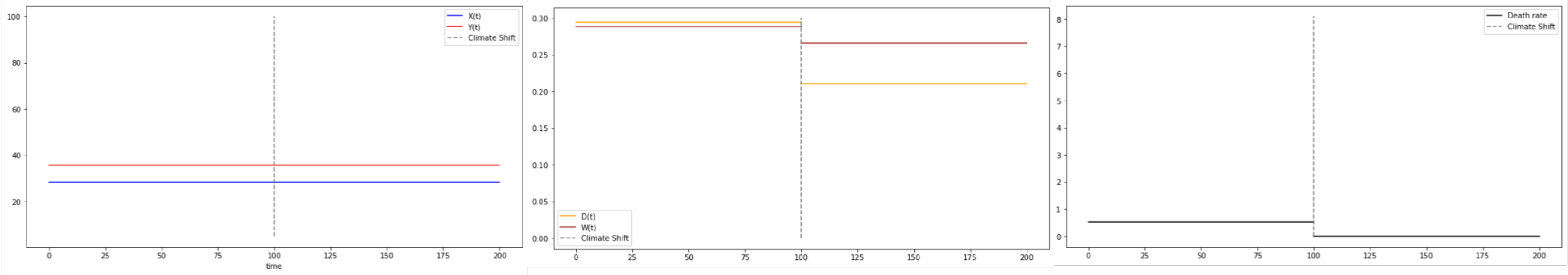}\caption{}
    \end{subfigure}\hfill
     \centering
     \begin{subfigure}{1\textwidth}
         \includegraphics[width=\textwidth]{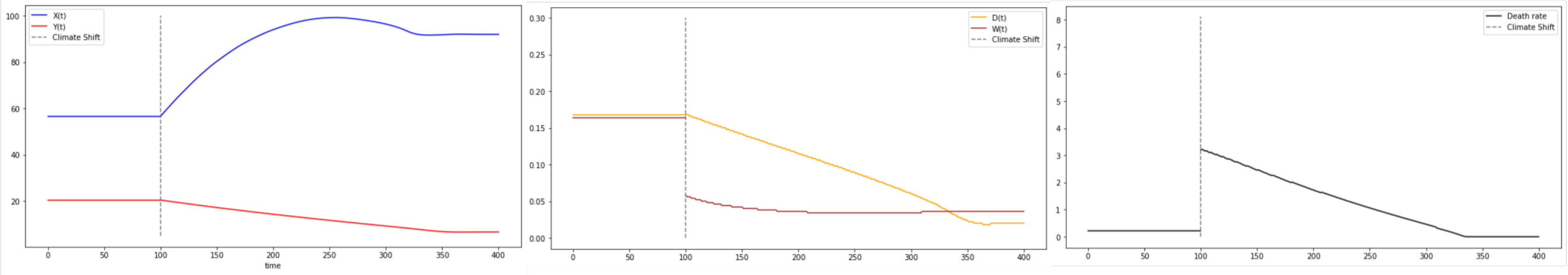}\caption{}
    \end{subfigure}\hfill
     \centering
     \begin{subfigure}{1\textwidth}
         \includegraphics[width=\textwidth]{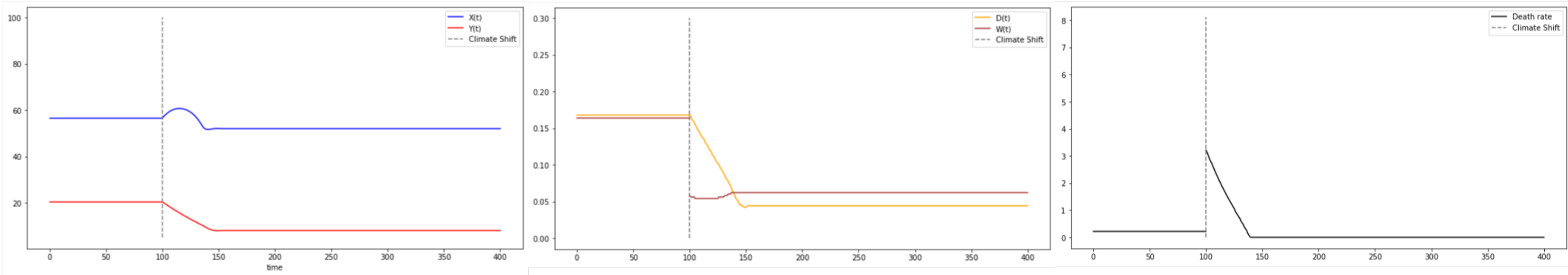}\caption{}
    \end{subfigure}\hfill
     \centering
     \begin{subfigure}{1\textwidth}
         \includegraphics[width=\textwidth]{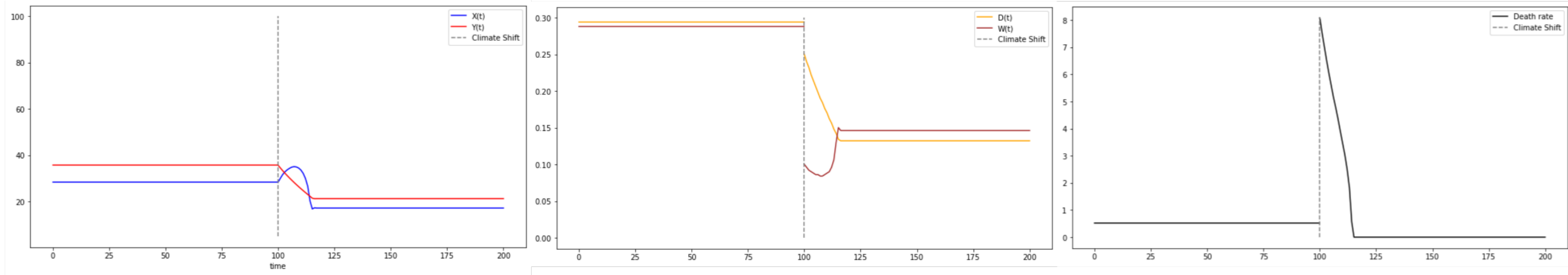}\caption{}
    \end{subfigure}\hfill
     \centering
    \caption{Temporal evolution of the phenotypic traits $X(t)$ and $Y(t)$ before and after climate shifts (left), the seasonal time when individuals become dormant $D(t)$ and the onset of water stress $W(t)$, as defined by \eqref{def:Dt} and \eqref{Def:Wt} (center), and the mortality rate $M(t)$ (right). (a) corresponds to Figure~\ref{fig:3}(a);(b) corresponds to Figure~\ref{fig:5}(a); (c) corresponds Figure~\ref{fig:5}(b);  (d) corresponds to Figure~\ref{fig:6}(c).}
    \label{fig:7}
\end{figure}

\section{Discussion}

In this manuscript, we have introduced a model for a tree population structured by two phenotypic traits: a seed production trait and a dormancy trait. Those two traits are breeding values, ie they are fully inherited, and we assume that the population is structured by these traits. Two phenological traits, namely the onset of dormancy and the onset of water stress, result from the breeding value we consider and the environmental conditions, and these phenological traits are then plastic. Furthermore, in some studies, the activity status of an individual, ie active or dormant, can itself be considered as discrete phenotypic trait (\cite{reid2022properties}). Our model can therefore be used to investigate the role of adaptative and plastic evolution in the context of a climate shift. Our model suggest that the plasticity of the onset of dormancy can lower the mortality of the population after an environmental shift, but it can also slow down the adaptation of the population. The phenology of the trees impacts their production of seeds (the production stops when they become dormant), as well as the phenotypes of the seeds produced by the population (dormant trees do not produce seeds nor pollen). We have tried to consider environmental factors and phenological traits that are commonly considered in field studies, to facilitate connection with other approaches in continuations of this work. {For instance, in \cite{richards2020quantitative}, a field study on the phenological phenotypes associated to winter dormancy was conducted, showing how measures of these phenotypes can be used to understand quantitatively their heritability and how they are influenced by meteorological conditions. }  \\

We have introduced two successive asymptotic limits to simplify our complex initial model (see e.g. \cite{auger2008structured} for related structured population models). The first one is based on the assumption that the yearly death rate of adult trees is low, we believe this is very reasonable for many tree species (\cite{das2016trees}). Note that to use this asymptotic limit, we have assumed that the weather conditions change slowly, but homogenization methods could provide a way to consider more realistic climate conditions. Typically, climate change could induce precipitation levels that are less stable from one year to another. In the amazonian forests, for instance, droughts lasting several years or abnormally intense precipitations are already threatening tree populations (\cite{flores2024critical,guyane}).  The second asymptotic limit we have used is based on the assumption that at the phenotypic variance of the population is small, which is similar to to the weak selection assumption (\cite{wakeley2005limits}) commonly used in population genetics, and it is related to assumption that the phenotypic variance of a population is constant (\cite{kirkpatrick1997evolution}). This limit is a powerful approach to obtain simple macroscopic models. This assumption is well established in the evolutionary biology community, and the mathematical aspects of this asymptotic limit are rapidly strengthening (\cite{raoul2017macroscopic,calvez2023uniform}). We believe specific models such as the one considered here could bring interesting new mathematical questions. It would be interesting to understand how macroscopic models should be written to take into account for the complex processes that are relevant for climate change problems.\\

Coming back to the case of amazonian forests, the lifespan of canopy trees is of the order of hundreds of years (\cite{laurance2004inferred}), so that climate change is a rapid phenomenon on that time scale. This means it would be pertinent to consider a rapid shift of the environmental conditions for our continuous structured population model. In this context, it is not clear that the second asymptotic limit we have used is biologically relevant. That second limit relies on the assumption that the phenotypic variance of the population is small. There is however an important phenotypic diversity observed in natural populations of trees (\cite{kremer2020oaks}), and this phenotypic diversity, coupled with adults tolerance to harsh meteorological conditions and to plasticity could allow the populations to evolve on a faster scale that what is described by the approach introduced in this manuscript. An analysis of the dynamics of the continuous structured population model when the phenotypic variance of the population is not small would be a great complement to this manuscript, and could open a fruitful discussion about the effect of plasticity and evolution in tree populations facing the current climate change. \\

Under the assumption of small phenotypic variance and Gaussian approximation of the population, the macroscopic model we obtain after the two asymptotic limits is a system of two coupled differential equations. The solutions can then be represented in a simple phase plane, and we took advantage of the simplicity of this final model to describe its evolutionary dynamics thoroughly, as well as its ecological outcomes. Note that the approach we have developed could be adapted to describe other phenological traits, such as winter dormancy (that can be driven by the photoperiod or by a threshold temperature, both of which can be adaptative) or situations where adult trees are more plastic and tolerant than young ones. We have also used neglected the detailed age structure of the population \cite{magal2018theory}, as well as their spatial structure \cite{cantrell2010spatial}, in order to obtain a simpler model. Finally, the assumption of a population of constant size that we made can be modified to observe ecological consequences of environmental effects, beyond the mortality rate that we consider as an output of the model here. We believe these macroscopic models, thanks to their simplicity, could be interesting tools to investigate the effect of different forest management practices.

\section*{Acknowledgements}
The authors were partially funded by the Chair Modélisation Mathématique et Biodiversité of Veolia - Ecole polytechnique - Museum national d’Histoire naturelle - Fondation X.
The second author acknowledges support from the ANR under grant DEEV: ANR-20-CE40-0011-01 and from the European Union (ERC-Adg SINGER, 101054787).

\section{Data Availability}
The codes to reproduce the figures of this article are available at https://github.com/SirineBoucenna/Plasticity-model.git

\begin{appendices}

\section{Preliminary technical lemma}

We define the application $\mathcal T$ below that will be useful to understand the effect of the sexual reproduction term in \eqref{model:m} and \eqref{model:n}:
\begin{align}
    &\mathcal T(f,g)(x,y)\nonumber\\
    &\quad =\iint_{\mathbb R^2}\iint_{\mathbb R^2}  \Gamma_{\sigma_x^2}\left(x-\frac{x_{1}+x_{2}}{2}\right)\Gamma_{\sigma_y^2}\left(y-\frac{y_{1}+y_{2}}{2}\right) x_1^+f(x_1,y_1)g(x_2,y_2)\,dx_1\,dy_1\,dx_2\,dy_2.\label{def:T}
\end{align}
$\mathcal T$ is actually a quadratic operator over $L^1((1+e^x)\,dx\,dy)$, as shown by the following proposition:
\begin{prop}\label{prop:operatorT}
    The non-linear operator $\mathcal T$ is defined on $L^1((1+e^x)\,dx\,dy)$ and there is $C>0$ such that for $f,g\in L^1((1+e^x)\,dx\,dy)$,
\[\|\mathcal T(f,g)\|_{L^1((1+e^x)\,dx\,dy)}\leq C\|f\|_{L^1((1+e^x)\,dx\,dy)}\|g\|_{L^1((1+e^x)\,dx\,dy)}.\]
\end{prop}
\begin{remark}
    This Proposition shows that $\mathcal T$ is well defined as an operator on $L^1((1+e^x)\,dx\,dy)$. This estimate allows us to define global solution of the annual structured population model \eqref{model:m}, see Theorem~\ref{thm:m-existence}. Since this is a quadratic estimate rather than a linear estimate, it  could only lead to a local existence result for the continuous structured population model~\eqref{model:n}. To get around this limitation and prove Theorem~\ref{thm:n-existence}, we will take advantage of the fact that the total number of seeds produced by an individual tree is limited by the available water (see \eqref{est:x1rho}), which is a property specific to the model we consider in this manuscript. This will allow us to use a more classical $L^1$ estimate (without weight) on the reproduction term, and the conservation of the total population mass \eqref{eq:mp1} can then be used to obtain a linear estimate on the reproduction term. 
\end{remark}

\begin{proof}[Proof of Proposition~\ref{prop:operatorT}]
For $f,g\in L^1((1+e^x)\,dx\,dy)$, thanks to a change of variable, we can write
\begin{align*}
    &\mathcal T(f,g)(x,y)\\
    &\quad=\iint_{\mathbb R^2}  \Gamma_{\sigma_x^2}\left(\frac12\left(2x-X_{1}\right)\right)\Gamma_{\sigma_y^2}\left(\frac 12\left(2y-Y_{1}\right)\right) \\
    &\phantom{\quad=\iint_{\mathbb R^2}  } \left[\iint_{\mathbb R^2} (X_1-x_2)^+f(X_1-x_2,Y_1-y_2)g(x_2,y_2)\,dx_2\,dy_2\right]\,dX_1\,dY_1,
\end{align*}
and using $*$ to represent convolutions over $\mathbb R^2$, we notice that $\mathcal T(f,g)$ can be seen as a convolution (see \cite{mirrahimi2013dynamics} for a similar idea):
\begin{align*}
    \mathcal T(f,g)&= \left[\left((x,y)\mapsto \Gamma_{\sigma_x^2}\left(\frac x 2\right)\Gamma_{\sigma_y^2}\left(\frac y 2\right)\right)\ast \big(\left((x,y)\mapsto xf(x,y)\right)\ast g\big)\right](2x,2y).
\end{align*}
Therefore
\begin{align*}
    \iint_{\mathbb R^2} \mathcal T(f,g)(x,y)\,dx\,dy&\leq \left\|(x,y)\mapsto \Gamma_{\sigma_x^2}\left(\frac x 2\right)\Gamma_{\sigma_y^2}\left(\frac y 2\right)\right\|_{L^1(\mathbb R^2)}\|(x,y)\mapsto x\,f(x,y)\|_{L^1(\mathbb R^2)}\|g\|_{L^1(\mathbb R^2)}\\
    &\leq 4\|f\|_{L^1((1+e^x)\,dx\,dy)}\|g\|_{L^1((1+e^x)\,dx\,dy)},
\end{align*}
and in particular $\mathcal T(f,g)\in L^1(\mathbb R^2)$.
Moreover,
\begin{align*}
    &\iint_{\mathbb R^2} e^{x}\mathcal T(f,g)(x,y)\,dx\,dy=\iint_{\mathbb R^2}\iint_{\mathbb R^2}  \left(\int e^{x}\Gamma_{\sigma_x^2}\left(x-\frac{x_{1}+x_{2}}{2}\right)\,dx\right) x_1^+\,g(x_1,y_1)f(x_2,y_2)\,dx_1\,dy_1\,dx_2\,dy_2\nonumber\\
    &\quad \leq C\iint_{\mathbb R^2}\iint_{\mathbb R^2}  \left(x_1^+\int_{\mathbb R} e^{x-\frac 1{2\sigma_x^2}\left(x-\frac{x_{1}+x_{2}}{2}\right)^2}\,dx\right) g(x_1,y_1)f(x_2,y_2)\,dx_1\,dy_1\,dx_2\,dy_2,
\end{align*}
where $C>0$ designates a universal constant that can change from one line to the next. Thanks to a change of variable $\tilde x:=\frac 1{\sqrt 2\sigma_x}\left(x-\frac{x_1+x_2}2\right)$,
\begin{align*}
    &\iint_{\mathbb R^2} e^{x}\mathcal T(f,g)(x,y)\,dx\,dy \leq C\iint_{\mathbb R^2}\iint_{\mathbb R^2}  \left(x_1^+\int_{\mathbb R} e^{\left(\tilde x+\frac{x_{1}+x_{2}}{2}\right)-\frac {\tilde x^2}{2\sigma_x^2}}\,d\tilde x\right) g(x_1,y_1)f(x_2,y_2)\,dx_1\,dy_1\,dx_2\,dy_2\nonumber\\
    &\quad \leq C\iint_{\mathbb R^2}\iint_{\mathbb R^2}  \left(x_1^+e^{\frac{x_{1}+x_{2}}{2}}\int_{\mathbb R} e^{\tilde x-\frac {\tilde x^2}{2\sigma_x^2}}\,d\tilde x\right) g(x_1,y_1)f(x_2,y_2)\,dx_1\,dy_1\,dx_2\,dy_2\nonumber\\
    &\quad \leq C\left(\iint_{\mathbb R^2}  x_1^+e^{\frac {x_1}2} g(x_1,y_1)\,dx_1\,dy_1\right)\left(\iint_{\mathbb R^2} e^{\frac {x_2}2}f(x_2,y_2)\,dx_2\,dy_2\right)\leq C\|g\|_{L^1((1+e^x)\,dx\,dy)}\|f\|_{L^1((1+e^x)\,dx\,dy)}.
\end{align*}

\end{proof}
We prove next the Lipschitz continuity of $t\mapsto \bar \phi^{water}(t,\cdot,x)$ for the total variation norm $\|\cdot\|_{TV}$, which we define as 
\[\|u-v\|_{TV}:=\int_0^1|u(\tau)-v(\tau)|\,d\tau,\]
for $u,v\in L^1([0,1])$). We also compare $\bar \phi^{water}(t,\cdot,x)$ and $\phi^{water}_{\lfloor t/\varepsilon\rfloor}(\cdot,x)$. These estimates will be used in Section~\ref{subsec:asymptoticsepsilon} to prove Theorem~\ref{thm:n-to-m} connecting the annual structured model and the continuous structured model.
\begin{lemma}\label{lem:phiwater}
    There is $C>0$ such that for $t,t'\geq 0$, $x,x'\in\mathbb R$ and $\varepsilon>0$,
    \[\left\|\bar \phi^{water}(t,\cdot,x)-\bar \phi^{water}(t',\cdot,x')\right\|_{TV}=\int_0^1\left|\bar \phi^{water}(t,\tau,x)-\bar \phi^{water}(t',\tau,x')\right|\,d\tau\leq C|x-x'|+C|t-t'|.\]
    \[\left\|\bar \phi^{water}(t,\cdot,x)-\phi^{water}_{\lfloor t/\varepsilon\rfloor}(\cdot,x)\right\|_{TV}=\int_0^1\left|\bar \phi^{water}(t,\tau,x)- \phi^{water}_{\lfloor t/\varepsilon\rfloor}(\tau,x)\right|\,d\tau\leq C\varepsilon.\]
\end{lemma}
\begin{proof}
Let $W(t,x)\in [0,1]$ as the seasonal time when the individual with trait $x$ runs out of water (see \eqref{Def:Wt}). Then $W(t,x)$ is defined by $W(t,x)=1$ if $(1+\alpha x^+)+\beta \int_{0}^1 (\bar T(t,\tau))^+\,d\tau-\bar P(t)\leq 0$, and $\bar \varphi( t,W(t,x),x)=0$ otherwise, where
\begin{align*}
    \bar \varphi(t,\tau,x)&=(1+\alpha x^+)\tau+\beta \int_{0}^\tau (\bar T(t,\tau))^+\,d\tau-\bar P(t).
\end{align*}
To prove the first estimate of the lemma, we notice that the definition~\ref{def:phiwater} of $\phi^{water}$ implies
    \[\int_0^1\left|\bar \phi^{water}(t,\tau,x)-\bar \phi^{water}(t',\tau,x')\right|\,d\tau=\left|W(t,x)-W(t',x')\right|.\]
If $W(t,x)=1=W(t',x')$, the result is proven, so that we may assume $W(t,x)<W(t',x')\leq 1$, and then
\begin{align*}
    0&\leq \left[(1+\alpha x^+)W(t,x)+\beta \int_{0}^{W(t,x)} (\bar T(t,\tau))^+\,d\tau-\bar P(t)\right]\\
    &\quad -\left[(1+\alpha (x')^+)W(t',x')+\beta \int_{0}^{W(t',x')} (\bar T(t',\tau))^+\,d\tau-\bar P(t')\right]\\
    &=-(1+\alpha x^+)\left(W(t',x')-W(t,x)\right)+\alpha((x')^+-x^+)W(t',x')\\
    &\quad -\beta \int_{W(t,x)}^{W(t',x')} (\bar T(t,\tau))^+\,d\tau+\beta \int_{0}^{W(t',x')} |\bar T(t,\tau)-\bar T(t',\tau)|\,d\tau\\
    &\leq -\left|W(t,x)-W(t',x)\right|+C|x-x'|+C|t-t'|.
\end{align*}
Therefore $\left|W(t,x)-W(t',x)\right|\leq C|x-x'|+C|t-t'|$, which proves the first statement of the lemma. To prove the second result, we introduce $W_{\lfloor t/\varepsilon \rfloor}(x)\geq 0$, that is defined by $W_{\lfloor t/\varepsilon \rfloor}(x)=1$ if $(1+\alpha x^+)+\beta \int_{0}^1 (T_k(\tau'))^+\,d\tau'-P_k>0$, and $ \varphi( {\lfloor t/\varepsilon \rfloor},W_{\lfloor t/\varepsilon \rfloor}(x),x)=0$ otherwise, where
\begin{align*}
\varphi(k,\tau,x)&=(1+\alpha x^+)\tau+\beta \int_{0}^\tau (T_k(\tau'))^+\,d\tau'-P_k.
\end{align*}
If $W_{\lfloor t/\varepsilon \rfloor}(x)=1=W(t,x)$, the statement is proven. Otherwise, we may assume w.l.o.g. that $W_{\lfloor t/\varepsilon \rfloor}(x)<W(t,x)$. Then,
\begin{align*}
    0&\leq \left[(1+\alpha x^+)W_{\lfloor t/\varepsilon \rfloor}(x)+\beta \int_{0}^{W_{\lfloor t/\varepsilon \rfloor}(x)} (\bar T_{\lfloor t/\varepsilon \rfloor}(\tau))^+\,d\tau-\bar P(\lfloor t/\varepsilon \rfloor)\right]\\
    &\quad -\left[(1+\alpha x^+)W(t,x)+\beta \int_{0}^{W(t',x)} (\bar T(t,\tau))^+\,d\tau-\bar P(t)\right]\\
    &\leq-(1+\alpha x^+)\left(W(t,x)-W_{\lfloor t/\varepsilon \rfloor}(x)\right)-\beta \int^{W(t,x)}_{W_{\lfloor t/\varepsilon \rfloor}(x)} (\bar T(t,\tau))^+\,d\tau\\
    &\quad +\beta \int_{0}^{W_{\lfloor t/\varepsilon \rfloor}(x)} |\bar T(t,\tau)-T_{\lfloor t/\varepsilon \rfloor}(\tau)|\,d\tau +\left(\bar P(t)-\bar P({\lfloor t/\varepsilon \rfloor})\right) \\
    &\leq -\left|W(t,x)-W_{\lfloor t/\varepsilon \rfloor}(x)\right|+C\varepsilon,
\end{align*}
and the result unfolds, since
\[\bar \phi^{water}(t,s,x)=1_{s\leq W(t,x)},\quad \phi^{water}_{\lfloor t/\varepsilon\rfloor}(s,x)=1_{s\leq W_{\lfloor t/\varepsilon\rfloor}(x)}.\]

\end{proof}

\section{Proof of Theorem~\ref{thm:m-existence} - Existence and uniqueness of the solution of the annual structured model}
We will prove this result through an induction on $k\in\mathbb N$. Assume that the non-negative function $m_{k-1}(0,\cdot,\cdot)$ satisfies $\left((x,y)\mapsto m_{k-1}(0,x,y)\right)\in L^1((1+e^{x})\,dx\,dy)$ and \eqref{eq:mp1}. Notice this holds for $k=1$ thanks to the assumptions made on $m^0$ in Theorem~\ref{thm:m-existence}. For any $(x,y)\in\mathbb R^2$, the first equation of \eqref{model:m} is an ordinary differential equation, with the following explicit solution for $s\in[0,1]$:
\begin{equation}\label{eq:mkm1}
    m_{k-1}(s,x,y) =  m_{k-1}(0,x,y)e^{-\varepsilon ds- \varepsilon \int_0^s\phi^{plast}_{k-1}(\tau,y)\gamma (T_{k-1}(\tau))^+\left(1-\phi^{water}_{k-1}(\tau,x)\right)\,d\tau},
\end{equation}
and since $0\leq \varepsilon ds+\varepsilon\int_0^s\phi^{plast}_{k-1}(\tau,y)\gamma (T_{k-1}(\tau))^+\left(1-\phi^{water}_{k-1}(\tau,x)\right)\,d\tau$, we have 
\begin{equation}\label{est:L1xmk}
\|m_{k-1}(s,\cdot,\cdot)\|_{L^1((1+e^x)\,dx\,dy)}\leq \|m_{k-1}(0,\cdot,\cdot)\|_{L^1((1+e^x)\,dx\,dy)},
\end{equation}
for $s\in[0,1]$. Note also that this solution is non-negative, and so is $s_k$ thanks to its definition. We can apply Proposition~\ref{prop:operatorT} to show 
\begin{align*}
&\|s_k\|_{L^1((1+e^x)\,dx\,dy)}\leq   \nu \|(x,y)\mapsto \Gamma_{2\sigma_x^2}(x)\Gamma_{2\sigma_y^2}(y)\|_{L^1((1+e^x)\,dx\,dy)}\\ 
&\qquad +\frac 1\eta\int_0^1\bigg\|\iint_{\mathbb R^2}\iint_{\mathbb R^2}  \Gamma_{\sigma_x^2}\left(x-\frac{x_{1}+x_{2}}{2}\right)\Gamma_{\sigma_y^2}\left(y-\frac{y_{1}+y_{2}}{2}\right)\\
&\phantom{\qquad +\frac 1\eta\int_0^1\bigg\|\iint_{\mathbb R^2}\iint_{\mathbb R^2} }x_1^+ m_k(s,x_1, y_1)m_k(s, x_2,y_2)   \,dx_1\,dx_2\,dy_1\,dy_2\bigg\|_{L^1((1+e^x)\,dx\,dy)}\,ds\\
&\quad \leq C\nu +\frac{C}{\eta}\int_0^1 \|m_k(s,\cdot,\cdot)\|_{L^1((1+e^x)\,dx\,dy)} ^2\,ds\leq C\nu +\frac{C}{\eta}\|m_k(0,\cdot,\cdot)\|_{L^1((1+e^x)\,dx\,dy)} ^2.
\end{align*}
We notice next that $m_k(0,\cdot,\cdot)$ is a convex combination of $m_{k-1}(1,\cdot,\cdot)$ and $\frac{s_k}{\iint s_k(x,y)\,dx\,dy}$, with a coefficient $\mathcal O(\varepsilon)$ multiplying the latter since $ \varepsilon\int_0^s\phi^{plast}_{k-1}(\tau,y)\gamma (T_{k-1}(\tau))^+\left(1-\phi^{water}_{k-1}(\tau,x)\right)\,d\tau\leq C\varepsilon$ and \eqref{eq:mkm1} implies 
\[0\leq 1-\iint_{\mathbb R^2}m_{k-1}(\hat x,\hat y)\,d\hat x\,d\hat y\leq C\varepsilon.\]
Then $m_{k}(0,\cdot,\cdot)$ is non-negative and  we have
\begin{align*}
    \|m_{k}(0,\cdot,\cdot)\|_{L^1((1+e^x)\,dx\,dy)}&\leq \|m_{k-1}(1,\cdot,\cdot)\|_{L^1((1+e^x)\,dx\,dy)}+C\varepsilon\|s_{k}\|_{L^1((1+e^x)\,dx\,dy)}\\
    &\leq \|m_{k-1}(0,\cdot,\cdot)\|_{L^1((1+e^x)\,dx\,dy)}+C\varepsilon\|m_{k-1}(0,\cdot,\cdot)\|_{L^1((1+e^x)\,dx\,dy)}^2+C\varepsilon.
\end{align*}
Therefore,
\begin{align*}
    \left(1+\|m_{k}(0,\cdot,\cdot)\|_{L^1((1+e^x)\,dx\,dy)}
    \right)&\leq C\left(1+\|m_{k-1}(0,\cdot,\cdot)\|_{L^1((1+e^x)\,dx\,dy)}\right)^2,
\end{align*}
and then $\|m_{k}(0,\cdot,\cdot)\|_{L^1((1+e^x)\,dx\,dy)}\leq C^{2^k-1}\left(1+\|m_{0}(0,\cdot,\cdot)\|_{L^1((1+e^x)\,dx\,dy)}\right)^{2^k}$. We are therefore able to construct $(m_k)\in L^1([0,1],L^1((1+e^{x})\,dx\,dy))$ for $k\in\mathbb N\cup\{0\}$ recursively, with 
\[\|m_{k}(s,\cdot,\cdot)\|_{L^1((1+e^x)\,dx\,dy)} \leq C^{2^k},\]
for $s\in[0,1]$. The solution is well defined, unique, non-negative and we can prove  equality \eqref{eq:mp1} thanks to the definition of $m_k(0,\cdot,\cdot)$:
\begin{align*}
    \iint_{\mathbb R^2}m_k(0,x,y)\,dx\,dy=\iint_{\mathbb R^2}m_{k-1}(1,x,y)\,dx\,dy+\frac{1-\iint_{\mathbb R^2} m_{k-1}(1,\hat x,\hat y)\,d\hat x\,d\hat y}{\iint_{\mathbb R^2} s_{k-1}(\hat x,\hat y)\,d\hat x\,d\hat y} \iint_{\mathbb R^2}s_{k-1}(x,y)\,dx\,dy=1.
\end{align*}

\section{Proof of Theorem~\ref{thm:n-existence} - Existence and uniqueness of the solution of the continuous structured model}

\noindent\textbf{Existence and uniqueness of solutions}

We define the operator $\mathcal G$ as follows 
\begin{align}
    \mathcal G[n](t,x,y)& = \iint_{\mathbb R^2}\iint_{\mathbb R^2}  \Gamma_{\sigma_x^2}\left(x-\frac{x_{1}+x_{2}}{2}\right)\Gamma_{\sigma_y^2}\left(y-\frac{y_{1}+y_{2}}{2}\right) x_1^+
  n(t,x_1, y_1)n(t, x_2,y_2)\nonumber\\
  &\phantom{ = \iiiint} \bar\rho[n(t,\cdot,\cdot)](t,x_1,x_2,y_1,y_2) \,dx_1\,dy_1\,dx_2\,dy_2 + \nu\Gamma_{2\sigma_x^2}(x)\Gamma_{2\sigma_y^2}(y),\label{eq:G}
\end{align}
and this operator satisfies
\begin{equation}\label{est:G}
\mathcal G[n](t,x,y)\leq C\left(\max_{(x_1,x_2,y_1,y_2)\in\mathbb R^4}\left(x_1^+\bar\rho[n(t,\cdot,\cdot)](t,x_1,x_2,y_1,y_2) \right)\eta\|n\|_{L^\infty([0,T],L^1(\mathbb R^2))}^2+\nu\right).
\end{equation}
Solutions of \eqref{model:n} satisfy $n(t,x,y) = \mathcal F[n](t,x,y)$, with
\begin{align}
    \mathcal F[n](t,x,y)&:=n^0(x,y)e^{-\int_0^t\bar a(s,x,y)\,ds} \nonumber\\
    &\quad +\int_0^t \frac{\iint_{\mathbb R^2} \bar a(s,\hat x,\hat y)n(s,\hat x,\hat y) \,d\hat x\,d\hat y }{\nu+\bar R[n](s)}\mathcal G[n](s,x,y)e^{-\int_s^t\bar a(\tau,x,y)\,d\tau}\,ds.\label{def:Fn}
\end{align}

Let $T>0$. The operator maps $L^\infty([0,T],L^1(\mathbb R^2,\mathbb R_+))$ into itself, since $\mathcal F[n](t,x,y)\geq 0$ whenever $n\geq 0$, and for $t\in[0,T]$,
\begin{align}
    &\|\mathcal F[n](t,\cdot,\cdot)\|_{L^1(\mathbb R^2)}\leq \|n^0\|_{L^1(\mathbb R^2)}\nonumber\\
    &\qquad + C\frac{T\gamma}\nu \|n\|_{L^\infty([0,T],L^1(\mathbb R^2))}\left(\frac {\max_{(x_1,x_2,y_1,y_2)\in\mathbb R^4}\left(x_1^+\bar\rho[n(t,\cdot,\cdot)](t,x_1,x_2,y_1,y_2) \right)}\eta\|n\|_{L^\infty([0,T],L^1(\mathbb R^2))}^2+\nu\right)\nonumber\\
    &\quad \leq C\left(1+\|n\|_{L^\infty([0,T],L^1(\mathbb R^2))}^3\right),\label{est:LinftL1}
\end{align}
where we have used that for $(x_1,x_2,y_1,y_2)\in\mathbb R^4$,
\begin{align}
    x_1^+\bar\rho[n(t,\cdot,\cdot)](t,x_1,x_2,y_1,y_2)&=x_1^+\int_{0}^{1}  \frac{\bar \phi^{plast}(t,s,y_1)\bar \phi^{plast}(t,s,y_2)\bar \phi^{water}(t,s,x_1)\bar \phi^{water}(t,s,x_2)}{\eta+\bar Q[n](t)} \,ds\nonumber\\
    &\leq \frac{x_1^+}\eta\int_{0}^{1}  \bar \phi^{water}(t,s,x_1)\,ds\leq \frac{x_1^+}\eta\int_{0}^{1}  1_{\alpha x_1^+s\leq \|\bar P\|_{L^\infty([0,\infty))}}\,ds\nonumber\\
    &\leq \frac {\|\bar P\|_{L^\infty([0,\infty))}}{\eta \alpha}\leq C.\label{est:x1rho}
\end{align}
For  $n,\tilde n\in L^\infty([0,T],L^1(\mathbb R^2))$ and $s\in[0,T]$, we have
\begin{align}
  &\iint_{\mathbb R^2}\left|\mathcal F[n]-\mathcal F[\tilde n]\right|(s,x,y)\,dx\,dy \nonumber\\
 & \quad \leq T \max_{t\in[0,T]} \iint_{\mathbb R^2} \left| \frac{\iint_{\mathbb R^2} \bar a(t,\hat x,\hat y)n(t,\hat x,\hat y) \,d\hat x\,d\hat y}{\nu +\bar R[n](t)}\mathcal G[n](t,x,y)- \frac{\iint_{\mathbb R^2} \bar a(t,\hat x,\hat y)\tilde n(t,\hat x,\hat y) \,d\hat x\,d\hat y}{ \nu+\bar R[\tilde n](t)} \mathcal G[\tilde n](t,x,y)\right|\,dx\,dy \nonumber \\
 & \quad \leq T \max_{t\in[0,T]}\bigg\{ \iint_{\mathbb R^2}  \bigg( \frac{\iint_{\mathbb R^2} \bar a(t,\hat x,\hat y)n(t,\hat x,\hat y) \,d\hat x\,d\hat y}{\nu+\bar R[n](t)} \left|  \mathcal G[n](t,x,y)-\mathcal G[\tilde n](t,x,y) \right| \nonumber \\
 & \qquad + \frac{\mathcal G[\tilde n](t,x,y)}{\nu + \bar R[\tilde n](t)} \iint _{\mathbb R^2} \bar a(t,\hat x, \hat y)\left| n(t, \hat x,\hat y) - \tilde n(t, \hat x,\hat y)  \right| \,d\hat x\,d\hat y \nonumber\\
 & \qquad + \frac{\mathcal G[\tilde n](t,x,y)}{\nu+\bar R[\tilde n](t)} \left| \frac{\nu+\bar R[\tilde n](t)}{\nu+\bar R[n](t)}-1\right|\iint_{\mathbb R^2}  \bar a(t,\hat x, \hat y)n(t, \hat x,\hat y) d\hat x\,d\hat y\bigg) \,dx\,dy \bigg\}\nonumber  \\
  &\quad\leq  T \gamma\|\bar T\|_{L^\infty}\max_{t \in [0,T]}  \bigg\{\iint_{\mathbb R^2}  \bigg( \frac{1}{\nu +\bar R[n](t)}\left|  \mathcal G[n](t,x,y)-\mathcal G[\tilde n](t,x,y) \right|\nonumber \\
  & \qquad + \frac{\mathcal G[\tilde n](t,x,y)}{\nu+\bar R[\tilde n](t)} \iint_{\mathbb R^2}  \left| n(t, \hat x,\hat y)- \tilde n(t, \hat x,\hat y) \right| d\hat x\,d\hat y + \frac{\mathcal G[\tilde n](t,x,y)}{\nu+\bar R[\tilde n](t)} \left| \frac{\bar R[\tilde n](t) - \bar R[n](t)}{\nu+ \bar R[n](t)} \right| \bigg)\, dx\,dy \bigg\}\nonumber  \\
 & \quad\leq T\frac{\gamma\|\bar T\|_{L^\infty}}{\nu}\max_{t \in [0,T]}  \bigg\{\iint_{\mathbb R^2}  \left|  \mathcal G[n](t,x,y)-\mathcal G[\tilde n](t,x,y) \right| \,dx\,dy \nonumber \\
  & \qquad+ C_{}  \iint_{\mathbb R^2} | n(t,x,y) - \tilde n(t,x,y) | \,dx\,dy+ \frac{C}{\nu} \left| \bar R[\tilde n](t) - \bar R[n](t)\right| \bigg\},
 \label{3.2}
\end{align}
where we have used the bound on $\left|\iint \mathcal G[n](t,x,y) \,dx\,dy\right|\leq C\|n(t,\cdot,\cdot)\|_{L^1(\mathbb R^2)}^2$ that follows from \eqref{est:G} and  \eqref{est:x1rho} (see \eqref{est:LinftL1} for a similar argument). Furthermore,
\begin{align}
  & \iint_{\mathbb R^2} \left|\mathcal G[n](s,x,y)-\mathcal G[\tilde n](s,x,y)\right| \,dx\,dy
    \leq \iint_{\mathbb R^2}  \iint_{\mathbb R^2}\iint_{\mathbb R^2}  \Gamma_{\sigma_x^2}\left(x-\frac{x_{1}+x_{2}}{2}\right)\Gamma_{\sigma_y^2}\left(y-\frac{y_{1}+y_{2}}{2}\right) 
     \nonumber \\
    &\qquad 
  \big|\tilde n(s,x_1, y_1)\tilde n(s, x_2,y_2)\left(x_1^+\bar\rho[n(s,\cdot,\cdot)](s,x_1,x_2,y_1,y_2) -x_1^+\bar\rho[\tilde n(s,\cdot,\cdot)](s,x_1,x_2,y_1,y_2) \right) \nonumber\\
  &\qquad+\left(n(s,x_1, y_1)n(s, x_2,y_2)-\tilde n(s,x_1, y_1)\tilde n(s, x_2,y_2)\right)x_1^+\bar\rho[n(s,\cdot,\cdot)](s,x_1,x_2,y_1,y_2)\big| \,dx_1\,dy_1\,dx_2\,dy_2\,dx\,dy. \nonumber \\
 &\quad \leq \iint_{\mathbb R^2}\iint_{\mathbb R^2}  
  \tilde n(s,x_1, y_1)\tilde n(s,x_2,y_2)\nonumber\\
  &\phantom{quad \leq \iint_{\mathbb R^2}\iint_{\mathbb R^2}  
  }\Big|x_1^+\bar\rho[n(s,\cdot,\cdot)](s,x_1,x_2,y_1,y_2) -x_1^+\bar\rho[\tilde n(s,\cdot,\cdot)](s,x_1,x_2,y_1,y_2) \Big| \,dx_1\,dx_2\,dy_1\,dy_2 \nonumber\\
  &\qquad+\iint_{\mathbb R^2}\iint_{\mathbb R^2}\left[\left|n-\tilde n\right|(s,x_1, y_1)\tilde n(s, x_2,y_2)+n(s,x_1,y_1)\left|n-\tilde n\right|(s,x_2, y_2)\right]\nonumber\\
  &\qquad \phantom{sqfezfaz}x_1^+\bar\rho[n(s,\cdot,\cdot)](s,x_1,x_2,y_1,y_2) \,dx_1\,dy_1\,dx_2\,dy_2. \nonumber
\end{align}
Since $(s,x_1,x_2,y_1,y_2)\mapsto x_1^+\bar\rho[n(s,\cdot,\cdot)](s,x_1,x_2,y_1,y_2)$ is uniformly bounded, we simply need to estimate the first term on the right hand side:
\begin{align}
    &x_1^+\Big|\bar\rho[n(s,\cdot,\cdot)](s,x_1,x_2,y_1,y_2) -\bar\rho[\tilde n(s,\cdot,\cdot)](s,x_1,x_2,y_1,y_2) \Big|\nonumber\\
    &\quad \leq \left|x_1^+ \int_0^1\bar \phi^{water}(t,s,x_1)\,ds\right|\left|\frac 1{\eta+\bar Q[n](s)}-\frac 1{\eta+\bar Q[\tilde n](s)}\right|\leq \frac C{\eta^2}\|n-\tilde n\|_{L^\infty([0,T],L^1(\mathbb R^2))},\label{est:xrho}
\end{align}
estimating $\left|x_1^+ \int_0^1\bar \phi^{water}(t,s,x_1)\,ds\right|\leq C$ thanks to the argument already employed in \eqref{est:x1rho}, and the definition of $\bar Q$, see \eqref{def:barQ}. Then,
\begin{equation}\label{est:FmF}
    \iint_{\mathbb R^2} \left|\mathcal G[n](s,x,y)-\mathcal G[\tilde n](s,x,y)\right| \,dx\,dy\leq C\|n-\tilde n\|_{L^\infty([0,T],L^1(\mathbb R^2))},
\end{equation}
where the constant only depends on a bound on $\|n\|_{L^\infty([0,T],L^1(\mathbb R^2))}$ and $\|\tilde n\|_{L^\infty([0,T],L^1(\mathbb R^2))}$. Moreover, we notice that  $\bar R[n](t) = \iint \mathcal G[n](t,x,y) \,dx\,dy - \nu$, therefore 
\[
| \bar R[n](t) - \bar R[\tilde n](t) | \leq \iint_{\mathbb R^2} | \mathcal G[n](t,x,y) - \mathcal G[\tilde n](t,x,y) | \,dx\,dy\leq C\|n-\tilde n\|_{L^\infty([0,T],L^1(\mathbb R^2))}.
\]
This estimate and \eqref{est:FmF} can be used to estimate the right hand side of \eqref{3.2} and obtain
\begin{align}
    &\|\mathcal F[n]-\mathcal F[\tilde n]\|_{L^\infty([0,T],L^1(\mathbb R^2))}\leq CT \|n-\tilde n\|_{L^\infty([0,T],L^1(\mathbb R^2))},\label{est:contract}
\end{align}
where the constant only depends on a bound on $\|n\|_{L^\infty([0,T],L^1(\mathbb R^2))}$ and $\|\tilde n\|_{L^\infty([0,T],L^1(\mathbb R^2))}$. If $T>0$ is small enough, this operator is a contraction on $L^\infty([0,T],L^1(\mathbb R^2))$. The existence of solutions for a small time interval $[0,T]$, $T>0$, then follows from the Banach contraction Theorem. These solutions are non-negative and an integration of \eqref{model:n} along $(x,y)\in\mathbb R^2$ shows that solutions $n$ satisfy $\iint_{\mathbb R^2} n(t,x,y)\,dx\,dy\equiv 1$. This property and \eqref{est:contract} prove the existence of solutions for $t\in[0,+\infty)$.

\medskip

\noindent\textbf{Regularity of solutions}
The boundedness of solutions when $n^0$ is bounded can be seen as follows, using the fact that the kernel $(x,y)\mapsto\Gamma_{\sigma_x^2}(x)\Gamma_{\sigma_y^2}(y)$ appearing in \eqref{eq:G} is bounded:
\begin{align*}
&\|n(t,\cdot,\cdot)\|_{L^\infty(\mathbb R^2)}\leq \|n^0\|_{L^\infty(\mathbb R^2)}\\
&\qquad +C \left(\|n\|_{L^\infty([0,T],L^1(\mathbb R^2))}^2\int_0^t\left(\max_{(x_1,x_2,y_1,y_2)\in\mathbb R^4}x_1^+\bar\rho[n(s,\cdot,\cdot)](s,x_1,x_2,y_1,y_2) \right)\,ds+C\nu t\right)\leq C(t+1).
\end{align*}
To prove that $n$ satisfies \eqref{eq:Lips}, we notice that thanks to \eqref{def:Fn},
\begin{align*}
    &\left|n(t,x,y)-n(t,x',y')\right|\leq \left|n^0(x,y)-n^0(x',y')\right|e^{-\int_0^t\bar a(s,x,y)\,ds}+n^0(x',y')\left(e^{-\int_0^t\bar a(s,x,y)\,ds}-e^{-\int_0^t\bar a(s,x',y')\,ds}\right)^+ \nonumber\\
    &\quad +\int_0^t \frac{\iint_{\mathbb R^2} \bar a(s,\hat x,\hat y)n(s,\hat x,\hat y) \,d\hat x\,d\hat y }{\nu+\bar R[n](s)}\\
    &\quad \left(\left|\mathcal G[n](s,x,y)-\mathcal G[n](s,x',y')\right|e^{-\int_s^t\bar a(\tau,x,y)\,d\tau}+\mathcal G[n](s,x',y')\left(e^{-\int_s^t\bar a(\tau,x,y)}-e^{-\int_s^t\bar a(\tau,x',y')\,d\tau}\right)^+\right)\,ds\\
    &\leq \|n^0\|_{W^{1,\infty}(\mathbb R^2)}\left(|x-x'|+|y-y'|\right)+n^0(x,y)\int_0^t\left(\bar a(s,x',y')-\bar a(s,x,y)\right)^+\,ds\nonumber\\
    &\quad +\int_0^t \frac{\iint_{\mathbb R^2} \bar a(s,\hat x,\hat y)n(s,\hat x,\hat y) \,d\hat x\,d\hat y }{\nu+\bar R[n](s)}\\
    &\quad \left(\left|\mathcal G[n](s,x,y)-\mathcal G[n](s,x',y')\right|+\|\mathcal G[n](s,\dot,\dot)\|_{L^\infty(\mathbb R^2)}\int_s^t\left(\bar a(\tau,x',y')-\bar a(\tau,x,y)\right)^+\,d\tau\right)\,ds.
\end{align*}
We notice that $y\mapsto \bar\phi^{plast}(t,s,y)$ is non decreasing and $x\mapsto \bar \phi^{water}(t,s,x)$ satisfies the regularity estimate provided by Lemma~\ref{lem:phiwater}. Then, if $y'\leq y$,
\begin{align*}
    \bar a(t,x,y)-\bar a(t,x',y') &= \int_0^{1} \bar \phi^{plast}(t,s,y) \gamma (\bar T(t,s))^+\left(\bar \phi^{water}(t,s,x')-\bar \phi^{water}(t,s,x)\right)\,ds\\
    &+ \int_0^{1} \left(\bar \phi^{plast}(t,s,y)-\bar \phi^{plast}(t,s,y')\right) \gamma (\bar T(t,s))^+\left(1-\bar \phi^{water}(t,s,x')\right)\,ds\\
    &\geq -C|x-x'|.
\end{align*}
Moreover,
\begin{align*}
    &\left|\mathcal G[n](t,x,y)-\mathcal G[n](t,x',y') \right|= \iint_{\mathbb R^2}\iint_{\mathbb R^2}  \bigg[\left|\Gamma_{\sigma_x^2}\left(x-\frac{x_{1}+x_{2}}{2}\right)-\Gamma_{\sigma_x^2}\left(x'-\frac{x_{1}+x_{2}}{2}\right)\right|\Gamma_{\sigma_y^2}\left(y-\frac{y_{1}+y_{2}}{2}\right)\\
    &\phantom{\left|\mathcal G[n](t,x,y)-\mathcal G[n](t,x',y') \right|= \iint_{\mathbb R^2}\iint_{\mathbb R^2}}\Gamma_{\sigma_x^2}\left(x'-\frac{x_{1}+x_{2}}{2}\right)\left|\Gamma_{\sigma_y^2}\left(y-\frac{y_{1}+y_{2}}{2}\right)-\Gamma_{\sigma_y^2}\left(y'-\frac{y_{1}+y_{2}}{2}\right)\right|\bigg]\\
  &\phantom{\left|\mathcal G[n](t,x,y)-\mathcal G[n](t,x',y') \right|= \iint_{\mathbb R^2}\iint_{\mathbb R^2}} x_1^+ n(t,x_1, y_1)n(t, x_2,y_2)\bar\rho[n(t,\cdot,\cdot)](t,x_1,x_2,y_1,y_2) \,dx_1\,dy_1\,dx_2\,dy_2 \\
  &\quad \leq \left(\max_{x_1,x_2,y_1,y_2\in\mathbb R} \left(x_1^+\rho[n(t,\cdot,\cdot)](t,x_1,x_2,y_1,y_2)\right)\right)\left(\|\Gamma_{\sigma_x^2}'\|_{L^\infty(\mathbb R)}|x-x'|+\|\Gamma_{\sigma_y^2}'\|_{L^\infty(\mathbb R)}|y-y'|\right).
  \end{align*}
  Therefore, for some $C>0$, if $y'\leq y$, $n$ satisfies \eqref{eq:Lips}. To obtain the last inequality of Theorem~\ref{thm:n-existence}, we observe that thanks to \eqref{model:n},
\begin{align*}
    \frac d{dt}\iint_{\mathbb R^2} x^2 n(t,x,y)\,dx\,dy&\leq \left(\max_{x_1,x_2,y_1,y_2\in\mathbb R} \left(x_1^+\rho[n(t,\cdot,\cdot)](t,x_1,x_2,y_1,y_2)\right)\right)\\
    &\qquad \iint_{\mathbb R ^2}\iint_{\mathbb R^2} \left|\frac{x_1+x_2}2\right|^2 
    n(t,x_1, y_1)n(t, x_2,y_2) \,dx_1\,dx_2\,dy_1\,dy_2+2\nu \sigma_x^2\\
    &\quad \leq C\iint_{\mathbb R^2} x^2 n(t,x,y)\,dx\,dy+2\nu \sigma_x^2,
\end{align*}
which implies \eqref{est:moment2}. Notice that to make this last calculation fully rigorous, it would be necessary to introduce a modified model with $x\in[-R,R]$, prove the existence of solutions for this modified model, and then pass to the limit $R\to\infty$ using the calculation above. We omit this very technical argument here. 

\section{Proof of Theorem~\ref{thm:n-to-m} - asymptotic limit from the annual structured model to the continuous structured model}\label{subsec:asymptoticsepsilon}


    Using the notation \eqref{eq:G}, we have
\begin{align*}
    n(t+\varepsilon,x,y) &= n(t,x,y)e^{-\int_t^{t+\varepsilon} \bar a(s,x,y)\,ds}+\int_t^{t+\varepsilon} \frac{\iint_{\mathbb R^2} \bar a(s,\hat x,\hat y)n(s,\hat x,\hat y) \,d\hat x\,d\hat y}{\nu+\bar R[n](s)}\mathcal G[n](s,x,y)e^{-\int_s^{t+\varepsilon}\bar a(\tau,x,y)\,d\tau}\,ds,
\end{align*}
and using \eqref{model:m} we write : 
\begin{align*}
     m_{\lfloor (t+\varepsilon)/\varepsilon\rfloor}(0,x,y)  &= m_{\lfloor t/\varepsilon \rfloor} (0,x,y) e^{-\varepsilon a_{\lfloor t/\varepsilon\rfloor}(1,x,y)}  \\
     &\quad + \frac{1- \iint_{\mathbb R^2} e^{-\varepsilon a_{\lfloor t/\varepsilon\rfloor}(1,\hat x, \hat y)} m_{\lfloor t/\varepsilon\rfloor}(0,\hat x,\hat y)\,d\hat x\,d\hat y}{\nu+R_{\lfloor t/\varepsilon \rfloor}} \int_0^1 \mathcal G[ e^{-\varepsilon a}m_{\lfloor t/\varepsilon \rfloor}(0,\cdot,\cdot)](s,x,y) \,ds.
\end{align*}
These imply
\begin{align}
    & \| n(t+\varepsilon,\cdot,\cdot) - m_{\lfloor (t+\varepsilon)/\varepsilon\rfloor}(0,\cdot,\cdot) \|_{L^1(\mathbb R^2)}\leq\iint_{\mathbb R^2}\left|\, n(t,x,y) e^{-\int_t^{t+\varepsilon} \bar a(s,x,y)ds}-m_{\lfloor t/\varepsilon \rfloor} (0,x,y) e^{-\varepsilon a_{\lfloor t/\varepsilon\rfloor}(1,x,y)} \right|\,dx\,dy\nonumber \\
    & \quad +  \iint_{\mathbb R^2}\bigg| \int_t^{t+\varepsilon} \frac{\iint_{\mathbb R^2} \bar a(s,\hat x,\hat y)n(s,\hat x,\hat y) \,d\hat x\,d\hat y}{\nu+\bar R[n](s)}\mathcal G[n(s,\cdot,\cdot)](x,y)e^{-\int_s^{t+\varepsilon}\bar a(\tau,x,y)\,d\tau}\,ds\nonumber \\
     & \qquad  - \frac{1- \iint_{\mathbb R^2} e^{-\varepsilon a_{\lfloor t/\varepsilon\rfloor}(1,\hat x, \hat y)} m_{\lfloor t/\varepsilon\rfloor}(0,\hat x,\hat y)\,d\hat x\,d\hat y}{\nu+R_{\lfloor t/\varepsilon \rfloor}}\int_0^1 \mathcal G[ e^{-\varepsilon \bar a(s,\cdot,\cdot)}m_{\lfloor t/\varepsilon \rfloor}(0,\cdot,\cdot)](x,y) \,ds\bigg|\,dx\,dy.\label{eq:norm-n-m}
\end{align}
To estimate the first term on the right hand side of \eqref{eq:norm-n-m}, we recall the definitions \eqref{eq:gk} of $a_k$ and \eqref{def:gt} of $\bar a(t,x,y)$. For $s\in[t,t+\varepsilon]$,
\begin{align}
    &\left|a_{\lfloor t/\varepsilon \rfloor}(1,x,y)-\bar a(s,x,y)\right|\nonumber\\
    &\quad \leq \gamma\|T\|_{L^\infty}\left(\int_0^1\left|\phi_{\lfloor t/\varepsilon \rfloor}^{plast}(\tau,y)-\bar \phi^{plast}(s,\tau,y)\right|\,d\tau+\int_0^1\left|\phi_{\lfloor t/\varepsilon \rfloor}^{water}(\tau,x)-\bar \phi^{water}(s,\tau,x)\right|\,d\tau\right)\nonumber\\
    &\quad =C\int_0^1\left|\phi_{\lfloor t/\varepsilon \rfloor}^{plast}(\tau,y)-\bar \phi^{plast}(s,\tau,y)\right|\,d\tau+\mathcal O(\varepsilon),\label{eq:estama}
\end{align}
 thanks to Lemma~\ref{lem:phiwater}, while we have
  \begin{align*}
    &\int_{\mathbb R}\left|\phi_{\lfloor t/\varepsilon \rfloor}^{plast}(\tau,y)-\bar \phi^{plast}(s,\tau,y)\right|\,dy=\int_{\mathbb R}\left|e^{-\xi \int_0^\tau 1_{T_{\lfloor t/\varepsilon \rfloor}(\tau')>y}\,d\tau'}-e^{-\xi \int_0^\tau 1_{\bar T(t,\tau')>y}\,d\tau'}\right|\,dy\\
    &\quad \leq \xi \int_{\mathbb R}\left|\int_0^\tau 1_{T_{\lfloor t/\varepsilon \rfloor}(\tau')>y}\,d\tau'- \int_0^\tau 1_{\bar T(t,\tau')>y}\,d\tau'\right|\,dy\leq \xi \int_0^1\int_{\mathbb R} \left|1_{T_{\lfloor t/\varepsilon \rfloor}(\tau')>y}- 1_{\bar T(t,\tau')>y}\right|\,dy\,d\tau'\\
    &\quad \leq \xi \int_0^1\left|T_{\lfloor t/\varepsilon \rfloor}(\tau')-\bar T(t,\tau')\right|\,d\tau'\leq\xi \|T_{\lfloor t/\varepsilon \rfloor}-\bar T(t,\cdot)\|_{L^\infty}\leq C\varepsilon,
\end{align*}
and then
\begin{equation}\label{eq:ama}
    \int_{\mathbb R} \left|a_{\lfloor t/\varepsilon \rfloor}(1,x,y)-\bar a(s,x,y)\right|\,dy\leq C\varepsilon,
\end{equation}
where the constant $C>0$ is independent from $t\geq 0$, $s\in[t,t+\varepsilon]$ and from $x\in\mathbb R$. We can use this estimate to control the first term on the right hand side of \eqref{eq:norm-n-m}:
\begin{align}
    &\iint_{\mathbb R^2}\left|\, n(t,x,y) e^{-\int_t^{t+\varepsilon} \bar a(s,x,y)ds}-m_{\lfloor t/\varepsilon \rfloor} (0,x,y) e^{-\varepsilon a_{\lfloor t/\varepsilon\rfloor}(1,x,y)} \right|\,dx\,dy\nonumber\\
    &\quad \leq \left\|\, n(t,\cdot,\cdot) -m_{\lfloor t/\varepsilon \rfloor} (0,\cdot,\cdot) \right\|_{L^1(\mathbb R^2)}+ \iint_{\mathbb R^2} n(t,x,y)\left|\,\int_t^{t+\varepsilon} \bar a(s,x,y)- a_{\lfloor t/\varepsilon\rfloor}(1,x,y)\,ds \right|\,dx\,dy\nonumber\\
    &\quad \leq \left\|\, n(t,x,y) -m_{\lfloor t/\varepsilon \rfloor} (0,x,y) \right\|_{L^1(x,y)}\nonumber\\
    &\qquad + \left\| \bar a(\cdot,\cdot,\cdot)- a_{\lfloor t/\varepsilon\rfloor}(1,\cdot,\cdot)\right\|_{L^\infty([0,1]\times \mathbb R^2)}\left(\int_t^{t+\varepsilon} ds\right)\,\int_{\mathbb R}\int_{x^2\geq 1/\varepsilon} n(t,x,y)\,dx\,dy\\
    &\qquad +\|n(t,\cdot,\cdot)\|_{L^{\infty}(\mathbb R^2)} \int_{x^2\leq 1/\varepsilon}\int_t^{t+\varepsilon} \int_{\mathbb R}\left|\bar a(s,x,y)- a_{\lfloor t/\varepsilon\rfloor}(1,x,y)\right|\,dy\,ds \,dx\nonumber\\
    &\quad \leq \left\| n(t,x,y) -m_{\lfloor t/\varepsilon \rfloor} (0,x,y) \right\|_{L^1(x,y)}+2\varepsilon^2 \int_{x^2\geq 1/\varepsilon} x^2n(t,x,y)\,dx\,dy+\nonumber\\
    &\qquad +\varepsilon\frac 2{\sqrt \varepsilon} \|n(t,\cdot,\cdot)\|_{L^{\infty}(\mathbb R^2)}\max_{|x|\leq 1/\sqrt\varepsilon,s\in[t,t+\varepsilon]} \int_{\mathbb R}\left|\bar a(s,x,y)- a_{\lfloor t/\varepsilon\rfloor}(1,x,y)\right|\,dy\nonumber\\
    &\quad \leq \left\| n(t,x,y) -m_{\lfloor t/\varepsilon \rfloor} (0,x,y) \right\|_{L^1(x,y)}+ Ce^{Ct}\varepsilon^{3/2},
    \label{est:termeini}
\end{align}
where we have used a Chebyshev's inequality, \eqref{est:moment2} and \eqref{eq:ama}. To estimate the last part of \eqref{eq:norm-n-m}, we notice that for $s\in[t,t+\varepsilon]$, 
\begin{align}
    & \iint_{\mathbb R^2}\bigg| \int_t^{t+\varepsilon}  \mathcal G[n(s,\cdot,\cdot)](x,y)e^{-\int_s^{t+\varepsilon}\bar a(\tau',x,y)\,d\tau'}\,ds-\varepsilon\int_0^1\mathcal G[ e^{-\varepsilon a_{\lfloor t/\varepsilon \rfloor}(s,\cdot,\cdot)}m_{\lfloor t/\varepsilon \rfloor}(0,\cdot,\cdot)](x,y) \,ds  \bigg|\,dx\,dy \nonumber\\
    &\quad = \iint_{\mathbb R^2}\bigg| \int_t^{t+\varepsilon}  \mathcal G[n(s,\cdot,\cdot)](x,y)\,ds-\varepsilon\int_0^1\mathcal G[m_{\lfloor t/\varepsilon \rfloor}(0,\cdot,\cdot)](x,y) \,ds  \bigg|\,dx\,dy+\mathcal O(\varepsilon^2) \nonumber\\
    &\quad =  \mathcal O(\varepsilon^2) +\iint_{\mathbb R^2}\iint \left(\iint_{\mathbb R^2} \Gamma_{\sigma_x^2}\left(x-\frac{x_{1}+x_{2}}{2}\right)\Gamma_{\sigma_y^2}\left(y-\frac{y_{1}+y_{2}}{2}\right) \,dx\,dy\right)\nonumber\\
    &\qquad \bigg|\int_{t}^{t+\varepsilon}n(s,x_1, y_1)n(s, x_2,y_2)\left(x_1^+ \bar\rho[n(t,\cdot,\cdot)](s,x_1,x_2,y_1,y_2)\right)\,ds \nonumber\\
  &\qquad -\varepsilon\left(\int_0^1x_1^+ \bar\rho[m_{\lfloor t/\varepsilon \rfloor}(0,\cdot,\cdot)](t,x_1,x_2,y_1,y_2)\,ds\right)m_{\lfloor t/\varepsilon \rfloor}(0,x_1,y_1)m_{\lfloor t/\varepsilon \rfloor}(t, x_2,y_2)\bigg|\,dx_1\,dy_1\,dx_2\,dy_2\nonumber\\
    &\quad = \mathcal O(\varepsilon^2)+ \varepsilon\iint_{\mathbb R^2}\iint_{\mathbb R^2}  \bigg|n(t,x_1, y_1)n(t, x_2,y_2)\left(x_1^+ \bar\rho[n(t,\cdot,\cdot)](t,x_1,x_2,y_1,y_2)\right)\nonumber \\
  &\qquad -m_{\lfloor t/\varepsilon \rfloor}(0,x_1,y_1)m_{\lfloor t/\varepsilon \rfloor}(t, x_2,y_2)\left(x_1^+ \bar\rho[m_{\lfloor t/\varepsilon \rfloor}(0,\cdot,\cdot)](t,x_1,x_2,y_1,y_2)\right) \bigg|\,dx_1\,dy_1\,dx_2\,dy_2,\nonumber
\end{align}
where we have used the Lipschitz regularity of $n$ in time as a function valued in $L^1(\mathbb R^2)$ (we recall that the term $x_1^+ \bar\rho[n(t,\cdot,\cdot)](t,x_1,x_2,y_1,y_2)$ is uniformly bounded, see \eqref{est:x1rho}). Then, using an estimate similar to \eqref{est:xrho}, we get
\begin{align}
    & \iint_{\mathbb R^2}\bigg| \int_t^{t+\varepsilon}  \mathcal G[n(s,\cdot,\cdot)](x,y)e^{-\int_s^{t+\varepsilon}\bar a(\tau',x,y)\,d\tau'}\,ds-\varepsilon\int_0^1\mathcal G[ e^{-\varepsilon a_{\lfloor t/\varepsilon \rfloor}(s,\cdot,\cdot)}m_{\lfloor t/\varepsilon \rfloor}(0,\cdot,\cdot)](x,y) \,ds  \bigg|\,dx\,dy \nonumber\\
    &\quad \leq \mathcal O(\varepsilon^2)+ C\varepsilon\left\|\, n(t,\cdot,\cdot) -m_{\lfloor t/\varepsilon \rfloor} (0,\cdot,\cdot) \right\|_{L^1(\mathbb R^2)}.\label{est:GG}
\end{align}
Moreover, when $s\in[t,t+\varepsilon]$,
\begin{align}
    &\left|\frac{\iint_{\mathbb R^2} \bar a(s,\hat x,\hat y)n(s,\hat x,\hat y) \,d\hat x\,d\hat y}{\nu+\bar R[n](s)}-\frac 1\varepsilon \frac{1- \iint_{\mathbb R^2} e^{-\varepsilon a_{\lfloor t/\varepsilon\rfloor}(1,\hat x, \hat y)} m_{\lfloor t/\varepsilon\rfloor}(0,\hat x,\hat y)\,d\hat x\,d\hat y}{\nu+R_{\lfloor t/\varepsilon \rfloor}[m]} \right|\nonumber\\
    &\quad =\left|\frac{\iint_{\mathbb R^2} \bar a(s,\hat x,\hat y)n(t,\hat x,\hat y) \,d\hat x\,d\hat y}{\nu+\bar R[n](t)}- \frac{\iint_{\mathbb R^2} a_{\lfloor t/\varepsilon\rfloor}(1,\hat x, \hat y) m_{\lfloor t/\varepsilon\rfloor}(0,\hat x,\hat y)\,d\hat x\,d\hat y}{\nu+R_{\lfloor t/\varepsilon \rfloor}[m]} \right|+\mathcal O(\varepsilon)\nonumber\\
    &\quad \leq\frac {1}\nu \iint_{\mathbb R^2} n(t,\hat x,\hat y)\left(\bar a(s,\hat x,\hat y)-a_{\lfloor t/\varepsilon\rfloor}(1,\hat x,\hat y)\right)\,d\hat x\,d\hat y+\frac {1}{\nu^2}\left|\bar R[n](t)-R_{\lfloor t/\varepsilon \rfloor}[m]\right|\nonumber\\
    &\qquad + \frac 1\nu \iint_{\mathbb R^2} \left|n(t,\hat x,\hat y)- m_{\lfloor t/\varepsilon\rfloor}(0,\hat x,\hat y)\right|\,d\hat x\,d\hat y+\mathcal O(\varepsilon)\leq Ce^{Ct }\sqrt\varepsilon+ C\left\|\, n(t,\cdot,\cdot) -m_{\lfloor t/\varepsilon \rfloor} (0,\cdot,\cdot) \right\|_{L^1(\mathbb R^2)},\label{est:aa}
\end{align}
where we have estimated the first term on the right-hand side with the argument used in \eqref{est:termeini}. For the second term, we noticed that $\bar R[n](t)=\iint_{\mathbb R^2} \mathcal G[n](t,x,y)\,dx\,dy$ (and a similar formula for $R_{\lfloor t/\varepsilon \rfloor}[m]$), so that  \eqref{est:GG} can be used to estimate the term $\left|\bar R[n](t)-R_{\lfloor t/\varepsilon \rfloor}[m]\right|$.

Bringing together \eqref{est:termeini}, \eqref{est:GG} and \eqref{est:aa}, the estimate \eqref{eq:norm-n-m} becomes 
\begin{align}
    & \| n(t+\varepsilon,\cdot,\cdot) - m_{\lfloor (t+\varepsilon)/\varepsilon\rfloor}(0,\cdot,\cdot) \|_{L^1(\mathbb R^2)}\nonumber\\
    &\leq  \left\|\, n(t,\cdot,\cdot) -m_{\lfloor t/\varepsilon \rfloor} (0,\cdot,\cdot)\right\|_{L^1(\mathbb R^2)}+C\varepsilon\left\|\, n(t,\cdot,\cdot) -m_{\lfloor t/\varepsilon \rfloor} (0,\cdot,\cdot)\right\|_{L^1(\mathbb R^2)}+Ce^{Ct}\varepsilon^{3/2}.\nonumber
\end{align}
We now consider a fixed time $t\geq 0$. We have
\begin{align*}
    \left\|\, n(t,\cdot,\cdot) -m_{\lfloor t/\varepsilon \rfloor} (0,\cdot,\cdot)\right\|_{L^1(\mathbb R^2)}&\leq \left\|\, n(\lfloor t/\varepsilon \rfloor\varepsilon,\cdot,\cdot) -m_{\lfloor t/\varepsilon \rfloor} (0,\cdot,\cdot)\right\|_{L^1(\mathbb R^2)}+C\varepsilon\\
    &\leq \max_{k\leq 1/\varepsilon} \left\|\, n(k\varepsilon ,\cdot,\cdot) -m_{k} (0,\cdot,\cdot)\right\|_{L^1(\mathbb R^2)}+C\varepsilon,
\end{align*}
while an induction shows
\begin{align*}
    \left\|\, n(k\varepsilon ,\cdot,\cdot) -m_{k} (0,\cdot,\cdot)\right\|_{L^1(\mathbb R^2)}\leq y(k\varepsilon),
\end{align*}
where $y'(s)=y(s)+C\sqrt\varepsilon e^{Ct}$ and $y(0)=0$, that is $y(s)=C\sqrt\varepsilon e^{ct}\left(e^s-1\right)$. Then, 
\begin{align*}
    \left\|\, n(t,\cdot,\cdot) -m_{\lfloor t/\varepsilon \rfloor} (0,\cdot,\cdot)\right\|_{L^1(\mathbb R^2)}
    &\leq C\sqrt\varepsilon e^{ct}\left(e^t-1\right)+C\varepsilon,
\end{align*}
and $\left\|\, n(t,\cdot,\cdot) -m_{\lfloor t/\varepsilon \rfloor} (0,\cdot,\cdot)\right\|_{L^1(\mathbb R^2)}\to 0$ as $\varepsilon\to 0$, which concludes the proof.

\section{Approximation of the population by a normal distribution }
\label{sec:Appendixmacro}
We want to briefly explain why we expect solutions of \eqref{model:n} to be well approximated by a normal distribution in the asymptotic limit proposed in Section~\ref{sec:macromain}. More specifically, we believe the population $(x,y)\mapsto n(t,x,y)$ is close to a normal distribution with variance $\sigma^2(\bar \sigma_x,\bar \sigma_y)$ when $\sigma>0$ is small. This approximation has been used for simpler models in e.g. (\cite{patout2023cauchy,raoul2017macroscopic}), and is related to the assumption that a population is normally distributed with a constant phenotypic variance (\emph{variance at linkage equilibrium}), which is a classical approach in population genetics (\cite{kirkpatrick1997evolution}). 

To explain this idea, assume the population is initially concentrated around $(X(t),Y(t))$. Then, for $t\geq 0$ small and $(x,y)$ in a neighborhood of $(X(t),Y(t)$, that is when $\left|(x,y)-(X(t),Y(t))\right|\leq C\sigma$, the equation \eqref{model:n} can be roughly approximated as follows:
\begin{align}
    &\partial_t n(t,x,y)   \sim - \bar a(t,X(t),Y(t)) n(t, x, y)  +\bar a(t,X(t),Y(t))\nonumber\\
    &\qquad  \iint_{\mathbb R^2}\iint_{\mathbb R^2}  \Gamma_{\sigma^2\bar \sigma_x^2}\left(x-\frac{x_{1}+x_{2}}{2}\right)\Gamma_{\sigma^2\bar \sigma_y^2}\left(y-\frac{y_{1}+y_{2}}{2}\right)  n(t,x_1, y_1)n(t, x_2,y_2)\,dx_1\,dy_1\,\,dx_2\,dy_2.\label{eq:rough}
\end{align}
The dynamics of \eqref{eq:rough} is then dominated on a time scale of order $1$ (more precisely for $0<t\ll 1/\sigma$) by the infinitesimal reproduction operator, defined for $f\in\mathcal P_2(\mathbb R^2)$ by:
\[\bar T[f](\tilde x)=\iint_{\mathbb R^2} \Gamma_{\bar \sigma_x^2}\left(\tilde x-\frac{\tilde x_{1}+\tilde x_{2}}{2}\right)\Gamma_{\bar\sigma_y^2}\left(\tilde y-\frac{\tilde y_{1}+\tilde y_{2}}{2}\right)  f(\tilde x_1, \tilde y_1)f( \tilde x_2,\tilde y_2)\,d\tilde x_1\,d\tilde x_2\,d\tilde y_1\,d\tilde y_2.\]
This reproduction operator does not change the population size (it is coherent with the model \eqref{model:n} where $\iint_{\mathbb R^2}n(t,x,y)\,dx\,dy\equiv 1$ for $t\geq 0$) and it also does not affect the mean phenotypic traits of the population. This operator however contracts solutions exponentially fast  to the distribution of the population to the Gaussian function (this can be made rigorous in the sense of the Wasserstein distance, see \cite{raoul2017macroscopic}):
\begin{equation*}
    \tilde n(t,x,y)\sim \Gamma_{2\sigma^2\bar \sigma_x^2}\left(x-X(t)\right)\Gamma_{2\sigma^2\bar \sigma_y^2}\left(y-Y(t)\right).
\end{equation*}
In particular, this keeps the population density concentrated around its mean phenotypic traits $(X(t),Y(t))$, so that when $\sigma>0$ is small, the relaxation toward a normal distribution centered around $(X(t),Y(t))$ is propagated in time and \eqref{eq:Maxwellian} holds for all times $t\geq 0$.
Note that the rough approximation \eqref{eq:rough} is introduced to explain why $n$ is normally distributed, but that a more precise approach would be necessary to obtain a rigorous asymptotic result, as we describe in Section~\ref{sec:macromain}. 

\section{Heuristic derivation of the macroscopic model \eqref{eq:dXY}}
\label{subsec:macrolimit}
In this section, we simplify the expressions \eqref{eq:dX0} and \eqref{eq:dY0} under the approximation \eqref{eq:Maxwellian}. We have
\begin{align*}
    &\iint_{\mathbb R^2} x\bar a(t,x,y) n(t,x,y)\,dx\,dy\sim\iint_{\mathbb R^2} x\bar a(t,x,y) \Gamma_{2\sigma_x^2}(x-X(t))\Gamma_{2\sigma_y^2}(y-Y(t))\,dx\,dy\\
    &\quad =\iint_{\mathbb R^2} x\Big(\bar a(t,X(t),Y(t))+(x-X(t))\partial_x \bar a(t,X(t),Y(t))+(y-Y(t))\partial_y \bar a(t,X(t),Y(t))+\frac{(x-X(t))^2}2\partial_{xx}\bar a(t,X(t),Y(t))\\
    &\qquad +\frac{(y-Y(t))^2}2\partial_{yy}\bar a(t,X(t),Y(t))+(x-X(t))(y-Y(t))\partial_{xy}\bar a(t,X(t),Y(t))+\mathcal O(\|(x-X(t),y-Y(t))\|^3) \Big)\\
    &\qquad \Gamma_{2\sigma_x^2}(x-X(t))\Gamma_{2\sigma_y^2}(y-Y(t))\,dx\,dy\\
    &\quad = X(t)\bar a(t,X(t),Y(t))+2\sigma_x^2\partial_x \bar a(t,X(t),Y(t))+\sigma_x^2 X(t)\partial_{xx}\bar a(t,X(t),Y(t))+\sigma_y^2 X(t)\partial_{yy}\bar a(t,X(t),Y(t))+\mathcal O\left(\|(\sigma_x,\sigma_y)\|^3\right).
\end{align*}
Similarly,
\begin{align*}
    &\iint_{\mathbb R^2} y\bar a(t,x,y) n(t,x,y)\,dx\,dy\sim\iint_{\mathbb R^2} y\bar a(t,x,y) \Gamma_{2\sigma_x^2}(x-X(t))\Gamma_{2\sigma_y^2}(y-Y(t))\,dx\,dy\\
    &\quad = Y(t)\bar a(t,X(t),Y(t))+2\sigma_y^2\partial_y \bar a(t,X(t),Y(t))+\sigma_x^2 Y(t)\partial_{xx}\bar a(t,X(t),Y(t))+\sigma_y^2 Y(t)\partial_{yy}\bar a(t,X(t),Y(t))+\mathcal O\left(\|(\sigma_x,\sigma_y)\|^3\right).
\end{align*}
We have
\begin{align*}
    &\iint_{\mathbb R^2} \bar a(t,x,y) n(t,x,y) \,dx\,dy\sim \iint_{\mathbb R^2} \bar a(t,x,y) \Gamma_{2\sigma_x^2}(x-X(t))\Gamma_{2\sigma_y^2}(y-Y(t)) \,dx\,dy\\
    &\quad =\iint_{\mathbb R^2} \Big(\bar a(t,X(t),Y(t))+(x-X(t))\partial_x \bar a(t,X(t),Y(t))+(y-Y(t))\partial_y \bar a(t,X(t),Y(t))+\frac{(x-X(t))^2}2\partial_{xx}\bar a(t,X(t),Y(t))\\
    &\qquad +\frac{(y-Y(t))^2}2\partial_{yy}\bar a(t,X(t),Y(t))+(x-X(t))(y-Y(t))\partial_{xy}\bar a(t,X(t),Y(t))+\mathcal O(\|(x-X(t),y-Y(t))\|^3) \Big)\\ 
    &\qquad \Gamma_{2\sigma_x^2}(x-X(t))\Gamma_{2\sigma_y^2}(y-Y(t)) \,dx\,dy\\
    &\quad =\bar a(t,X(t),Y(t))+\sigma_x^2\partial_{xx}\bar a(t,X(t),Y(t))+\sigma_y^2\partial_{yy}\bar a(t,X(t),Y(t))+\mathcal O\left(\|(\sigma_x,\sigma_y)\|^3\right).
\end{align*}
and, if we use the shortened notation $\bar\rho:=\bar\rho[\delta_{x=X(t)}\delta_{y=Y(t)}](t,X(t),Y(t),X(t),Y(t))$ and the derivatives of this quantity, we get
\begin{align*}
    &\iint_{\mathbb R^2}\iint_{\mathbb R^2}  x_1^+n(t,x_1, y_1)n(t, x_2,y_2)\bar\rho[n(t,\cdot,\cdot)](t,x_1,y_1,x_2,y_2) \,dx_1\,dy_1\,dx_2\,dy_2\\
    &\quad \sim\iint_{\mathbb R^2}\iint_{\mathbb R^2}  x_1^+\bar\rho[\Gamma_{2\sigma_x^2}(\cdot-X(t))\Gamma_{2\sigma_y^2}(\cdot-Y(t))](t,x_1,y_1,x_2,y_2) \\
    &\phantom{\quad \sim\iint_{\mathbb R^2}\iint_{\mathbb R^2}}\Gamma_{2\sigma_x^2}(x_1-X(t))\Gamma_{2\sigma_y^2}(y_1-Y(t))\Gamma_{2\sigma_x^2}(x_2-X(t))\Gamma_{2\sigma_y^2}(y_2-Y(t))\,dx_1\,dx_2\,dy_1\,dy_2\\
    &\quad =\iint_{\mathbb R^2}\iint_{\mathbb R^2}  x_1^+\Big( \bar\rho+ (x_1-X(t))\partial_{x_1}{\bar\rho}+ (x_2-X(t))\partial_{x_2}\bar\rho+ (y_1-Y(t))\partial_{y_1}{\bar\rho}+ (y_2-Y(t))\partial_{y_2}\bar\rho\\
    &\qquad +\frac{(x_1-X(t))^2}2\partial_{x_1x_1}\bar\rho+ \frac{(x_2-X(t))^2}2\partial_{x_2x_2}\bar\rho+ \frac{(y_1-Y(t))^2}2\partial_{y_1y1}\bar\rho+ \frac{(y_2-Y(t))^2}2\partial_{y_2y_2}\bar\rho\\
    &\qquad+(x_1-X(t))(x_2-X(t))\partial_{x_1x_2}\bar\rho+(y_1-Y(t))(y_2-Y(t))\partial_{y_1y_2}\bar\rho\\
    &\qquad +(x_1-X(t))(y_1-Y(t))\partial_{x_1y_1}\bar\rho+(x_2-X(t))(y_2-Y(t))\partial_{x_2y_2}\bar\rho\\
    &\qquad +(x_1-X(t))(y_2-Y(t))\partial_{x_1y_2}\bar\rho+(x_2-X(t))(y_1-Y(t))\partial_{x_2y_1}\bar\rho+\mathcal O(\|(x-X(t),y-Y(t))\|^3) \Big)\\
    &\qquad \Gamma_{2\sigma_x^2}(x_1-X(t))\Gamma_{2\sigma_y^2}(y_1-Y(t))\Gamma_{2\sigma_x^2}(x_2-X(t))\Gamma_{2\sigma_y^2}(y_2-Y(t))\,dx_1\,dy_1\,dx_2\,dy_2\\
    &\quad = \bar\rho X(t)+ 2\sigma_x^2\partial_{x_1}\bar\rho+\sigma_x^2 X(t)\partial_{x_1x_1}\bar\rho+ \sigma_x^2X(t)\partial_{x_2x_2}\bar\rho+ \sigma_y^2 X(t)\partial_{y_1y1}\bar\rho+ \sigma_y^2 X(t)\partial_{y_2y_2}\bar\rho+\mathcal O\left(\|(\sigma_x,\sigma_y)\|^3\right).
\end{align*}
We consider now the last factor of \eqref{eq:dX0}:
\begin{align*}
    &\iint_{\mathbb R^2}\iint_{\mathbb R^2}  \frac{x_{1}+x_{2}}{2}x_1^+\bar\rho[n(t,\cdot,\cdot)](t,x_1,y_1,x_2,y_2)n(t,x_1, y_1)n(t, x_2,y_2)  \,dx_1\,dy_1\,dx_2\,dy_2\\
    &\quad\sim \frac 12 \iint_{\mathbb R^2}\iint_{\mathbb R^2}  (x_{1}^2+x_{2}x_1)\bar\rho[\Gamma_{2\sigma_x^2}(\cdot-X(t))\Gamma_{2\sigma_y^2}(\cdot-Y(t))](t,x_1,y_1,x_2,y_2)\\
    &\phantom{\quad \sim\frac 12\iint_{\mathbb R^2}\iint_{\mathbb R^2}}\Gamma_{2\sigma^2}(x_1-X(t))\Gamma_{2\sigma^2}(y_1-Y(t))\Gamma_{2\sigma^2}(x_2-X(t))\Gamma_{2\sigma^2}(y_2-Y(t)) \,dx_1\,dy_1\,dx_2\,dy_2\\
    &\quad = (X(t)^2+\sigma_x^2)\bar\rho+3\sigma_x^2X(t)\partial_{x_1}\bar\rho+ \sigma_x^2 X(t)\partial_{x_2}\bar\rho\\
    &\qquad +\sigma_x^2 X^2\partial_{x_1x_1}\bar\rho+ \sigma_x^2 X^2\partial_{x_2x_2}\bar\rho+ \sigma_y^2 X^2\partial_{y_1y1}\bar\rho+ \sigma_y^2 X^2\partial_{y_2y_2}\bar\rho+\mathcal O\left(\|(\sigma_x,\sigma_y)\|^3\right),
\end{align*}
as we prove in the next calculations:
\begin{align*}
    &\iint_{\mathbb R^2}\iint_{\mathbb R^2}  x_{1}^2\bar\rho[\Gamma_{2\sigma_x^2}(\cdot-X(t))\Gamma_{2\sigma_y^2}(\cdot-Y(t))](t,x_1,y_1,x_2,y_2) \\
    &\phantom{\iint_{\mathbb R^2}\iint_{\mathbb R^2}  }\Gamma_{2\sigma^2}(x_1-X(t))\Gamma_{2\sigma^2}(y_1-Y(t))\Gamma_{2\sigma^2}(x_2-X(t))\Gamma_{2\sigma^2}(y_2-Y(t))\,dx_1\,dy_1\,dx_2\,dy_2\\
    &\quad =\iint_{\mathbb R^2}\iint_{\mathbb R^2}  x_{1}^2\Big( \bar\rho+ (x_1-X(t))\partial_{x_1}\bar\rho+ (x_2-X(t))\partial_{x_2}\bar\rho+ (y_1-Y(t))\partial_{y_1}\bar\rho+ (y_2-Y(t))\partial_{y_2}\bar\rho\\
    &\qquad +\frac{(x_1-X(t))^2}2\partial_{x_1x_1}\bar\rho+ \frac{(x_2-X(t))^2}2\partial_{x_2x_2}\bar\rho+ \frac{(y_1-Y(t))^2}2\partial_{y_1y1}\bar\rho+ \frac{(y_2-Y(t))^2}2\partial_{y_2y_2}\bar\rho\\
    &\qquad+(x_1-X(t))(x_2-X(t))\partial_{x_1x_2}\bar\rho+(y_1-Y(t))(y_2-Y(t))\partial_{y_1y_2}\bar\rho\\
    &\qquad +(x_1-X(t))(y_1-Y(t))\partial_{x_1y_1}\bar\rho+(x_2-X(t))(y_2-Y(t))\partial_{x_2y_2}\bar\rho\\
    &\qquad +(x_1-X(t))(y_2-Y(t))\partial_{x_1y_2}\bar\rho+(x_2-X(t))(y_1-Y(t))\partial_{x_2y_1}\bar\rho+\mathcal O(\|(x-X(t),y-Y(t))\|^3) \Big)\\
    &\qquad \Gamma_{2\sigma_x^2}(x_1-X(t))\Gamma_{2\sigma_y^2}(y_1-Y(t))\Gamma_{2\sigma_x^2}(x_2-X(t))\Gamma_{2\sigma_y^2}(y_2-Y(t))\,dx_1\,dy_1\,dx_2\,dy_2\\
    &\quad =(X(t)^2+2\sigma_x^2)\bar\rho+ 4\sigma_x^2X(t)\partial_{x_1}\bar\rho+\sigma_x^2 X(t)^2\partial_{x_1x_1}\bar\rho+ \sigma_x^2 X(t)^2\partial_{x_2x_2}\bar\rho+ \sigma_y^2 X(t)^2\partial_{y_1y1}\bar\rho\\
    &\qquad + \sigma_y^2 X(t)^2\partial_{y_2y_2}\bar\rho +\mathcal O\left(\|(\sigma_x,\sigma_y)\|^3\right), 
\end{align*}
where we noticed that 
\begin{align*}
&\partial_{x_1x_1}\bar\rho \int_{\mathbb R} x_1^2\frac{(x_1-X(t))^2}2\Gamma_{2\sigma_x^2}(x_1-X(t))\,dx_1\\
&\quad =\partial_{x_1x_1}\bar\rho\left(\sigma_x^2+\frac 12 \int_{\mathbb R} x^4\Gamma_{2\sigma_x^2}(x)\,dx\right)=\sigma_x^2\partial_{x_1x_1}\bar\rho+\mathcal O\left(\|(\sigma_x,\sigma_y)\|^3\right).
\end{align*}
Also,
\begin{align*}
    &\iint_{\mathbb R^2}\iint_{\mathbb R^2}  x_{1}x_2\bar\rho[\Gamma_{2\sigma_x^2}(\cdot-X(t))\Gamma_{2\sigma_y^2}(\cdot-Y(t))](t,x_1,y_1,x_2,y_2) \\
    &\phantom{\iint_{\mathbb R^2}\iint_{\mathbb R^2}  }\Gamma_{2\sigma^2}(x_1-X(t))\Gamma_{2\sigma^2}(y_1-Y(t))\Gamma_{2\sigma^2}(x_2-X(t))\Gamma_{2\sigma^2}(y_2-Y(t))\,dx_1\,dy_1\,dx_2\,dy_2\\
    &\quad =\iint_{\mathbb R^2}\iint_{\mathbb R^2}  x_{1}x_2\Big( \bar\rho+ (x_1-X(t))\partial_{x_1}\bar\rho+ (x_2-X(t))\partial_{x_2}\bar\rho+ (y_1-Y(t))\partial_{y_1}\bar\rho+ (y_2-Y(t))\partial_{y_2}\bar\rho\\
    &\qquad +\frac{(x_1-X(t))^2}2\partial_{x_1x_1}\bar\rho+ \frac{(x_2-X(t))^2}2\partial_{x_2x_2}\bar\rho+ \frac{(y_1-Y(t))^2}2\partial_{y_1y1}\bar\rho+ \frac{(y_2-Y(t))^2}2\partial_{y_2y_2}\bar\rho\\
    &\qquad+(x_1-X(t))(x_2-X(t))\partial_{x_1x_2}\bar\rho+(y_1-Y(t))(y_2-Y(t))\partial_{y_1y_2}\bar\rho\\
    &\qquad +(x_1-X(t))(y_1-Y(t))\partial_{x_1y_1}\bar\rho+(x_2-X(t))(y_2-Y(t))\partial_{x_2y_2}\bar\rho\\
    &\qquad +(x_1-X(t))(y_2-Y(t))\partial_{x_1y_2}\bar\rho+(x_2-X(t))(y_1-Y(t))\partial_{x_2y_1}\bar\rho+\mathcal O(\|(x-X(t),y-Y(t))\|^3) \Big)\\
    &\qquad \Gamma_{2\sigma^2}(x_1-X(t))\Gamma_{2\sigma^2}(y_1-Y(t))\Gamma_{2\sigma^2}(x_2-X(t))\Gamma_{2\sigma^2}(y_2-Y(t))\,dx_1\,dy_1\,dx_2\,dy_2\\
    &\quad =X(t)^2\bar\rho+ 2\sigma_x^2 X(t)\partial_{x_1}\bar\rho+ 2\sigma_x^2 X(t)\partial_{x_2}\bar\rho\\
    &\qquad +\sigma_x^2 X(t)^2\partial_{x_1x_1}\bar\rho+ \sigma_x^2 X(t)^2\partial_{x_2x_2}\bar\rho+ \sigma_y^2 X(t)^2\partial_{y_1y1}\bar\rho+ \sigma_y^2 X(t)^2\partial_{y_2y_2}\bar\rho+\mathcal O\left(\|(\sigma_x,\sigma_y)\|^3\right). 
\end{align*}
And we can finally consider now the last factor of \eqref{eq:dY0}:
\begin{align*}
    &\iint_{\mathbb R^2}\iint_{\mathbb R^2}  \frac{y_{1}+y_{2}}{2}x_1^+n(t,x_1, y_1)n(t, x_2,y_2) \bar\rho[n(t,\cdot,\cdot)](t,x_1,y_1,x_2,y_2) \,dx_1\,dy_1\,dx_2\,dy_2\\
    &\quad\sim \frac 12 \iint_{\mathbb R^2}\iint_{\mathbb R^2}  (y_1x_{1}+y_2x_1)\bar\rho[\Gamma_{2\sigma_x^2}(\cdot-X(t))\Gamma_{2\sigma_y^2}(\cdot-Y(t))](t,x_1,y_1,x_2,y_2)\\
    &\phantom{\quad\sim \frac 12 \iint_{\mathbb R^2}\iint_{\mathbb R^2}}\Gamma_{2\sigma^2}(x_1-X(t))\Gamma_{2\sigma^2}(y_1-Y(t)) \Gamma_{2\sigma^2}(x_2-X(t))\Gamma_{2\sigma^2}(y_2-Y(t)) \,dx_1\,dy_1\,dx_2\,dy_2\\
    &\quad =  X(t)Y(t)\bar\rho+ 2\sigma_x^2Y(t)\partial_{x_1}\bar\rho+ \sigma_y^2 X(t)\partial_{y_1}\bar\rho+\sigma_y^2 X(t)\partial_{y_2}\bar\rho\\
    &\qquad +\sigma_x^2X(t)Y(t)\partial_{x_1x_1}\bar\rho+ \sigma_x^2X(t)Y(t)\partial_{x_2x_2}\bar\rho+ \sigma_y^2X(t)Y(t)\partial_{y_1y1}\bar\rho+ \sigma_y^2X(t)Y(t)\partial_{y_2y_2}\bar\rho+\mathcal O\left(\|(\sigma_x,\sigma_y)\|^3\right),
\end{align*}
since
\begin{align*}
    &\iint_{\mathbb R^2}\iint_{\mathbb R^2}  y_1x_1\bar\rho[\Gamma_{2\sigma_x^2}(\cdot-X(t))\Gamma_{2\sigma_y^2}(\cdot-Y(t))](t,x_1,y_1,x_2,y_2)\\
    &\phantom{\iint_{\mathbb R^2}\iint_{\mathbb R^2}  }\Gamma_{2\sigma^2}(x_1-X(t))\Gamma_{2\sigma^2}(y_1-Y(t))\Gamma_{2\sigma^2}(x_2-X(t))\Gamma_{2\sigma^2}(y_2-Y(t)) \,dx_1\,dx_2\,dy_1\,dy_2\\
    &\quad =\iint_{\mathbb R^2}\iint_{\mathbb R^2}  x_1y_1\Big( \bar\rho+ (x_1-X(t))\partial_{x_1}\bar\rho+ (x_2-X(t))\partial_{x_2}\bar\rho+ (y_1-Y(t))\partial_{y_1}\bar\rho+ (y_2-Y(t))\partial_{y_2}\bar\rho\\
    &\qquad +\frac{(x_1-X(t))^2}2\partial_{x_1x_1}\bar\rho+ \frac{(x_2-X(t))^2}2\partial_{x_2x_2}\bar\rho+ \frac{(y_1-Y(t))^2}2\partial_{y_1y1}\bar\rho+ \frac{(y_2-Y(t))^2}2\partial_{y_2y_2}\bar\rho\\
    &\qquad+(x_1-X(t))(x_2-X(t))\partial_{x_1x_2}\bar\rho+(y_1-Y(t))(y_2-Y(t))\partial_{y_1y_2}\bar\rho\\
    &\qquad +(x_1-X(t))(y_1-Y(t))\partial_{x_1y_1}\bar\rho+(x_2-X(t))(y_2-Y(t))\partial_{x_2y_2}\bar\rho\\
    &\qquad +(x_1-X(t))(y_2-Y(t))\partial_{x_1y_2}\bar\rho+(x_2-X(t))(y_1-Y(t))\partial_{x_2y_1}\bar\rho+\mathcal O(\|(x-X(t),y-Y(t))\|^3) \Big)\\
    &\qquad \Gamma_{2\sigma^2}(x_1-X(t))\Gamma_{2\sigma^2}(y_1-Y(t))\Gamma_{2\sigma^2}(x_2-X(t))\Gamma_{2\sigma^2}(y_2-Y(t))\,dx_1\,dy_1\,dx_2\,dy_2\\
    &\quad = X(t)Y(t)\bar\rho+ 2\sigma_x^2Y(t)\partial_{x_1}\bar\rho+ 2\sigma_y^2 X(t)\partial_{y_1}\bar\rho+\sigma_x^2X(t)Y(t)\partial_{x_1x_1}\bar\rho\\
    &\qquad + \sigma_x^2X(t)Y(t)\partial_{x_2x_2}\bar\rho+ \sigma_y^2X(t)Y(t)\partial_{y_1y1}\bar\rho+ \sigma_y^2X(t)Y(t)\partial_{y_2y_2}\bar\rho+\mathcal O\left(\|(\sigma_x,\sigma_y)\|^3\right),
\end{align*}
and the similarly when $y_2$ replaces $y_1$.

Bringing all these expressions together, we obtain
\begin{align*}
    &\frac d{dt}X(t)  = - \left(X(t)a(t,X(t),Y(t))+2\sigma_x^2\partial_x \bar a(t,X(t),Y(t))+\sigma_x^2 X(t)\partial_{xx}\bar a(t,X(t),Y(t))+\sigma_y^2 X(t)\partial_{yy}\bar a(t,X(t),Y(t))\right) \\
    &\qquad + \frac{\bar a(t,X(t),Y(t))+\sigma_x^2\partial_{xx}\bar a(t,X(t),Y(t))+\sigma_y^2\partial_{yy}\bar a(t,X(t),Y(t))}{\bar\rho X(t)+ 2\sigma_x^2\partial_{x_1}\bar\rho+\sigma_x^2 X(t)\partial_{x_1x_1}\bar\rho+ \sigma_x^2X(t)\partial_{x_2x_2}\bar\rho+ \sigma_y^2 X(t)\partial_{y_1y1}\bar\rho+ \sigma_y^2 X(t)\partial_{y_2y_2}\bar\rho} \\
    &\qquad \Big((X(t)^2+\sigma_x^2)\bar\rho+3\sigma_x^2X(t)\partial_{x_1}\bar\rho+ \sigma_x^2 X(t)\partial_{x_2}\bar\rho\\
    &\qquad +\sigma_x^2 X^2\partial_{x_1x_1}\bar\rho+ \sigma_x^2 X^2\partial_{x_2x_2}\bar\rho+ \sigma_y^2 X^2\partial_{y_1y1}\bar\rho+ \sigma_y^2 X^2\partial_{y_2y_2}\bar\rho\Big)+\mathcal O\left(\|(\sigma_x,\sigma_y)\|^3\right)\\
    &\quad =- X(t)\bar a(t,X(t),Y(t))-2\sigma_x^2\partial_x \bar a(t,X(t),Y(t))-\sigma_x^2 X(t)\partial_{xx}\bar a(t,X(t),Y(t))-\sigma_y^2 X(t)\partial_{yy}\bar a(t,X(t),Y(t)) \\
    &\qquad + \bar a(t,X(t),Y(t)) X(t)\frac{1+\sigma_x^2\frac{\partial_{xx}\bar a(t,X(t),Y(t))}{\bar a(t,X(t),Y(t))}+\sigma_y^2\frac{\partial_{yy}\bar a(t,X(t),Y(t))}{\bar a(t,X(t),Y(t))}}{1+ 2\sigma_x^2\frac{\partial_{x_1}\bar\rho}{\bar\rho X(t)}+\sigma_x^2 \frac{\partial_{x_1x_1}\bar\rho}{\bar\rho }+ \sigma_x^2\frac{\partial_{x_2x_2}\bar\rho}{\bar\rho }+ \sigma_y^2 \frac{\partial_{y_1y1}\bar\rho}{\bar\rho}+ \sigma_y^2 \frac{\partial_{y_2y_2}\bar\rho}{\bar\rho}} \\
    &\qquad \Big(1+\frac{\sigma_x^2}{X(t)^2}+3\sigma_x^2\frac{\partial_{x_1}\bar\rho}{X(t)\bar\rho}+ \sigma_x^2 \frac{\partial_{x_2}\bar\rho}{X(t)\bar\rho}\\
    &\qquad +\sigma_x^2 \frac{\partial_{x_1x_1}\bar\rho}{\bar\rho}+ \sigma_x^2 \frac{\partial_{x_2x_2}\bar\rho}{\bar\rho}+ \sigma_y^2 \frac{\partial_{y_1y1}\bar\rho}{\bar\rho}+ \sigma_y^2 \frac{\partial_{y_2y_2}\bar\rho}{\bar\rho}\Big)+\mathcal O\left(\|(\sigma_x,\sigma_y)\|^3\right)\\
    &\quad =- X(t)\bar a(t,X(t),Y(t))-2\sigma_x^2\partial_x \bar a(t,X(t),Y(t))-\sigma_x^2 X(t)\partial_{xx}\bar a(t,X(t),Y(t))-\sigma_y^2 X(t)\partial_{yy}\bar a(t,X(t),Y(t)) \\
    &\qquad + \bar a(t,X(t),Y(t)) X(t)\left(1+\sigma_x^2\frac{\partial_{xx}\bar a(t,X(t),Y(t))}{\bar a(t,X(t),Y(t))}+\sigma_y^2\frac{\partial_{yy}\bar a(t,X(t),Y(t))}{\bar a(t,X(t),Y(t))}\right)\\
    &\qquad \left(1- 2\sigma_x^2\frac{\partial_{x_1}\bar\rho}{\bar\rho X(t)}-\sigma_x^2 \frac{\partial_{x_1x_1}\bar\rho}{\bar\rho }- \sigma_x^2\frac{\partial_{x_2x_2}\bar\rho}{\bar\rho }- \sigma_y^2 \frac{\partial_{y_1y1}\bar\rho}{\bar\rho }- \sigma_y^2 \frac{\partial_{y_2y_2}\bar\rho}{\bar\rho }\right) \\
    &\qquad \Big(1+\frac{\sigma_x^2}{X(t)^2}+3\sigma_x^2\frac{\partial_{x_1}\bar\rho}{X(t)\bar\rho}+ \sigma_x^2 \frac{\partial_{x_2}\bar\rho}{X(t)\bar\rho}\\
    &\qquad +\sigma_x^2 \frac{\partial_{x_1x_1}\bar\rho}{\bar\rho}+ \sigma_x^2 \frac{\partial_{x_2x_2}\bar\rho}{\bar\rho}+ \sigma_y^2 X^2\frac{\partial_{y_1y1}\bar\rho}{\bar\rho}+ \sigma_y^2 X^2\frac{\partial_{y_2y_2}\bar\rho}{\bar\rho}\Big)+\mathcal O\left(\|(\sigma_x,\sigma_y)\|^3\right)\\
    &\quad =- X(t)\bar a(t,X(t),Y(t))-2\sigma_x^2\partial_x \bar a(t,X(t),Y(t))-\sigma_x^2 X(t)\partial_{xx}\bar a(t,X(t),Y(t))-\sigma_y^2 X(t)\partial_{yy}\bar a(t,X(t),Y(t)) \\
    &\qquad + \bar a(t,X(t),Y(t)) X(t)+\left(\sigma_x^2X(t)\partial_{xx}\bar a(t,X(t),Y(t))+\sigma_y^2X(t)\partial_{yy}\bar a(t,X(t),Y(t))\right)\\
    &\qquad +\Big(- 2\sigma_x^2\bar a(t,X(t),Y(t)) \frac{\partial_{x_1}\bar\rho}{\bar\rho }-\sigma_x^2 \bar a(t,X(t),Y(t)) X(t)\frac{\partial_{x_1x_1}\bar\rho}{\bar\rho }- \sigma_x^2\bar a(t,X(t),Y(t)) X(t)\frac{\partial_{x_2x_2}\bar\rho}{\bar\rho}\\
    &\qquad \phantom{+(sf}- \sigma_y^2 \bar a(t,X(t),Y(t)) X(t)\frac{\partial_{y_1y1}\bar\rho}{\bar\rho }- \sigma_y^2 \bar a(t,X(t),Y(t)) X(t)\frac{\partial_{y_2y_2}\bar\rho}{\bar\rho }\Big) \\
    &\qquad +\Big(\bar a(t,X(t),Y(t))\frac{\sigma_x^2}{X(t)}+3\sigma_x^2\bar a(t,X(t),Y(t)) \frac{\partial_{x_1}\bar\rho}{\bar\rho}+ \sigma_x^2 \bar a(t,X(t),Y(t))\frac{\partial_{x_2}\bar\rho}{\bar\rho}\\
    &\qquad \phantom{+(sf}+\sigma_x^2 \bar a(t,X(t),Y(t)) X(t)\frac{\partial_{x_1x_1}\bar\rho}{\bar\rho}+ \sigma_x^2\bar a(t,X(t),Y(t)) X(t) \frac{\partial_{x_2x_2}\bar\rho}{\bar\rho}\\
    &\qquad \phantom{+(sf}+ \sigma_y^2\bar a(t,X(t),Y(t)) X(t)\frac{\partial_{y_1y1}\bar\rho}{\bar\rho}+ \sigma_y^2\bar a(t,X(t),Y(t)) X(t)\frac{\partial_{y_2y_2}\bar\rho}{\bar\rho}\Big)+\mathcal O\left(\|(\sigma_x,\sigma_y)\|^3\right)\\
    &\quad =2\sigma_x^2\left[-\partial_x \bar a(t,X(t),Y(t))+\bar a(t,X(t),Y(t))\frac{\partial_{x_1}\bar\rho}{2\bar\rho} +\bar a(t,X(t),Y(t))\frac{1}{2X(t)}+\bar a(t,X(t),Y(t))\frac{\partial_{x_2}\bar\rho}{2\bar\rho}\right]+\mathcal O\left(\|(\sigma_x,\sigma_y)\|^3\right)\\
    &\quad =2\sigma_x^2\left[-\partial_x \bar a(t,X(t),Y(t))+\frac{\bar a(t,X(t),Y(t))}{2X(t)}+\bar a(t,X(t),Y(t))\frac{\partial_{x_1}\bar\rho}{\bar\rho}
    \right]+\mathcal O\left(\|(\sigma_x,\sigma_y)\|^3\right).
\end{align*}
where the last equality uses the symmetry of $\bar\rho$. Similarly for $Y(t)$,
\begin{align*}
    &\frac d{dt}Y(t)  = - \left(Y(t)\bar a(t,X(t),Y(t))+2\sigma_y^2\partial_y \bar a(t,X(t),Y(t))+\sigma_x^2 Y(t)\partial_{xx}\bar a(t,X(t),Y(t))+\sigma_y^2 Y(t)\partial_{yy}\bar a(t,X(t),Y(t))\right) \\
    &\qquad + \frac{\bar a(t,X(t),Y(t))+\sigma_x^2\partial_{xx}\bar a(t,X(t),Y(t))+\sigma_y^2\partial_{yy}\bar a(t,X(t),Y(t))}{\bar\rho X(t)+ 2\sigma_x^2\partial_{x_1}\bar\rho+\sigma_x^2 X(t)\partial_{x_1x_1}\bar\rho+ \sigma_x^2X(t)\partial_{x_2x_2}\bar\rho+ \sigma_y^2 X(t)\partial_{y_1y1}\bar\rho+ \sigma_y^2 X(t)\partial_{y_2y_2}\bar\rho} \\
    &\qquad \Big( X(t)Y(t)\bar\rho+ 2\sigma_x^2Y(t)\partial_{x_1}\bar\rho+ \sigma_y^2 X(t)\partial_{y_1}\bar\rho+\sigma_y^2 X(t)\partial_{y_2}\bar\rho\\
    &\qquad +\sigma_x^2X(t)Y(t)\partial_{x_1x_1}\bar\rho+ \sigma_x^2X(t)Y(t)\partial_{x_2x_2}\bar\rho+ \sigma_y^2X(t)Y(t)\partial_{y_1y1}\bar\rho+ \sigma_y^2X(t)Y(t)\partial_{y_2y_2}\bar\rho\Big)+\mathcal O\left(\|(\sigma_x,\sigma_y)\|^3\right)\\
    &\quad = - Y(t)\bar a(t,X(t),Y(t))-2\sigma_y^2\partial_y \bar a(t,X(t),Y(t))-\sigma_x^2 Y(t)\partial_{xx}\bar a(t,X(t),Y(t))-\sigma_y^2 Y(t)\partial_{yy}\bar a(t,X(t),Y(t))\\
    &\qquad + Y(t) \bar a(t,X(t),Y(t)) \left(1+\sigma_x^2\frac{\partial_{xx}\bar a(t,X(t),Y(t))}{\bar a(t,X(t),Y(t))}+\sigma_y^2\frac{\partial_{yy}\bar a(t,X(t),Y(t))}{\bar a(t,X(t),Y(t))}\right)\\
    &\qquad \left(1- 2\sigma_x^2\frac{\partial_{x_1}\bar\rho}{\bar\rho X(t)}-\sigma_x^2 \frac{\partial_{x_1x_1}\bar\rho}{\bar\rho}- \sigma_x^2\frac{\partial_{x_2x_2}\bar\rho}{\bar\rho}- \sigma_y^2 \frac{\partial_{y_1y1}\bar\rho}{\bar\rho}- \sigma_y^2 \frac{\partial_{y_2y_2}\bar\rho}{\bar\rho}\right) \\
    &\qquad \Big( 1+ 2\sigma_x^2\frac{\partial_{x_1}\bar\rho}{X(t)\bar\rho}+ \sigma_y^2 \frac{\partial_{y_1}\bar\rho}{Y(t)\bar\rho}+\sigma_y^2 \frac{\partial_{y_2}\bar\rho}{Y(t)\bar\rho}+\sigma_x^2\frac{\partial_{x_1x_1}\bar\rho}{\bar\rho}+ \sigma_x^2\frac{\partial_{x_2x_2}\bar\rho}{\bar\rho}+ \sigma_y^2\frac{\partial_{y_1y1}\bar\rho}{\bar\rho}+ \sigma_y^2\frac{\partial_{y_2y_2}\bar\rho}{\bar\rho}\Big)+\mathcal O\left(\|(\sigma_x,\sigma_y)\|^3\right)\\
    &\quad = - Y(t)\bar a(t,X(t),Y(t))-2\sigma_y^2\partial_y \bar a(t,X(t),Y(t))-\sigma_x^2 Y(t)\partial_{xx}\bar a(t,X(t),Y(t))-\sigma_y^2 Y(t)\partial_{yy}\bar a(t,X(t),Y(t))\\
    &\qquad + Y(t) \bar a(t,X(t),Y(t))+\sigma_x^2Y(t)\partial_{xx}\bar a(t,X(t),Y(t))+\sigma_y^2Y(t)\partial_{yy}\bar a(t,X(t),Y(t))\\
    &\qquad - 2\sigma_x^2Y(t) \bar a(t,X(t),Y(t))\frac{\partial_{x_1}\bar\rho}{\bar\rho X(t)}-\sigma_x^2 Y(t) \bar a(t,X(t),Y(t))\frac{\partial_{x_1x_1}\bar\rho}{\bar\rho}- \sigma_x^2Y(t) \bar a(t,X(t),Y(t))\frac{\partial_{x_2x_2}\bar\rho}{\bar\rho}\\
    &\qquad - \sigma_y^2 Y(t) \bar a(t,X(t),Y(t))\frac{\partial_{y_1y1}\bar\rho}{\bar\rho}- \sigma_y^2 Y(t) \bar a(t,X(t),Y(t))\frac{\partial_{y_2y_2}\bar\rho}{\bar\rho} \\
    &\qquad + 2\sigma_x^2Y(t) \bar a(t,X(t),Y(t))\frac{\partial_{x_1}\bar\rho}{X(t)\bar\rho}+ \sigma_y^2 \bar a(t,X(t),Y(t))\frac{\partial_{y_1}\bar\rho}{\bar\rho}+\sigma_y^2 \bar a(t,X(t),Y(t))\frac{\partial_{y_2}\bar\rho}{\bar\rho}\\
    &\qquad +\sigma_x^2Y(t) \bar a(t,X(t),Y(t))\frac{\partial_{x_1x_1}\bar\rho}{\bar\rho}+ \sigma_x^2Y(t) \bar a(t,X(t),Y(t))\frac{\partial_{x_2x_2}\bar\rho}{\bar\rho}+ \sigma_y^2Y(t) \bar a(t,X(t),Y(t))\frac{\partial_{y_1y1}\bar\rho}{\bar\rho}\\
    &\qquad + \sigma_y^2Y(t) \bar a(t,X(t),Y(t))\frac{\partial_{y_2y_2}\bar\rho}{\bar\rho}+\mathcal O\left(\|(\sigma_x,\sigma_y)\|^3\right)\\
    &\quad = 2\sigma_y^2\left[-\partial_y \bar a(t,X(t),Y(t)) + \bar a(t,X(t),Y(t))\frac{\partial_{y_1}\bar\rho}{2\bar\rho}+ \bar a(t,X(t),Y(t))\frac{\partial_{y_2}\bar\rho}{2\bar\rho}\right]+\mathcal O\left(\|(\sigma_x,\sigma_y)\|^3\right)\\
    &\quad = 2\sigma_y^2\left[-\partial_y \bar a(t,X(t),Y(t)) + \bar a(t,X(t),Y(t))\frac{\partial_{y_1}\bar\rho}{\bar\rho}\right]+\mathcal O\left(\|(\sigma_x,\sigma_y)\|^3\right).
\end{align*}
We therefore recover \eqref{eq:dXY}, noticing that $\frac{\partial_{x_1}\bar\rho}{\bar\rho}$ can be replaced in the equations above by $\frac{\partial_{x}\hat \rho[t,X(t),Y(t)](t,X(t),Y(t)}{\hat \rho[t,X(t),Y(t)](t,X(t),Y(t)}$, and similarly $\frac{\partial_{y_1}\bar\rho}{\bar\rho}$ can be replaced in the equations above by $\frac{\partial_{y}\hat \rho[t,X(t),Y(t)](t,X(t),Y(t)}{\hat\rho[t,X(t),Y(t)](t,X(t),Y(t)}$.

\end{appendices}

\printbibliography

\Addresses

\end{document}